\documentclass[11pt]{article}
\usepackage[hmargin={3.0cm,3.0cm},vmargin={3.0cm,3.0cm}]{geometry}

\usepackage{amsmath,amssymb}
\usepackage{amsfonts}
\usepackage{mathtools}
\usepackage{amsthm}
\usepackage{tikz}
\usetikzlibrary{arrows.meta}
\usetikzlibrary{decorations.markings}

\usepackage{bm}

\usepackage[skip=2pt,font=small]{caption}

\usepackage{hyperref}
\hypersetup{
colorlinks,
linkcolor={blue!60!black},
citecolor={blue!90!black},
urlcolor={blue!90!black}}

\numberwithin{equation}{section}

\usepackage{cite}

\usepackage[titletoc,title]{appendix}

\newtheorem{theorem}{Theorem}[section]

\newtheorem{prop}[theorem]{Proposition}
\theoremstyle{remark}
\newtheorem{example}{Example}

\tikzset{->-/.style={decoration={
  markings,
  mark=at position .74 with {\arrow{Latex}}},postaction={decorate}}}

\newcommand\Square[1]{+(-#1,-#1) rectangle +(#1,#1)}

\renewcommand{\author}[1]{\large\rm #1\\ \bigskip}
\newcommand{\address}[1]{{\normalsize\it #1\\}\bigskip}
\renewcommand{\title}[1]{\bigskip\bigskip\Large\bf #1\bigskip\bigskip\\}

\DeclarePairedDelimiter{\abs}{|}{|}
\newcommand{\R}{\mathbb{R}}
\newcommand{\N}{\mathbb{N}}
\newcommand{\Z}{\mathbb{Z}}
\newcommand{\Cn}{\mathbb{C}}
\newcommand{\K}{\mathbb{K}}
\DeclareMathOperator{\ud}{d}
\DeclareMathOperator{\sgn}{sgn}

\newcommand{\x}{{\boldsymbol{x}}}
\newcommand{\y}{{\boldsymbol{y}}}
\newcommand{\al}{{\bm{\alpha}}}
\newcommand{\bt}{{\bm{\beta}}}

\renewcommand{\imath}{{\mathrm{i}}}

\newcommand{\Pj}{\mathbb{CP}}
\newcommand{\PTN}{\mathcal{X}_1}
\newcommand{\PNpPNm}{\mathcal{X}_2}

\newcommand{\ccy}{y_{m,n}}
\newcommand{\ccya}{y_{m-1,n+1}}
\newcommand{\ccyb}{y_{m+1,n+1}}
\newcommand{\ccyc}{y_{m-1,n-1}}
\newcommand{\ccyd}{y_{m+1,n-1}}

\newcommand{\ccx}{x_{m,n}}
\newcommand{\cca}{x_{m-1,n+1}}
\newcommand{\ccb}{x_{m+1,n+1}}
\newcommand{\ccc}{x_{m-1,n-1}}
\newcommand{\ccd}{x_{m+1,n-1}}

\newcommand{\ccpc}{\alpha}
\newcommand{\ccpa}{\alpha_1}
\newcommand{\ccpb}{\alpha_2}
\newcommand{\ccqc}{\beta}
\newcommand{\ccqa}{\beta_1}
\newcommand{\ccqb}{\beta_2}

\newcommand{\ccpp}{{\al}}
\newcommand{\ccqq}{{\bt}}

\newcommand{\A}[7]{A(#1;#2,#3,#4,#5;#6,#7)}
\newcommand{\Bf}[7]{\overline{B}(#1;#2,#3,#4,#5;#6,#7)}
\newcommand{\B}[7]{B(#1;#2,#3,#4,#5;#6,#7)}
\newcommand{\Cf}[7]{\overline{C}(#1;#2,#3,#4,#5;#6,#7)}
\newcommand{\C}[7]{C(#1;#2,#3,#4,#5;#6,#7)}

\newcounter{app}
\newcounter{sapp}[app]

\begin{document}

\vglue 2cm

\begin{center}

\title{Algebraic entropy for face-centered quad equations}
\author{Giorgio Gubbiotti$^{1,2}$ and Andrew P.~Kels$^2$}
\address{
    $^1$School of Mathematics and Statistics F07, The University of Sydney, NSW 2006, Australia,
    \\
    \emph{email: \texttt{giorgio.gubbiotti@sydney.edu.au}}
    \\[0.1cm]
    $^2$Scuola Internazionale Superiore di Studi Avanzati, Via Bonomea 265, 34136 Trieste, Italy,
    \\
    \emph{email: \texttt{giorgio.gubbiotti@sissa.it}}
    \\
    \emph{email: \texttt{akels@sissa.it}}
}

\end{center}

\begin{abstract}
    In this paper we define the algebraic entropy test for face-centered 
    quad equations, which are equations defined on vertices of a quadrilateral plus an additional interior vertex.
    This notion of algebraic entropy is applied to a recently introduced
    class of these equations that satisfy a new form of multidimensional consistency called consistency-around-a-face-centered-cube (CAFCC), whereby the system of equations is consistent on a face-centered cubic unit cell.  It is found that for certain arrangements of equations (or pairs of equations) in the square lattice, all known CAFCC equations pass the algebraic entropy test possessing either quadratic or linear growth.
\end{abstract}

\section{Introduction}

Recently the second author has introduced a new set of five-point
lattice equations, called face-centered quad equations,
which satisfy a new form of multidimensional consistency,
called consistency-around-a-face-centered-cube (CAFCC)
\cite{Kels2020cafcc}.  The latter property may be regarded
as an analogue of the consistency-around-a-cube (CAC) property
\cite{DoliwaSantini1997,Nijhoff2001} that is satisfied by regular
integrable quad equations \cite{BobenkoSuris2002}.
For such quad equations there is also the concept of algebraic entropy
\cite{BellonViallet1999}, which provides an algorithmic tool for
determining the integrability of the equations based on examining the
growth of their degrees under lattice evolution.
The main goals of this paper are to introduce an analogous concept of
algebraic entropy that is applicable to face-centered quad equations,
and to apply this concept to equations that satisfy CAFCC to quantify
the nature of their integrability (or non-integrability) in terms of
their degree growth.

The concept of algebraic entropy is a method to measure the
{complexity} of birational maps, analogous to the one introduced
by Arnold for diffeomorphisms \cite{Arnold1990}.
For finite dimensional maps this concept was first introduced in
\cite{Veselov1992,FalquiViallet1993,HietarintaViallet1997} and finally
formalised in \cite{BellonViallet1999}.
The concept of algebraic entropy has been generalised to various
different kinds of infinite-dimensional systems,
most notably to quad-equations in \cite{Tremblay2001,Viallet2006,Viallet2009}.
A small generalisation of the concept of algebraic entropy for quad equations,
considering the possibility of multiple growths coexisting in the same
quad equation, was then given in \cite{GSL_general}. 

Integrability for both continuous and discrete systems is associated with 
regularity and predictability of the solutions of a given system.
Intuitively, this means that the complexity of an integrable system
is expected to be {low}.
For birational maps defined over complex projective spaces, this
condition translates to the vanishing of the algebraic entropy.
Positive algebraic entropy implies that the complexity of the 
iterates of a birational map grow exponentially.
This makes the evolution of the motion almost impossible to predict and 
very sensitive to the initial conditions: 
an indication of what is usually referred as {chaos}.
For a complete account on the history and the definition of algebraic
entropy in the mentioned cases we refer to the reviews
\cite{GrammaticosHalburdRamaniViallet2009,GubbiottiASIDE16,HietarintaBook}.

Face-centered quad equations are generalisations of quad equations.
While quad equations are equations defined on the four vertices of a
quadrilateral polygon, face-centered quad equations also involve an
additional interior vertex.
Thus a face-centered quad equation can be expressed as a sum of powers
of the variable at the interior vertex with coefficients given by regular
quad equations.
In \cite{Kels2020cafcc} the second author introduced CAFCC as a form of
multidimensional consistency for face-centered quad equations.
The latter property involves satisfying an over-determined system
of fourteen equations on the face-centered cube (or face-centered cubic
unit cell) for eight unknown variables, and may be regarded as an analogue
of the CAC property defined for regular quad equations.
The algorithmic nature of CAC has been exploited to obtain several
different classifications of quad equations
\cite{ABS2003,ABS2009,Boll2011,Hietarinta2005,Hietarinta2018}.
The quad equations that satisfy CAC have also been shown
to arise in the asymptotics of more general equations for
hypergeometric functions associated to the star-triangle relations
\cite{Bazhanov:2007mh,Bazhanov:2010kz,Bazhanov:2016ajm,Kels:2018xge},
and this type of connection was used by the second author to derive the
set of CAFCC equations.
Similarly to the case of integrable quad equations, it was subsequently
shown that the property of CAFCC allows for an algorithmic derivation
of Lax pairs for face-centered quad equations \cite{Kels2020lax}.

The goals of this paper are to define and apply the concept of algebraic
entropy to face-centered quad equations.
Since the CAFCC equations satisfy a form of multidimensional consistency
and have associated Lax pairs, it is expected that they will have a
polynomial degree growth.
Indeed, for specific arrangements of equations in the lattice it is
found that
the type-A and type-C equations each have quadratic growths of degrees,
and the type-B equations each have linear growths of degrees.
The formulation of algebraic entropy for face-centered quad equations and
the computation of the generating functions describing the polynomial
degree growths for equations satisfying CAFCC are the main results
of this paper.

In the process of applying algebraic entropy we also solve the non-trivial
problem of finding integrable combinations of
CAFCC equations in a $\Z^{2}$ lattice.
This problem is non-trivial because
an arbitrarily chosen arrangement of equations is typically not sufficient
to obtain a sub-exponential
growth of degrees.
The reason for this is that type-B and -C CAFCC equations do not satisfy
the full
symmetries of the square, so a general arrangement of these equations
will not give rise to a
well-defined lattice of consistent face-centered cubes in three
dimensions.
This situation is similar to the case of the $H4$- and $H6$-type quad
equations \cite{GSL_general}.
Our constructions of lattices for type-B and -C CAFCC equations may be
regarded as
the analogue of the black-and-white lattices for the $H4$ equations
\cite{ABS2009,Xenitidis2009} and the four-colours lattice for the $H6$
equations \cite{Boll2011,GSL_general}.
This means that such type-B and type-C equations are \emph{non-autonomous
face-centered quad equations} alternating
on the lattice between different CAFCC equations (see \cite{GSL_general}
for the analogous case on quad equations).

This paper is organised as follows.
In Section \ref{sec:fcqeaae} we give an overview of the CAFCC
equations, and then define the concept of algebraic
entropy for face-centered quad equations.
We discuss the main properties of the algebraic entropy,
and we present a few different examples of growth of degrees.
In Section \ref{sec:aecafcc} we give the arrangements
of type-A, type-B, and type-C CAFCC equations in the square lattice,
and give the results from the computations of the
algebraic entropy for the corresponding systems of equations.
Finally, in Section \ref{sec:concl} we give some conclusion and outlook
for future works.

\section{Face-centered quad equations and algebraic entropy}
\label{sec:fcqeaae}

In Section \ref{sec:Mcomp} an overview of the concept of the face-centered
quad equations will be given.
Then in Section \ref{sec:aefcqe} a definition of algebraic entropy
for the face-centered quad equations will be given.
The concept of algebraic entropy requires the definition of a
birational evolution, which in turn leads to the problem of defining
proper initial conditions.
These topics will be treated in Section \ref{sec:aefcqe}.

\subsection{Face-centered quad equations}
\label{sec:Mcomp}

The face-centered quad equations are defined on a five-point stencil of
the square lattice shown in Figure \ref{fig-face}.  Figure \ref{fig-face}
may be regarded as a face of the face-centered cube, which was the central
idea in the formulation of the property of multidimensional consistency
for these equations \cite{Kels2020cafcc}. In Figure \ref{fig-face}
the four variables $\cca,\ccb,\ccc,\ccd$,
 are associated to the four corner vertices, and the variable $\ccx$
 is associated to the face vertex.  The two-component parameters
\begin{align}
\label{pardefs}
\ccpp=(\ccpa,\ccpb),\qquad\ccqq=(\ccqa,\ccqb),
\end{align}
are associated to the edges of the face, where two opposite edges that
are parallel are assigned the same components of the parameters.

 \begin{figure}[tbh]
\centering
\begin{tikzpicture}[scale=0.8]

\draw[-,gray,very thin,dashed] (5,-1)--(5,3)--(1,3)--(1,-1)--(5,-1);
\draw[-] (5,3)--(1,-1);
\draw[-] (5,-1)--(1,3);
\fill (0.8,-0.0) circle (0.1pt)
node[left=0.5pt]{\color{black}\small $\ccpa$};
\fill (5.2,-0.0) circle (0.1pt)
node[right=0.5pt]{\color{black}\small $\ccpa$};
\fill (4,-1.2) circle (0.1pt)
node[below=0.5pt]{\color{black}\small $\ccqb$};
\fill (4,3.2) circle (0.1pt)
node[above=0.5pt]{\color{black}\small $\ccqb$};
\fill (0.8,2.0) circle (0.1pt)
node[left=0.5pt]{\color{black}\small $\ccpb$};
\fill (5.2,2.0) circle (0.1pt)
node[right=0.5pt]{\color{black}\small $\ccpb$};
\fill (2,-1.2) circle (0.1pt)
node[below=0.5pt]{\color{black}\small $\ccqa$};
\fill (2,3.2) circle (0.1pt)
node[above=0.5pt]{\color{black}\small $\ccqa$};
\fill (3,1) circle (3.5pt)
node[below=7.5pt]{\color{black} $\ccx$};
\fill (1,-1) circle (3.5pt)
node[left=1.5pt]{\color{black} $\ccc$};
\fill (1,3) circle (3.5pt)
node[left=1.5pt]{\color{black} $\cca$};
\fill (5,3) circle (3.5pt)
node[right=1.5pt]{\color{black} $\ccb$};
\fill (5,-1) circle (3.5pt)
node[right=1.5pt]{\color{black} $\ccd$};

\draw[-,dotted,thick] (2,-1.2)--(2,3.2);\draw[-,dotted,thick]
(4,-1.2)--(4,3.2);
\draw[-,dotted,thick] (0.8,-0.0)--(5.2,-0.0);\draw[-,dotted,thick]
(0.8,2.0)--(5.2,2.0);

\end{tikzpicture}

\caption{Variables and parameters associated to the vertices and
edges of a face of the face-centered cube.  In terms of the graphical
representation of CAFCC this is a type-A equation.}
\label{fig-face}
\end{figure}
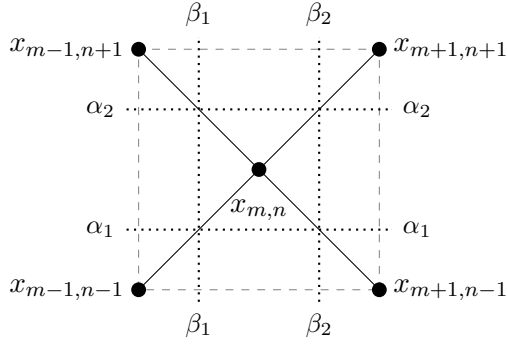

The face-centered quad equations that we are interested in, are always
linear in the four corner variables $\cca,\ccb,\ccc,\ccd$,
but not the face variable $\ccx$
(typically they are quadratic in $\ccx$,
except for the more complicated elliptic case).  Such face-centered quad
equations can be written in a general polynomial form
\begin{equation}
    \begin{split}
    A=\kappa_1x_ax_bx_cx_d+\kappa_2x_ax_bx_c+\kappa_3x_ax_bx_d+\kappa_4x_ax_cx_d+\kappa_5x_bx_cx_d
    \phantom{,}
    \\
    +\kappa_6x_ax_b+\kappa_7x_ax_c+\kappa_8x_ax_d+\kappa_9x_bx_c+\kappa_{10}x_bx_d+\kappa_{11}x_cx_d
    \phantom{,}
    \\
    +\kappa_{12}x_a+\kappa_{13}x_b+\kappa_{14}x_c+\kappa_{15}x_d+\kappa_{16}=0,
    \end{split}
    \label{afflin}
\end{equation}
where $x_a=\cca,x_b=\ccb,x_c=\ccc,x_d=\ccd$, and the coefficients
$\kappa_i(x_{m,n};\ccpp,\ccqq)$ ($i=1,\ldots,16$) depend on the face
variable $\ccx$ and the
four components $\ccpa,\ccpb,\ccqa,\ccqb$ of the parameters $\ccpp,\ccqq$.
The multi-linear expression \eqref{afflin} resembles the expression for
regular quad equations, except in the latter case the coefficients should
only depend on two parameters and no variables.

The face-centered quad equations that were found to satisfy CAFCC
\cite{Kels2020cafcc} are grouped into three different types.
Type-A equations $\A{\ccx}{\cca}{\ccb}{\ccc}{\ccd}{\ccpp}{\ccqq}$ are
invariant under the following exchanges of variables and parameters
\begin{subequations}
    \begin{gather}
        \label{refsym}
        \{\cca\leftrightarrow\ccb, \ccc\leftrightarrow\ccd,
        \beta_1\leftrightarrow\beta_2\},
        \\
        \label{symAB}
        \{\cca\leftrightarrow\ccc,\ccb\leftrightarrow\ccd,\alpha_1\leftrightarrow\alpha_2\},
        \\
        \label{symA}
        \{\cca\leftrightarrow\ccd,\al\leftrightarrow\bt\}.
    \end{gather}
\end{subequations}
Type-B equations are only invariant under \eqref{refsym} and
\eqref{symAB}, while type-C equations are only invariant under
\eqref{refsym}.  It is also useful to distinguish these three different
types of equations by using three different graphical representations,
where the type-A equation is shown in Figure \ref{fig-face} and the
type-B and type-C equations are shown in Figure \ref{fig:3fig4quad}.
Type-A equations are drawn with single-line edges and type-B equations
are drawn with double-line edges.  Type-C equation are drawn with both
single- and double-line edges, and the reason for this is that they arise
on the face-centered cube as an intermediate equation that connects a
type-A equation and a type-B equation.

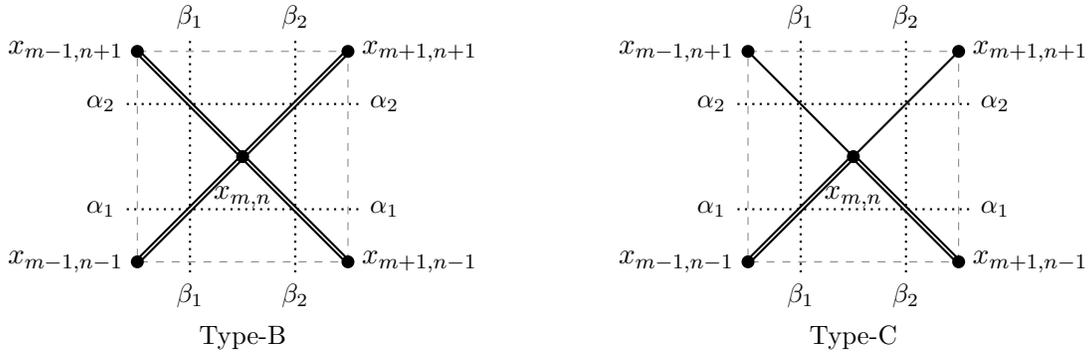
\begin{figure}[tbh]
\centering
\begin{tikzpicture}[scale=0.7]

\draw[-,gray,very thin,dashed] (5,-1)--(5,3)--(1,3)--(1,-1)--(5,-1);
\draw[-,double,thick] (5,3)--(1,-1);
\draw[-,double,thick] (5,-1)--(1,3);
\fill (0.8,-0.0) circle (0.1pt)
node[left=0.5pt]{\color{black}\small $\ccpa$};
\fill (5.2,-0.0) circle (0.1pt)
node[right=0.5pt]{\color{black}\small $\ccpa$};
\fill (4,-1.2) circle (0.1pt)
node[below=0.5pt]{\color{black}\small $\ccqb$};
\fill (4,3.2) circle (0.1pt)
node[above=0.5pt]{\color{black}\small $\ccqb$};
\fill (0.8,2.0) circle (0.1pt)
node[left=0.5pt]{\color{black}\small $\ccpb$};
\fill (5.2,2.0) circle (0.1pt)
node[right=0.5pt]{\color{black}\small $\ccpb$};
\fill (2,-1.2) circle (0.1pt)
node[below=0.5pt]{\color{black}\small $\ccqa$};
\fill (2,3.2) circle (0.1pt)
node[above=0.5pt]{\color{black}\small $\ccqa$};
\fill (3,1) circle (3.5pt)
node[below=7.5pt]{\color{black} $\ccx$};
\fill (1,-1) circle (3.5pt)
node[left=1.5pt]{\color{black} $\ccc$};
\fill (1,3) circle (3.5pt)
node[left=1.5pt]{\color{black} $\cca$};
\fill (5,3) circle (3.5pt)
node[right=1.5pt]{\color{black} $\ccb$};
\fill (5,-1) circle (3.5pt)
node[right=1.5pt]{\color{black} $\ccd$};

\draw[-,dotted,thick] (2,-1.2)--(2,3.2);\draw[-,dotted,thick] (4,-1.2)--(4,3.2);
\draw[-,dotted,thick] (0.8,-0.0)--(5.2,-0.0);\draw[-,dotted,thick] (0.8,2.0)--(5.2,2.0);

\fill (3,-2) circle (0.01pt)
node[below=0.5pt]{\color{black}\small Type-B};

\begin{scope}[xshift=330pt]

\draw[-,gray,very thin,dashed] (5,-1)--(5,3)--(1,3)--(1,-1)--(5,-1);
\draw[-,thick] (5,3)--(3,1)--(1,3);
\draw[-,double,thick] (1,-1)--(3,1)--(5,-1);
\fill (0.8,-0.0) circle (0.1pt)
node[left=0.5pt]{\color{black}\small $\ccpa$};
\fill (5.2,-0.0) circle (0.1pt)
node[right=0.5pt]{\color{black}\small $\ccpa$};
\fill (4,-1.2) circle (0.1pt)
node[below=0.5pt]{\color{black}\small $\ccqb$};
\fill (4,3.2) circle (0.1pt)
node[above=0.5pt]{\color{black}\small $\ccqb$};
\fill (0.8,2.0) circle (0.1pt)
node[left=0.5pt]{\color{black}\small $\ccpb$};
\fill (5.2,2.0) circle (0.1pt)
node[right=0.5pt]{\color{black}\small $\ccpb$};
\fill (2,-1.2) circle (0.1pt)
node[below=0.5pt]{\color{black}\small $\ccqa$};
\fill (2,3.2) circle (0.1pt)
node[above=0.5pt]{\color{black}\small $\ccqa$};
\fill (3,1) circle (3.5pt)
node[below=7.5pt]{\color{black} $\ccx$};
\fill (1,-1) circle (3.5pt)
node[left=1.5pt]{\color{black} $\ccc$};
\fill (1,3) circle (3.5pt)
node[left=1.5pt]{\color{black} $\cca$};
\fill (5,3) circle (3.5pt)
node[right=1.5pt]{\color{black} $\ccb$};
\fill (5,-1) circle (3.5pt)
node[right=1.5pt]{\color{black} $\ccd$};

\draw[-,dotted,thick] (2,-1.2)--(2,3.2);\draw[-,dotted,thick] (4,-1.2)--(4,3.2);
\draw[-,dotted,thick] (0.8,-0.0)--(5.2,-0.0);\draw[-,dotted,thick] (0.8,2.0)--(5.2,2.0);

\fill (3,-2) circle (0.01pt)
node[below=0.5pt]{\color{black}\small Type-C};

\end{scope}

\end{tikzpicture}
\caption{Graphical representations of type-B (left) and type-C (right) face-centered quad equations.}
\label{fig:3fig4quad}
\end{figure}

The three types of face-centered quad equations associated to Figures
\ref{fig-face} and \ref{fig:3fig4quad} can typically also be written in
the equivalent forms
\begin{align}
\label{4leg}
&\textrm{Type-A:}\;&\frac{
a(\ccx;\cca;\ccpb,\ccqa)a(\ccx;\ccd;\ccpa,\ccqb)}{a(\ccx;\ccb;\ccpb,\ccqb)a(\ccx;\ccc;\ccpa,\ccqa)}=1,
\\[0.1cm]
%
%
%
%
\label{4legb}
&\textrm{Type-B:}\;&\frac{b(\ccx;\cca;\ccpb,\ccqa)b(\ccx;\ccd;\ccpa,\ccqb)}{b(\ccx;\ccb;\ccpb,\ccqb)b(\ccx;\ccc;\ccpa,\ccqa)}=1,
\\[0.1cm]
%
\label{4legc}
&\textrm{Type-C:}\;&\frac{a(\ccx;\cca;\ccpb,\ccqa)c(\ccx;\ccd;\ccpa,\ccqb)}{a(\ccx;\ccb;\ccpb,\ccqb)c(\ccx;\ccc;\ccpa,\ccqa)}=1.
\end{align}
in terms of three types of ``leg'' functions $a(x;y;\ccpc,\ccqc)$,
$b(x;y;\ccpc,\ccqc)$, and $c(x;y;\ccpc,\ccqc)$.
The ``leg'' function $a(x_i;x_j;\ccpc,\ccqc)$ is associated to
the solid edges of both type-A and type-C equations, and the leg
functions $b(x_i;x_j;\ccpc,\ccqc)$ and $c(x_i;x_j;\ccpc,\ccqc)$
are associated to the solid double edges of type-B and type-C
equations respectively.  The leg function $a(x_i;x_j;\ccpc,\ccqc)$
satisfies a symmetry $a(x_i;x_j;\ccpc,\ccqc)a(x_i;x_j;\ccqc,\ccpc)=1$.
The leg functions $a(x_i;x_j;\ccpc,\ccqc)$, $b(x_i;x_j;\ccpc,\ccqc)$,
and $c(x_i;x_j;\ccpc,\ccqc)$, are each rational linear functions of a
corner variable $y$, and depend also on the face variable $x$ and one
component from each of the parameters $\ccpp$, $\ccqq$.  The expressions \eqref{4leg}--\eqref{4legc} may be regarded as analogues of the three-leg forms for the regular quad equations \cite{ABS2003}.  For the cases of type-A equations, the expressions of the form \eqref{4leg} turn out to be equivalent to expressions for discrete Laplace-type (or Toda-type) equations associated to type-Q ABS equations \cite{BobenkoSuris2002}.

Explicit expressions for the different equations and combinations
that have been found to satisfy CAFCC are given in Appendix
\ref{app:equations}.  The property of CAFCC itself will not be considered
here as it is not essential for this paper.

\subsection{Algebraic entropy for face-centered quad equations}
\label{sec:aefcqe}

In this section we introduce the concept of algebraic entropy
for face-centered quad equations.
This may be regarded as an extension of the definition of algebraic
entropy for quad equations 
\cite{Tremblay2001,Viallet2006,Viallet2009,GSL_general,GubbiottiASIDE16}.
In what follows we will define the forms of the standard initial
conditions and describe the essential tools that are involved in the entropy computations.

\subsubsection{Double staircases of initial conditions}

Like quad equations, face-centered quad equations are multilinear 
in the extremal variables $x_{m\pm1,n\pm1}$.
This implies that face-centered quad equations satisfy a fundamental 
condition required to apply the algebraic entropy criterion: in every 
direction of evolution the associated map is rational and invertible with
a rational inverse, that is they are birational.
There are four possible directions of evolution in the square lattice,
corresponding to solving the equations with respect to one of the four different
corner variables.

In general, initial conditions can be given along straight half-lines
in the four directions in the square lattice.
To be more precise, in terms of the following half-lines
\begin{subequations}
    \begin{align}
        \mathcal{M}_{m_{0},n_{0}}^{\left( + \right)} &=
        \bigl\{ \left( m,n_{0} \right)\in\Z^{2}\, |\, m_{0}\leq m < +\infty  \bigr\},
        \\
        \mathcal{M}_{m_{0},n_{0}}^{\left( - \right)} &=
        \bigl\{ \left( m,n_{0} \right)\in\Z^{2} \, |\, -\infty\leq m < m_{0}  \bigr\},
        \\
        \mathcal{N}_{m_{0},n_{0}}^{\left( + \right)} &=
        \bigl\{ \left( m_{0},n \right)\in\Z^{2} \, |\, n_{0}\leq n < +\infty  \bigr\},
        \\
        \mathcal{N}_{m_{0},n_{0}}^{\left( - \right)} &=
        \bigl\{ \left( m_{0},n \right)\in\Z^{2} \, |\, -\infty\leq n < n_{0}  \bigr\},
    \end{align}
    \label{eq:halflines}%
\end{subequations}
the standard initial conditions for face-centered quad equations are defined on the following sets
of indices:
\begin{align}
    \begin{gathered}
        \mathcal{I}^{\left(+,+\right)}_{m_{0},n_{0}}
        =
        \mathcal{M}_{m_{0},n_{0}}^{\left( + \right)} 
        \cup
        \mathcal{M}_{m_{0}+1,n_{0}+1}^{\left( + \right)} 
        \cup
        \mathcal{N}_{m_{0},n_{0}}^{\left( + \right)} 
        \cup
        \mathcal{N}_{m_{0}+1,n_{0}+1}^{\left( + \right)}, 
        \\
        \mathcal{I}^{\left(-,+\right)}_{m_{0},n_{0}}
        =
        \mathcal{M}_{m_{0},n_{0}}^{\left( - \right)} 
        \cup
        \mathcal{M}_{m_{0}+1,n_{0}-1}^{\left( - \right)} 
        \cup
        \mathcal{N}_{m_{0},n_{0}}^{\left( + \right)} 
        \cup
        \mathcal{N}_{m_{0}+1,n_{0}-1}^{\left( + \right)},
        \\
        \mathcal{I}^{\left(+,-\right)}_{m_{0},n_{0}}
        =
        \mathcal{M}_{m_{0},n_{0}}^{\left( + \right)} 
        \cup
        \mathcal{M}_{m_{0}-1,n_{0}+1}^{\left( + \right)} 
        \cup
        \mathcal{N}_{m_{0},n_{0}}^{\left( - \right)} 
        \cup
        \mathcal{N}_{m_{0}-1,n_{0}+1}^{\left( - \right)},
        \\
        \mathcal{I}^{\left(-,-\right)}_{m_{0},n_{0}}
        =
        \mathcal{M}_{m_{0},n_{0}}^{\left( - \right)} 
        \cup
        \mathcal{M}_{m_{0}+1,n_{0}+1}^{\left( - \right)} 
        \cup
        \mathcal{N}_{m_{0},n_{0}}^{\left( - \right)} 
        \cup
        \mathcal{N}_{m_{0}+1,n_{0}+1}^{\left( - \right)}.
    \end{gathered}
    \label{eq:inihalflines}
\end{align}

However, to compute the algebraic entropy of face-centered 
quad equations it is more convenient to give initial conditions 
on \emph{double staircase} configurations. 
We discuss later the main rationale behind this choice.
An evolution from any staircase-like arrangement of initial values is possible 
in the lattices of face-centered quad equations.
Some prototypical examples of double staircases of points in the lattice are shown in Figure 
\ref{fig:doublestaircases}.
In general, a double staircase configuration might include hook-like configurations as in example 
$(d)$ of Figure \ref{fig:doublestaircases}.
However, a hook-like configuration requires compatibility conditions
on the initial data since there will be more than one way to calculate 
the same value for the dependent variable.
For this reason we exclude these kinds of configurations from our consideration.

\begin{figure}[htb!]
   \centering
   \begin{tikzpicture}[scale=0.5]
        \draw[style=help lines,dashed] (0,0) grid[step=1cm] (30,30);
        \draw[ultra thick] 
            (8,0) node[circle,fill,inner sep=2pt](){}--
            (6,0) node[circle,fill,inner sep=2pt](){}--
            (6,2) node[circle,fill,inner sep=2pt](){}--
            (6,4) node[circle,fill,inner sep=2pt](){}--
            (4,4) node[circle,fill,inner sep=2pt](){}--
            (4,6) node[circle,fill,inner sep=2pt](){}--
            (4,8) node[circle,fill,inner sep=2pt](){}--
            (2,8) node[circle,fill,inner sep=2pt](){}--
            (2,10) node[circle,fill,inner sep=2pt](){}--
            (2,12) node[circle,fill,inner sep=2pt](){}--
            (0,12) node[circle,fill,inner sep=2pt](){};
        \draw[very thick] (7,0)--
            (7,1) node[circle,fill,color=gray,inner sep=2pt](){}--
            (7,3) node[circle,fill,color=gray,inner sep=2pt](){}--
            (5,3) node[circle,fill,color=gray,inner sep=2pt](){}--
            (5,5) node[circle,fill,color=gray,inner sep=2pt](){}--
            (5,7) node[circle,fill,color=gray,inner sep=2pt](){}--
            (3,7) node[circle,fill,color=gray,inner sep=2pt](){}--
            (3,9) node[circle,fill,color=gray,inner sep=2pt](){}--
            (1,9) node[circle,fill,color=gray,inner sep=2pt](){}--
            (1,11) node[circle,fill,color=gray,inner sep=2pt](){}--
            (1,13) node[circle,fill,color=gray,inner sep=2pt](){}--(0,13);
        \node[above right] at (0,13) {(a)};
        \draw[ultra thick, dash dot] (30,18)--
            (29,18) node[circle,fill,inner sep=2pt](){}--
            (29,20) node[circle,fill,inner sep=2pt](){}--
            (27,20) node[circle,fill,inner sep=2pt](){}--
            (27,22) node[circle,fill,inner sep=2pt](){}--
            (25,22) node[circle,fill,inner sep=2pt](){}--
            (25,24) node[circle,fill,inner sep=2pt](){}--
            (23,24) node[circle,fill,inner sep=2pt](){}--
            (23,26) node[circle,fill,inner sep=2pt](){}--
            (21,26) node[circle,fill,inner sep=2pt](){}--
            (21,28) node[circle,fill,inner sep=2pt](){}--
            (19,28) node[circle,fill,inner sep=2pt](){}--
            (19,30) node[circle,fill,inner sep=2pt](){}--
            (17,30) node[circle,fill,inner sep=2pt](){};
        \draw[very thick,dash dot] (30,17) node[circle,fill,color=gray,inner sep=2pt](){}--
            (30,19) node[circle,fill,color=gray,inner sep=2pt](){}--
            (28,19) node[circle,fill,color=gray,inner sep=2pt](){}--
            (28,21) node[circle,fill,color=gray,inner sep=2pt](){}--
            (26,21) node[circle,fill,color=gray,inner sep=2pt](){}--
            (26,23) node[circle,fill,color=gray,inner sep=2pt](){}--
            (24,23) node[circle,fill,color=gray,inner sep=2pt](){}--
            (24,25) node[circle,fill,color=gray,inner sep=2pt](){}--
            (22,25) node[circle,fill,color=gray,inner sep=2pt](){}--
            (22,27) node[circle,fill,color=gray,inner sep=2pt](){}--
            (20,27) node[circle,fill,color=gray,inner sep=2pt](){}--
            (20,29) node[circle,fill,color=gray,inner sep=2pt](){}--
            (18,29) node[circle,fill,color=gray,inner sep=2pt](){}--
            (18,30);
        \node[below left] at (29,18) {(b)};
        \draw[ultra thick, dash pattern={on 7pt off 2pt on 1pt off 3pt}] 
            (0,17) node[circle,fill,inner sep=2pt](){}--
            (2,17) node[circle,fill,inner sep=2pt](){}--
            (4,17) node[circle,fill,inner sep=2pt](){}--
            (4,19) node[circle,fill,inner sep=2pt](){}--
            (6,19) node[circle,fill,inner sep=2pt](){}--
            (8,19) node[circle,fill,inner sep=2pt](){}--
            (10,19) node[circle,fill,inner sep=2pt](){}--
            (12,19) node[circle,fill,inner sep=2pt](){}--
            (12,21) node[circle,fill,inner sep=2pt](){}--
            (12,23) node[circle,fill,inner sep=2pt](){}--
            (12,25) node[circle,fill,inner sep=2pt](){}--
            (12,27) node[circle,fill,inner sep=2pt](){}--
            (12,29) node[circle,fill,inner sep=2pt](){}--
            (12,30);
        \draw[very thick, dash pattern={on 7pt off 2pt on 1pt off 3pt}] 
            (0,18) -- 
            (1,18) node[circle,fill,color=gray,inner sep=2pt](){}--
            (3,18) node[circle,fill,color=gray,inner sep=2pt](){}--
            (5,18) node[circle,fill,color=gray,inner sep=2pt](){}--
            (5,20) node[circle,fill,color=gray,inner sep=2pt](){}--
            (7,20) node[circle,fill,color=gray,inner sep=2pt](){}--
            (9,20) node[circle,fill,color=gray,inner sep=2pt](){}--
            (11,20) node[circle,fill,color=gray,inner sep=2pt](){}--
            (11,22) node[circle,fill,color=gray,inner sep=2pt](){}--
            (11,24) node[circle,fill,color=gray,inner sep=2pt](){}--
            (11,26) node[circle,fill,color=gray,inner sep=2pt](){}--
            (11,28) node[circle,fill,color=gray,inner sep=2pt](){}--
            (11,30) node[circle,fill,color=gray,inner sep=2pt](){};
        \node[above] at (1,18) {(c)};
        \draw[ultra thick, dash pattern={on 7pt off 3pt on 1pt off 2pt}] 
            (30,4) node[circle,fill,inner sep=2pt](){}--
            (28,4) node[circle,fill,inner sep=2pt](){}--
            (26,4) node[circle,fill,inner sep=2pt](){}--
            (26,6) node[circle,fill,inner sep=2pt](){}--
            (26,8) node[circle,fill,inner sep=2pt](){}--
            (24,8) node[circle,fill,inner sep=2pt](){}--
            (22,8) node[circle,fill,inner sep=2pt](){}--
            (22,10) node[circle,fill,inner sep=2pt](){}--
            (22,12) node[circle,fill,inner sep=2pt](){}--
            (24,12) node[circle,fill,inner sep=2pt](){}--
            (26,12) node[circle,fill,inner sep=2pt](){}--
            (26,14) node[circle,fill,inner sep=2pt](){}--
            (24,14) node[circle,fill,inner sep=2pt](){}--
            (22,14) node[circle,fill,inner sep=2pt](){}--
            (20,14) node[circle,fill,inner sep=2pt](){}--
            (18,14) node[circle,fill,inner sep=2pt](){}--
            (16,14) node[circle,fill,inner sep=2pt](){}--
            (16,14) node[circle,fill,inner sep=2pt](){}--
            (16,16) node[circle,fill,inner sep=2pt](){}--
            (16,18) node[circle,fill,inner sep=2pt](){}--
            (18,18) node[circle,fill,inner sep=2pt](){}--
            (18,20) node[circle,fill,inner sep=2pt](){}--
            (18,22) node[circle,fill,inner sep=2pt](){}--
            (18,24) node[circle,fill,inner sep=2pt](){}--
            (16,24) node[circle,fill,inner sep=2pt](){}--
            (14,24) node[circle,fill,inner sep=2pt](){}--
            (14,26) node[circle,fill,inner sep=2pt](){}--
            (14,28) node[circle,fill,inner sep=2pt](){}--
            (14,30) node[circle,fill,inner sep=2pt](){};
        \draw[very thick, dash pattern={on 7pt off 3pt on 1pt off 2pt}] (30,5)--
            (29,5) node[circle,fill,color=gray,inner sep=2pt](){}--
            (27,5) node[circle,fill,color=gray,inner sep=2pt](){}--
            (27,7) node[circle,fill,color=gray,inner sep=2pt](){}--
            (27,9) node[circle,fill,color=gray,inner sep=2pt](){}--
            (25,9) node[circle,fill,color=gray,inner sep=2pt](){}--
            (23,9) node[circle,fill,color=gray,inner sep=2pt](){}--
            (23,11) node[circle,fill,color=gray,inner sep=2pt](){}--
            (25,11) node[circle,fill,color=gray,inner sep=2pt](){}--
            (27,11) node[circle,fill,color=gray,inner sep=2pt](){}--
            (27,13) node[circle,fill,color=gray,inner sep=2pt](){}--
            (27,15) node[circle,fill,color=gray,inner sep=2pt](){}--
            (25,15) node[circle,fill,color=gray,inner sep=2pt](){}--
            (23,15) node[circle,fill,color=gray,inner sep=2pt](){}--
            (21,15) node[circle,fill,color=gray,inner sep=2pt](){}--
            (19,15) node[circle,fill,color=gray,inner sep=2pt](){}--
            (17,15) node[circle,fill,color=gray,inner sep=2pt](){}-- 
            (17,17) node[circle,fill,color=gray,inner sep=2pt](){}-- 
            (19,17) node[circle,fill,color=gray,inner sep=2pt](){}-- 
            (19,19) node[circle,fill,color=gray,inner sep=2pt](){}-- 
            (19,21) node[circle,fill,color=gray,inner sep=2pt](){}-- 
            (19,23) node[circle,fill,color=gray,inner sep=2pt](){}-- 
            (19,25) node[circle,fill,color=gray,inner sep=2pt](){}-- 
            (17,25) node[circle,fill,color=gray,inner sep=2pt](){}-- 
            (15,25) node[circle,fill,color=gray,inner sep=2pt](){}--
            (15,27) node[circle,fill,color=gray,inner sep=2pt](){}--
            (15,29) node[circle,fill,color=gray,inner sep=2pt](){}--
            (15,30);
        \node[below left] at (14,30) {(d)};
     \end{tikzpicture}
    \caption{General examples of double staircases of points for possible configurations of initial conditions for face-centered quad equations.}
    \label{fig:doublestaircases}
\end{figure}
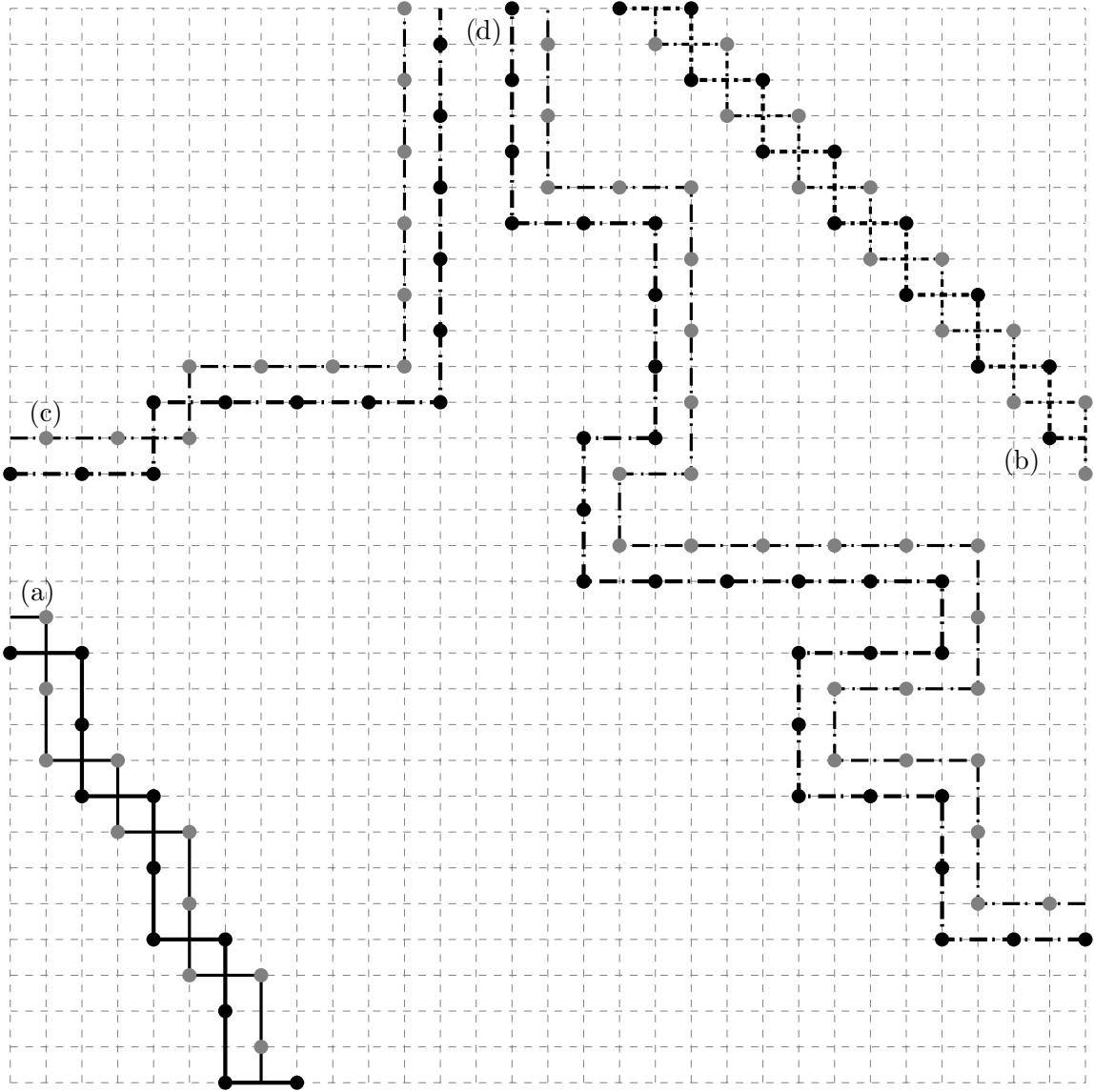

In general, double staircases go from $(-\infty,-\infty)$ to $(\infty, \infty)$, 
or from $(-\infty, +\infty)$ to $(\infty, -\infty)$, implying that
the space of initial conditions is infinite.
Like in the case of quad equations, we decide to restrict ourselves
to a specific kind of double staircases that we call 
\emph{regular double staircases}.
We define a double staircase to be regular when the external points
(the black ones in Figure \ref{fig:doublestaircases}) lie on a 
\emph{regular single staircase}, that is a staircase with 
steps of {constant} horizontal length, and {constant} vertical height.
This will leave some freedom in the choice of the interior points
(that is the grey ones in Figure \ref{fig:doublestaircases}).
We will make a specific choice later.
For example, in Figure \ref{fig:doublestaircases} the double staircases
$(a)$ and $(b)$ would be regular, $(c)$ would be irregular, 
and $(d)$ is excluded since it may lead to incompatibilities.
We remark that for the case of quad equations,
the problem of computing the entropy using irregular (single) staircases
was considered in \cite{Hietarintaetal2019}.

Given a positive integer $N$, a pair of coprime integers
${\mu,\nu}$, and a starting point $\boldsymbol{p}_{0}\in\Cn^{2}$, 
we define a restricted regular single staircase 
$\Delta_{[\mu,\nu]}^{(N)}\left( \boldsymbol{p}_{0} \right)$ to be the set of points:  
\begin{equation}
    \Delta_{[\mu,\nu]}^{(N)}\left( \boldsymbol{p}_{0} \right)=
    \boldsymbol{p}_{0}
    +\left\{\left(  k{\mu}+\sgn\mu j_{1},k\nu \right),  
        \left( (k+1)\mu, k{\nu}+\sgn\nu j_{2}\right)
    \right\}_{\substack{k=0,\dots,N-1\\j_{1}=0,\dots,\abs{\mu}\\j_{2}=0,\dots,\abs{\nu}}}.
    \label{eq:Deltadef}
\end{equation}
Note that for either $\mu=0$ or $\nu=0$, 
$\Delta_{[\mu,\nu]}^{(1)}\left( \boldsymbol{p}_{0} \right)$ collapses to a set of 
points on a straight line, %
and that $\Delta_{[0,0]}^{(1)}\left( \boldsymbol{p}_{0} \right)=\left\{ \boldsymbol{p}_{0} \right\}$. 

Then a \emph{restricted regular double staircase} denoted by 
$\Gamma_{[\mu,\nu]}^{(N)}\left( \boldsymbol{p}_{0} \right)$, can be defined in 
terms of regular single staircases as follows: 
\begin{equation}
    \Gamma_{[\mu,\nu]}^{(N)}\left( \boldsymbol{p}_{0} \right) = 
    2\Delta_{[\mu,\nu]}^{(N)}\left(\tfrac{\boldsymbol{p}_{0}}{2} \right)
    \cup
    2\Delta_{[\mu,\nu]}^{(N-1)}\left(\tfrac{\boldsymbol{p}_{1}}{2} \right)
    \cup
    2\Delta_{[\mu-\sgn\mu,\nu-\sgn\nu]}^{(1)}\left(\tfrac{\boldsymbol{p}_{2}}{2} \right),
    \label{eq:Gammaregdecomp}
\end{equation}
where:
\begin{subequations}
    \begin{align}
        \boldsymbol{p}_{1} &=
        \boldsymbol{p}_{0}+\left(\sgn\mu,\sgn\nu \right),
        \label{eq:m1n1}
        \\
        \boldsymbol{p}_{2} &=
        \boldsymbol{p}_{1} + 2\left( N-1 \right)(\mu,\nu).
        \label{eq:m2n2}
    \end{align}
    \label{eq:mng}%
\end{subequations}
Some examples of restricted regular double staircases are shown in 
Figure \ref{fig:regulardiagonals}. In \eqref{eq:Gammaregdecomp}, $2\Delta_{[\mu,\nu]}^{(N)}\left(\tfrac{\boldsymbol{p}_{0}}{2} \right)$ is a staircase corresponding to black points in Figure \ref{fig:regulardiagonals}, and $2\Delta_{[\mu,\nu]}^{(N-1)}\left(\tfrac{\boldsymbol{p}_{1}}{2} \right)\cup 2\Delta_{[\mu-\sgn\mu,\nu-\sgn\nu]}^{(1)}\left(\tfrac{\boldsymbol{p}_{2}}{2} \right)$ is a staircase corresponding to grey points in Figure \ref{fig:regulardiagonals}. 
The notation $\Gamma_{[\mu,\nu]}\left( \boldsymbol{p}_{0} \right)$ will be used 
when considering infinite (unrestricted) regular double staircases.

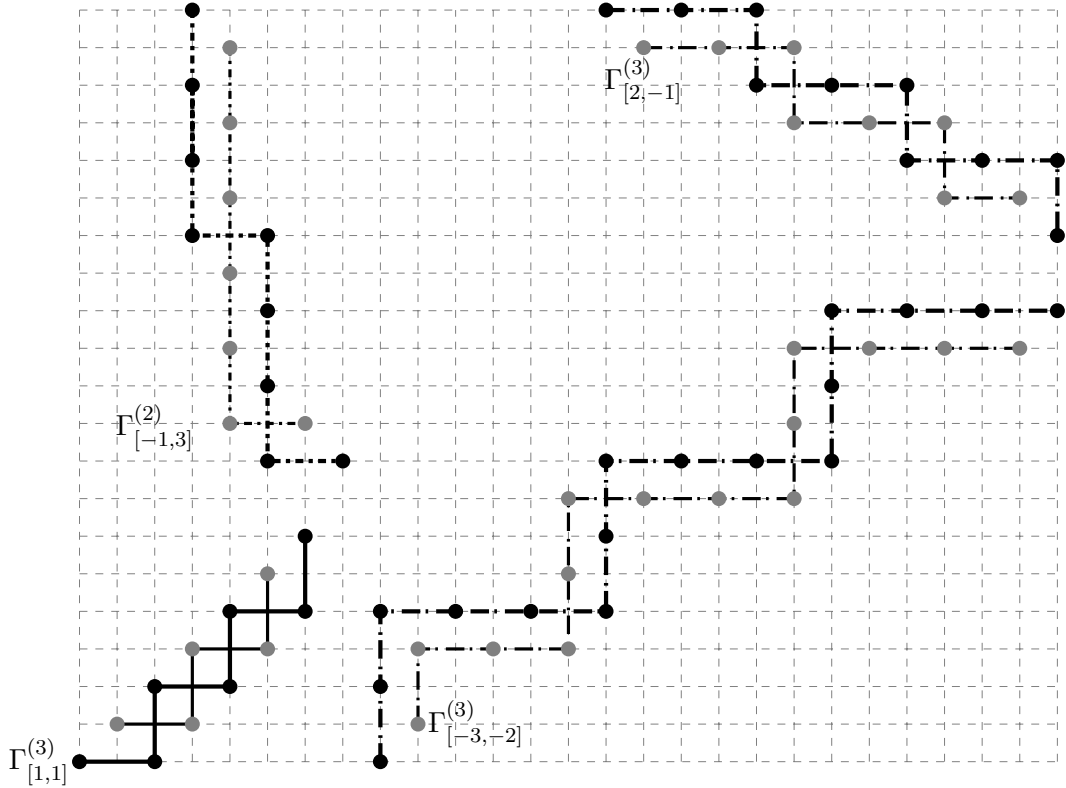
\begin{figure}[htb!]
   \centering
    \begin{tikzpicture}[scale=0.5]
         \draw[style=help lines,dashed] (0,0) grid[step=1cm] (26,20);
         \draw[ultra thick] (0,0) node[circle,fill,inner sep=2pt](){}--
            (2,0) node[circle,fill,inner sep=2pt](){}--
            (2,2) node[circle,fill,inner sep=2pt](){}--
            (4,2) node[circle,fill,inner sep=2pt](){}--
            (4,4) node[circle,fill,inner sep=2pt](){}--
            (6,4) node[circle,fill,inner sep=2pt](){}--
            (6,6) node[circle,fill,inner sep=2pt](){};
            \draw[very thick] (1,1) node[circle,fill,color=gray,inner sep=2pt](){}--
                (3,1) node[circle,fill,color=gray,inner sep=2pt](){}--
                (3,3) node[circle,fill,color=gray,inner sep=2pt](){}--
                (5,3) node[circle,fill,color=gray,inner sep=2pt](){}--
                (5,5) node[circle,fill,color=gray,inner sep=2pt](){};
         \node[left] at (0,0) {$\Gamma_{[1,1]}^{(3)}$};
         \draw[ultra thick, dash dot] (3,20) node[circle,fill,inner sep=2pt](){}--
            (3,16) node[circle,fill,inner sep=2pt](){}--
            (3,18) node[circle,fill,inner sep=2pt](){}--
            (3,14) node[circle,fill,inner sep=2pt](){}--
            (5,14) node[circle,fill,inner sep=2pt](){}--
            (5,12) node[circle,fill,inner sep=2pt](){}--
            (5,10) node[circle,fill,inner sep=2pt](){}--
            (5,8) node[circle,fill,inner sep=2pt](){}--
            (7,8) node[circle,fill,inner sep=2pt](){};
         \draw[very thick, dash dot] (4,19) node[circle,fill,color=gray,inner sep=2pt](){}--
            (4,17) node[circle,fill,color=gray,inner sep=2pt](){}--
            (4,15) node[circle,fill,color=gray,inner sep=2pt](){}--
            (4,13) node[circle,fill,color=gray,inner sep=2pt](){}--
            (4,11) node[circle,fill,color=gray,inner sep=2pt](){}--
            (4,9) node[circle,fill,color=gray,inner sep=2pt](){}--
            (6,9) node[circle,fill,color=gray,inner sep=2pt](){};
         \node[above] at (2,8) {$\Gamma^{(2)}_{[-1,3]}$};
         \draw[ultra thick, dash pattern={on 7pt off 2pt on 1pt off 3pt}] 
            (8,0) node[circle,fill,inner sep=2pt](){}--
            (8,2) node[circle,fill,inner sep=2pt](){}--
            (8,4) node[circle,fill,inner sep=2pt](){}--
            (10,4) node[circle,fill,inner sep=2pt](){}--
            (12,4) node[circle,fill,inner sep=2pt](){}--
            (14,4) node[circle,fill,inner sep=2pt](){}--
            (14,6) node[circle,fill,inner sep=2pt](){}--
            (14,8) node[circle,fill,inner sep=2pt](){}--
            (16,8) node[circle,fill,inner sep=2pt](){}--
            (18,8) node[circle,fill,inner sep=2pt](){}--
            (20,8) node[circle,fill,inner sep=2pt](){}--
            (20,10) node[circle,fill,inner sep=2pt](){}--
            (20,12) node[circle,fill,inner sep=2pt](){}--
            (22,12) node[circle,fill,inner sep=2pt](){}--
            (24,12) node[circle,fill,inner sep=2pt](){}--
            (26,12) node[circle,fill,inner sep=2pt](){};
         \draw[very thick, dash pattern={on 7pt off 2pt on 1pt off 3pt}] 
            (9,1) node[circle,fill,color=gray,inner sep=2pt](){}--
            (9,3) node[circle,fill,color=gray,inner sep=2pt](){}--
            (11,3) node[circle,fill,color=gray,inner sep=2pt](){}--
            (13,3) node[circle,fill,color=gray,inner sep=2pt](){}--
            (13,5) node[circle,fill,color=gray,inner sep=2pt](){}--
            (13,7) node[circle,fill,color=gray,inner sep=2pt](){}--
            (15,7) node[circle,fill,color=gray,inner sep=2pt](){}--
            (17,7) node[circle,fill,color=gray,inner sep=2pt](){}--
            (19,7) node[circle,fill,color=gray,inner sep=2pt](){}--
            (19,9) node[circle,fill,color=gray,inner sep=2pt](){}--
            (19,11) node[circle,fill,color=gray,inner sep=2pt](){}--
            (21,11) node[circle,fill,color=gray,inner sep=2pt](){}--
            (23,11) node[circle,fill,color=gray,inner sep=2pt](){}--
            (25,11) node[circle,fill,color=gray,inner sep=2pt](){};
         \node[right] at (9,1) {$\Gamma^{(3)}_{[-3,-2]}$};
         \draw[ultra thick, dash pattern={on 7pt off 3pt on 1pt off 2pt}] 
            (14,20) node[circle,fill,inner sep=2pt](){}--
            (16,20) node[circle,fill,inner sep=2pt](){}--
            (18,20) node[circle,fill,inner sep=2pt](){}--
            (18,18) node[circle,fill,inner sep=2pt](){}--
            (20,18) node[circle,fill,inner sep=2pt](){}--
            (22,18) node[circle,fill,inner sep=2pt](){}--
            (22,16) node[circle,fill,inner sep=2pt](){}--
            (24,16) node[circle,fill,inner sep=2pt](){}--
            (26,16) node[circle,fill,inner sep=2pt](){}--
            (26,14) node[circle,fill,inner sep=2pt](){};
         \draw[very thick, dash pattern={on 7pt off 3pt on 1pt off 2pt}] 
            (15,19) node[circle,fill,color=gray,inner sep=2pt](){}--
            (17,19) node[circle,fill,color=gray,inner sep=2pt](){}--
            (19,19) node[circle,fill,color=gray,inner sep=2pt](){}--
            (19,17) node[circle,fill,color=gray,inner sep=2pt](){}--
            (21,17) node[circle,fill,color=gray,inner sep=2pt](){}--
            (23,17) node[circle,fill,color=gray,inner sep=2pt](){}--
            (23,15) node[circle,fill,color=gray,inner sep=2pt](){}--
            (25,15) node[circle,fill,color=gray,inner sep=2pt](){};
         \node[below] at (15,19) {$\Gamma_{[2,-1]}^{(3)}$};
     \end{tikzpicture}   
    \caption{Examples of initial conditions on restricted regular double staircases $\Gamma_{[\mu,\nu]}^{(N)}$.}
    \label{fig:regulardiagonals}
\end{figure}

An important observation is that {if we want to investigate the 
    evolution of the equations for a finite number of iterations, 
    then we only need a double staircase of initial conditions with finite extent}.
More specifically, if we fix the initial conditions on
$\Gamma_{[\lambda,\mu]}^{(N)}\left( m_{0},n_{0} \right)$,
then the iterations may be computed over a rectangle of size $4N^{2} \vert \mu \vert\vert \nu \vert$. 
Starting from the double staircase, there is a certain direction that uses 
all initial values, and we will calculate the evolution only in that direction.
This amounts to iteratively computing the values of the dependent variable in the chosen direction.

The total number of initial points in a double staircase
is $M = 2N(|\mu|+|\nu|)$.  We consider the associated initial values in an appropriate 
compactification of $\Cn^{M}$.
In this paper we will consider two different compactifications
of $\Cn^{M}$, namely $\Pj^{M}$ and $\Pj^{\frac{M}{2}+1}\times\Pj^{\frac{M}{2}-1}$.
In the first case we consider the initial conditions to be
parametrised by a line in $\Pj^{M}$ in the following form:
\begin{equation}
    x_{m,n} = \frac{\alpha_{m,n} t +\beta_{m,n}}{\alpha t + \beta},
    \quad \left( m,n \right)\in\Gamma_{[\mu,\nu]}^{(N)},
    \label{eq:xlmini}
\end{equation}
where $\alpha_{m,n},\beta_{m,n},\alpha,\beta\in\Cn$ are fixed parameters.
In the second case we differentiate between the corner vertices and the face vertex of an equation, and parametrise the initial conditions with respective lines in $\Pj^{\frac{M}{2}+1}$
and $\Pj^{\frac{M}{2}-1}$:
\begin{subequations}
    \begin{align}
        x_{m,n} &= \frac{\alpha_{m,n} t +\beta_{m,n}}{\alpha_{1} t + \beta_{1}},
        \quad \left( m,n \right)\in 2\Delta_{[\mu,\nu]}^{(N)}\left( \tfrac{\boldsymbol{p}_{0}}{2} \right),
        \\
        x_{m,n} &= \frac{\alpha_{m,n} t +\beta_{m,n}}{\alpha_{2} t + \beta_{2}},
        \quad \left( m,n \right)\in 2\Delta_{[\mu,\nu]}^{(N-1)}\left( \tfrac{\boldsymbol{p}_{1}}{2} \right)
        \cup 2\Delta_{[\mu-\sgn\mu,\nu-\sgn\nu]}^{(1)}\left( \tfrac{\boldsymbol{p}_{2}}{2} \right),
    \end{align}
    \label{eq:xlmini2}%
\end{subequations}
where $\alpha_{m,n},\beta_{m,n},\alpha_{i},\beta_{i}\in\Cn$ are fixed parameters. 
We emphasise that these two are not the only possible compactifications
of the space $\Cn^{M}$.
However, it is expected that the final
result for algebraic entropy will be independent of the choice of the compactification,
and thus our two choices of initial conditions should be sufficient for our investigations of the algebraic entropy for face-centered quad equations.

\subsubsection{Growth of degrees and algebraic entropy}

Now we may calculate the iterates by
considering the initial values as given in \eqref{eq:xlmini} or \eqref{eq:xlmini2}.
The simplest choice is to apply this construction to the
restricted double staircases $\Gamma_{[\pm 1, \pm 1]}^{(N)}$.
We will call these the \emph{fundamental double staircases}
(the upper index $(N)$ will be omitted for infinite lines).  Furthermore, we 
denote by $\Gamma_{NE}^{(N)}$, $\Gamma_{NW}^{(N)}$, $\Gamma_{SW}^{(N)}$,
and $\Gamma_{SE}^{(N)}$, the fundamental double staircases where
the evolution is directed in the north-east, north-west,
south-west, and south-east directions, respectively. These fundamental double staircases have relatively simple
expressions given by
\begin{subequations}
    \begin{align}
        \Gamma_{NE}^{(N)} &=
        \left\{ \left( -n,n \right)\in\Z^{2} | 0\leq n \leq 2N \right\}
        \cup
        \left\{ \left( -n-2,n \right)\in\Z^{2} | 0\leq n \leq 2N-2 \right\},
        \label{eq:GammaNE}
        \\
        \Gamma_{NW}^{(N)} &=
        \left\{ \left( n,n \right)\in\Z^{2} | 0\leq n \leq 2N \right\}
        \cup
        \left\{ \left( n+2,n \right)\in\Z^{2} | 0\leq n \leq 2N-2 \right\},
        \label{eq:GammaNW}
        \\
        \Gamma_{SW}^{(N)} &=
        \left\{ \left( n,-n \right)\in\Z^{2} | 0\leq n \leq 2N \right\}
        \cup
        \left\{ \left( n+2,-n \right)\in\Z^{2} | 0\leq n \leq 2N-2 \right\},
        \label{eq:GammaSW}
        \\
        \Gamma_{SE}^{(N)} &=
        \left\{ \left( -n,-n \right)\in\Z^{2} | 0\leq n \leq 2N \right\}
        \cup
        \left\{ \left( -n-2,-n \right)\in\Z^{2} | 0\leq n \leq 2N-2 \right\}.
        \label{eq:GammaSE}
    \end{align}
    \label{eq:GammasDir}%
\end{subequations}
The above fundamental double staircases are shown in Figure \ref{fig:princdiag}.
Note that in the case of fundamental regular staircase, there is no
ambiguity in the definition of the inner sequence of points.

\begin{figure}[hbt]
    \centering
   \begin{tikzpicture}[scale=0.5]
        \draw[style=help lines,dashed] (0,0) grid[step=1cm] (24,24);
        \draw[ultra thick] (12,0) node[circle,fill,inner sep=2pt](){}--
            (10,0) node[circle,fill,inner sep=2pt](){}--
            (10,2) node[circle,fill,inner sep=2pt](){}--
            (8,2) node[circle,fill,inner sep=2pt](){}--
            (8,4) node[circle,fill,inner sep=2pt](){}--
            (6,4) node[circle,fill,inner sep=2pt](){}--
            (6,6) node[circle,fill,inner sep=2pt](){}--
            (4,6) node[circle,fill,inner sep=2pt](){}--
            (4,8) node[circle,fill,inner sep=2pt](){}--
            (2,8) node[circle,fill,inner sep=2pt](){}--
            (2,10) node[circle,fill,inner sep=2pt](){}--
            (0,10) node[circle,fill,inner sep=2pt](){}--
            (0,12) node[circle,fill,inner sep=2pt](){};
        \draw[very thick] (11,1) node[circle,fill,color=gray,inner sep=2pt](){}--
            (9,1) node[circle,fill,color=gray,inner sep=2pt](){}--
            (9,3) node[circle,fill,color=gray,inner sep=2pt](){}--
            (7,3) node[circle,fill,color=gray,inner sep=2pt](){}--
            (7,5) node[circle,fill,color=gray,inner sep=2pt](){}-- 
            (5,5) node[circle,fill,color=gray,inner sep=2pt](){}--
            (5,7) node[circle,fill,color=gray,inner sep=2pt](){}--
            (3,7) node[circle,fill,color=gray,inner sep=2pt](){}--
            (3,9) node[circle,fill,color=gray,inner sep=2pt](){}--
            (1,9) node[circle,fill,color=gray,inner sep=2pt](){}--
            (1,11) node[circle,fill,color=gray,inner sep=2pt](){};
        \node[above left] at (10,9) {$\Gamma_{NE}^{(N)}$};
        \draw[thick,->] (8,7) --(10,9);
        \draw[ultra thick, dash dot] (0,16) node[circle,fill,inner sep=2pt](){}--
            (0,18) node[circle,fill,inner sep=2pt](){}--
            (2,18) node[circle,fill,inner sep=2pt](){}--
            (2,20) node[circle,fill,inner sep=2pt](){}--
            (4,20) node[circle,fill,inner sep=2pt](){}--
            (4,22) node[circle,fill,inner sep=2pt](){}--
            (6,22) node[circle,fill,inner sep=2pt](){}--
            (6,24) node[circle,fill,inner sep=2pt](){}--
            (8,24) node[circle,fill,inner sep=2pt](){};
        \draw[very thick,dash dot] (1,17) node[circle,fill,color=gray,inner sep=2pt](){}--
            (1,19) node[circle,fill,color=gray,inner sep=2pt](){}--
            (3,19) node[circle,fill,color=gray,inner sep=2pt](){}--
            (3,21) node[circle,fill,color=gray,inner sep=2pt](){}--
            (5,21) node[circle,fill,color=gray,inner sep=2pt](){}--
            (5,23) node[circle,fill,color=gray,inner sep=2pt](){}--
            (7,23) node[circle,fill,color=gray,inner sep=2pt](){};
        \node[above right] at (6,16) {$\Gamma_{SE}^{(N)}$};
        \draw[thick,->] (4,18)--(6,16);
        \draw[ultra thick, dash pattern={on 7pt off 2pt on 1pt off 3pt}] 
            (12,24) node[circle,fill,inner sep=2pt](){}--
            (14,24) node[circle,fill,inner sep=2pt](){}--
            (14,22) node[circle,fill,inner sep=2pt](){}--
            (16,22) node[circle,fill,inner sep=2pt](){}--
            (16,20) node[circle,fill,inner sep=2pt](){}--
            (18,20) node[circle,fill,inner sep=2pt](){}--
            (18,18) node[circle,fill,inner sep=2pt](){}--
            (20,18) node[circle,fill,inner sep=2pt](){}--
            (20,16) node[circle,fill,inner sep=2pt](){}--
            (22,16) node[circle,fill,inner sep=2pt](){}--
            (22,14) node[circle,fill,inner sep=2pt](){}--
            (24,14) node[circle,fill,inner sep=2pt](){}--
            (24,12) node[circle,fill,inner sep=2pt](){};
        \draw[very thick, dash pattern={on 7pt off 2pt on 1pt off 3pt}] 
            (13,23) node[circle,fill,color=gray,inner sep=2pt](){}--
            (15,23) node[circle,fill,color=gray,inner sep=2pt](){}--
            (15,21) node[circle,fill,color=gray,inner sep=2pt](){}--
            (17,21) node[circle,fill,color=gray,inner sep=2pt](){}--
            (17,19) node[circle,fill,color=gray,inner sep=2pt](){}--
            (19,19) node[circle,fill,color=gray,inner sep=2pt](){}--
            (19,17) node[circle,fill,color=gray,inner sep=2pt](){}--
            (21,17) node[circle,fill,color=gray,inner sep=2pt](){}--
            (21,15) node[circle,fill,color=gray,inner sep=2pt](){}--
            (23,15) node[circle,fill,color=gray,inner sep=2pt](){}--
            (23,13) node[circle,fill,color=gray,inner sep=2pt](){};
        \draw[thick,->] (16,17)--(14,15);
        \node[below left] at (14,15) {$\Gamma_{SW}^{(N)}$};
        \draw[ultra thick, dash pattern={on 7pt off 3pt on 1pt off 2pt}] 
            (24,8) node[circle,fill,inner sep=2pt](){}--
            (24,6) node[circle,fill,inner sep=2pt](){}--
            (22,6) node[circle,fill,inner sep=2pt](){}--
            (22,4) node[circle,fill,inner sep=2pt](){}--
            (20,4) node[circle,fill,inner sep=2pt](){}--
            (20,2) node[circle,fill,inner sep=2pt](){}--
            (18,2) node[circle,fill,inner sep=2pt](){}--
            (18,0) node[circle,fill,inner sep=2pt](){}--
            (16,0) node[circle,fill,inner sep=2pt](){};
        \draw[very thick, dash pattern={on 7pt off 3pt on 1pt off 2pt}] 
            (23,7) node[circle,fill,color=gray,inner sep=2pt](){}--
            (23,5) node[circle,fill,color=gray,inner sep=2pt](){}--
            (21,5) node[circle,fill,color=gray,inner sep=2pt](){}--
            (21,3) node[circle,fill,color=gray,inner sep=2pt](){}--
            (19,3) node[circle,fill,color=gray,inner sep=2pt](){}--
            (19,1) node[circle,fill,color=gray,inner sep=2pt](){}--
            (17,1) node[circle,fill,color=gray,inner sep=2pt](){};
        \node[above right] at (16,6) {$\Gamma_{NW}^{(N)}$};
        \draw[thick,->] (18,4)--(16,6);
     \end{tikzpicture}
     \caption{The four fundamental regular double staircases (for fixed $N$).}
    \label{fig:princdiag}
\end{figure}
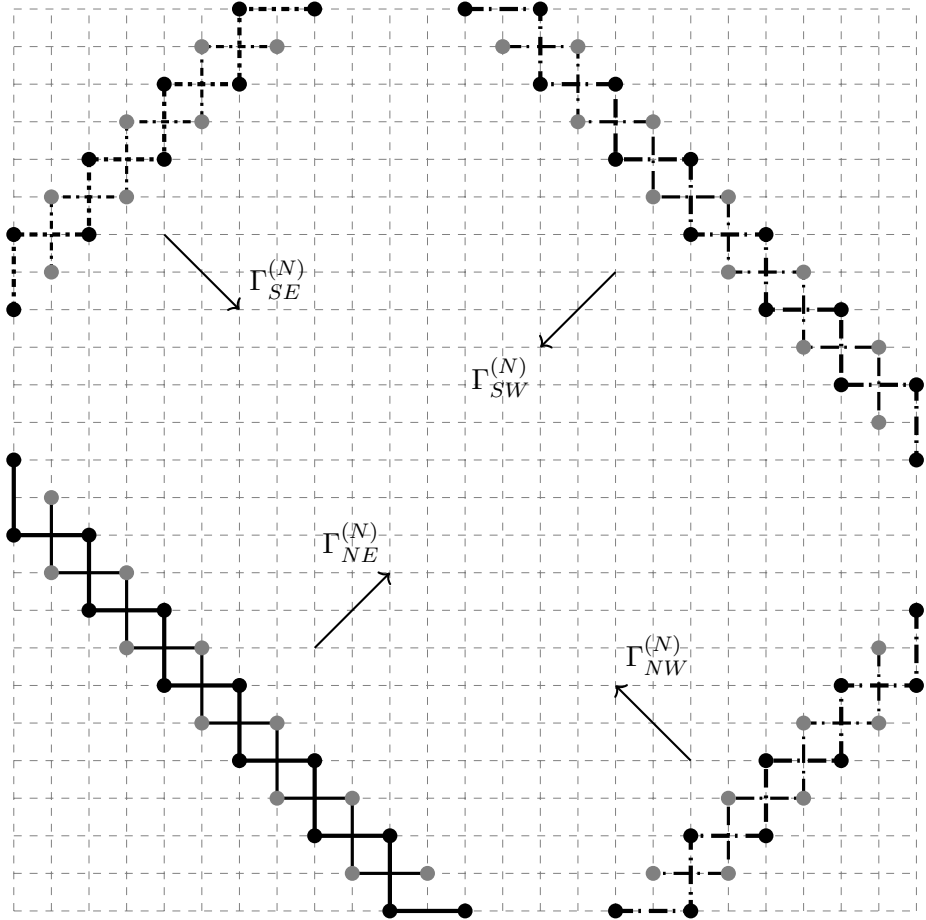

The number of points in the fundamental double staircases \eqref{eq:GammasDir} 
is  $4N$.  Thus the initial condition \eqref{eq:xlmini} on the fundamental 
double staircases is taken in $\Pj^{4N}$. 
Furthermore, according to \eqref{eq:Gammaregdecomp} the fundamental double staircases 
decompose into two regular single staircases which possess $2N+1$ and $2N-1$
points respectively.  Thus the initial condition \eqref{eq:xlmini2} on the fundamental double staircases is taken in $\Pj^{2N+1}\times\Pj^{2N-1}$
For convenience, in the remainder of the paper a special notation $\PTN$ and $\PNpPNm$ will be used to denote the respective spaces of initial conditions for fundamental double staircases as
\begin{align}\label{initialconditionsdef}
    \PTN=\Pj^{4N},\qquad \PNpPNm=\Pj^{2N+1}\times\Pj^{2N-1}.
\end{align}

Consider the evolution from the fundamental double staircase in the NE 
direction. 
From this evolution, sequences of degrees $d^{(i)}_k$, $k=1,\ldots,N$, $i=1,\ldots,N-k$, 
will be generated that take the shape shown in Figure \ref{fig:ourdegrees}.
The degree $d_{k}^{\left( i \right)}$ is computed as the maximum of
the degree of numerator and the denominator of the dependent variable
on the lattice.
In this example of the NE direction, we will call the sequence
\begin{equation}
    1, d_{1}^{(1)}, d_{2}^{(1)}, d_{3}^{(1)}, d_{4}^{(1)},\ldots
    \label{vialletseq}
\end{equation}
the \emph{primary sequence} of growth, while sequence as
\begin{equation}
    1, d_{1}^{(m)}, d_{2}^{(m)}, d_{3}^{(m)}, d_{4}^{(m)},\ldots
    \label{eq:ourseq}
\end{equation}
with $m=2,3,\dots$ will be \emph{secondary sequence} of growth, \emph{tertiary sequence}, 
and so on.  The sequences in the other directions of fundamental double staircases are defined similarly.
In Section \ref{sec:aecafcc}, it will be seen that
CAFCC equations can have up to {four} different growth patterns 
in each direction of evolution.

In general, we take as principal sequence the longest one, and then we 
move towards the direction with less elements.
That is, in the NE and SE direction we move left, while in the NW and
SW direction we move right.
Periodicity along the columns is found looking at repeating patterns
on different columns.

\begin{figure}[htb!]
\begin{equation}\nonumber
    \begin {array}{ccccccccccc} 
        1& &d_{1}^{(N-1)}& &d_{2}^{(N-2)}& &\dots& &d_{N-1}^{(3)}& &d_{N}^{(1)}
        \\ 
        &1& &d_{1}^{(N-2)}& &\dots& &\dots& &d_{N-1}^{(2)}& 
        \\ 
        1& &1& &\dots& &\dots& &\dots& &d_{N-1}^{(1)}
        \\ 
        &1& &1& &\dots& &\dots& &\dots& 
        \\ 
        & &\dots& &\dots& &\dots& &\dots& &\dots
        \\ 
        & & &1& &1& &\dots& &d_{2}^{(2)}& 
        \\ 
        & & & &1& &1& &d_{1}^{(3)}& &d_{2}^{(1)}
        \\ 
        & & & & &1& &1& &d_{1}^{(2)}& 
        \\ 
         & & & & & &1& &1& &d_{1}^{(1)}
        \\ 
         & & & & & & &1& &1& 
        \\ 
         & & & & & & & &1& &1
    \end {array} 
\end{equation}
\caption{Example of sequences of degrees for the fundamental double staircase $\Gamma_{NE}^{(N)}$.}
    \label{fig:ourdegrees}
\end{figure}

The fundamental algebraic entropies of a face-centered quad equation
are given in terms of the generated degrees by:
\begin{equation}
    S^{(i)} = \lim_{k\to\infty}\frac{1}{k}
        \log d_{k}^{(i)}.
    \label{eq:algentdefi}
\end{equation}
Furthermore, we define a \emph{maximal sequence of growth}:
\begin{equation}
    1, D_{1}, D_{2}, D_{3}, D_{4},D_{5},\ldots,
    \label{eq:maxseq}
\end{equation}
where $D_{k} = \max_{m=1,\dots,N-k} d_{k}^{(m)}$.
Associated to \eqref{eq:maxseq}, we define the maximal algebraic entropy of a face-centred 
quad equation:
\begin{equation}
    S_\text{Max} = \lim_{k\to\infty}\frac{1}{k}\log D_{k}.
    \label{eq:algentdefmax}
\end{equation}
The existence of the limits \eqref{eq:algentdefi} and \eqref{eq:algentdefmax} 
may be proven in an analogous way as for finite dimensional systems \cite{BellonViallet1999}.

We have the following general result on the growth of degrees for an 
\emph{autonomous} face-centered quad equation: 

\begin{prop}
    Assume we are given an autonomous face-centered quad equation,
    that is an equation of the form \eqref{afflin} where the polynomials
    $\kappa_{i}$ are \emph{independent of the variables $\left( m,n \right)$}.
    Then its maximal growth is asymptotically given by the solution of 
    the linear difference equation:
    \begin{equation}
        D_{k+1} = \left( K+2 \right) D_{k} + D_{k-1},
        \label{eq:asydef}
    \end{equation}
    that is:
    \begin{equation}
        D_{k} \sim 
        \left[\frac{1}{2}\left( K+2+\sqrt{K^2+4K+8} \right)\right]^{k},
        \quad
        k\to\infty,
        \label{eq:dnasy}
    \end{equation}
    where $K=\max_{i=1,\dots,16}\deg_{\xi} \kappa_{i}\left( \xi \right)$.
    So, for a fixed $K$ the maximal value of the algebraic entropy is:
    \begin{equation}
        S_{K} = \log \left[\frac{1}{2}\left( K+2+\sqrt{K^2+4K+8} \right)\right].
        \label{eq:aemax}
    \end{equation}
    \label{eq:generic}
\end{prop}

\begin{proof}
    Consider equation \eqref{afflin} solved with respect to $x_{m+1,n+1}$.
    Assuming that no factorisation occurs, we can take the degrees 
    over the one-dimensional plaquette in the $(k+1)$-th iteration,
    given in Figure \ref{fig-face}.
    On the right-hand side we have $x_{m,n}$ at most at the $K$th power,
    while the equation is linear in $x_{m+1,n-1}$, $x_{m-1,n+1}$, and
    $x_{m-1,n-1}$.
    This means that at most we have:
    \begin{equation}
        \deg x_{m+1,n+1} =
        K\deg x_{m,n} + \deg x_{m+1,n-1} + \deg x_{m+1,n-1} + \deg x_{m-1,n-1}.
        \label{eq:xlmdeg}
    \end{equation}
    To obtain the estimate \eqref{eq:asydef} it is enough to note
    that we can assume that the maximal sequence $D_{k}$ is such
    that 
    \begin{subequations}
        \begin{gather}
        \deg x_{m+1,n+1} = D_{k+1},
        \\
        \deg x_{m,n} = \deg x_{m+1,n-1} = \deg x_{m+1,n-1} = D_{k},
        \\
        \deg x_{m-1,n-1}=D_{k-1}.
        \end{gather}
        \label{eq:xlmest}%
    \end{subequations}
    Substituting \eqref{eq:xlmest} into \eqref{eq:xlmdeg}, 
    we readily obtain formula \eqref{eq:asydef}.
\end{proof}

We note that for $K=0$, the formula \eqref{eq:asydef} reduces to:
\begin{equation}
    D_{k} \sim \left( 1+\sqrt{2}\right)^{k}.
    \label{eq:dkm0}
\end{equation}
As expected, this is the default growth for regular quad equations 
\cite{HietarintaBook,HietarintaViallet2007}.
This result implies that from the point of view of algebraic entropy
the structure of non-integrable face-centered quad equations is 
much richer than that of regular quad equations. 

Analogously to the case of regular quad equations \cite{HietarintaViallet2007}, according to the 
algebraic entropy and degree of iterates growth for the face-centered quad equations they may be classified as follows:
\begin{description}
    \item[Linear growth:] The equation is linearisable.
    \item[Polynomial growth:] The equation is integrable.
    \item[Exponential growth:] The equation is non-integrable.
\end{description}

\subsubsection{Generating functions of degree growths}

To compute the degree of iterates we adopt some simplification techniques that are used for the case of regular quad equations \cite{GubbiottiASIDE16,GrammaticosHalburdRamaniViallet2009}.
First, we choose all of the parameters involved in the equations to 
be integers.
Furthermore, to avoid potential factorisations that may inadvertently affect the results 
we choose these integers to be prime numbers. 
A final simplification to speed up the computations is given by considering the
factorisation of the iterates in some finite field $\K_{r}$, with $r$
a given prime number.
Using these rules we are able to avoid accidental cancellations
and produce a finite sequence of degrees:
\begin{equation}
    1,d_{1},d_{2},d_{3},d_{4},d_{5},\dots,d_{N}.
    \label{eq:onedimdegreesfin}
\end{equation}
A standard way to extract the asymptotic behaviour from a finite sequence like
\eqref{eq:onedimdegreesfin} is to compute its \emph{generating function}.
A generating function is a function $g=g\left( z \right)$ such
that the coefficients of its Taylor series 
\begin{equation}
    g(z) = \sum_{k=0}^{\infty} d_{k} z^{k}
    \label{eq:genfunc}
\end{equation}
up to order $N$ coincides with the finite sequence \eqref{eq:onedimdegreesfin}.
The adoption of this heuristic method is justified by the strong algebraic geometry 
structure behind the definition of algebraic entropy, see \cite{BellonViallet1999}.
Up to present, in all known examples the method is known to work except,
in some pathological cases, see \cite{Hasselblatt2007}.

If a generating function is rational, as it is safe to assume,
it can be computed exactly using a finite number of iterates
using the method of Pad\'e approximants \cite{Pade1892,BakerGraves1996}. 
An important property of rational generating functions is given in
the following theorem:

\begin{theorem}[\cite{Elaydi2005}]
    A sequence $\left\{ d_{k} \right\}_{k\in\N_{0}}$ admits a rational
    generating function $g\in\Cn\left( z \right)$ if and only if it 
    solves a \emph{linear difference equation with constant coefficients}.
    In such case if $g=P/Q$, with $P,Q\in\Cn\left[ z \right]$ and the
    polynomial $Q$ is monic, then the difference equation is given by 
    $Q\left( T_{k}^{-1} \right)\left( d_{k} \right)=0$, where $T_{k}d_{k}=d_{k+1}$
    is the translation operator.
    \label{th:rgf}
\end{theorem}

From Theorem \ref{th:rgf} we have that the order of the difference
equation solved by the sequence $\left\{ d_{k} \right\}_{k\in\N_{0}}$
is equal to the degree of the denominator of the generating function.
Moreover, we see that the numerator of the generating function
acts as a sort of ``noise'' linked to the initial conditions of the
sequence.
For instance, this implies that if two sequence have the same denominator
but different numerators, they will satisfy the same difference equation with
different initial conditions.
Several examples of this kind of occurrence will appear in Section
\ref{sec:aecafcc}.

Once obtained a generating function is a \emph{predictive} tool.
Indeed one can readily compute the successive terms
in the Taylor expansion for \eqref{eq:genfunc} and confront them with
the degrees calculated with the iterations. 
This means that the assumption that the value of the algebraic entropy is given 
by the approximate method is in fact very strong and it is very unlikely that the 
real value will differ from it. 

Having a rational generating function will also yield the value
of the algebraic entropy from the modulus of the
smallest pole of the generating function:
\begin{equation}
    S = 
    \log \min\left\{|z|\in\R^{+} \,\,\middle|\,\, \frac{1}{g\left( z \right)}=0\right\}^{-1}.
    \label{eq:algentgenfunc}
\end{equation}
Let us underline that while the value of the algebraic entropy obtained
through \eqref{eq:algentgenfunc} is canonical as it depends only on the 
equation itself, the sequence of degrees and its generating function 
are not canonical quantities and might
change under appropriate invertible transformations.

From the generating function one can also find an asymptotic fit for 
the degrees \eqref{eq:onedimdegreesfin}. 
This can be done by using the inverse $\mathcal{Z}$-transform \cite{Elaydi2005,Jury1964}.
Assume we are given a function $f=f\left( \zeta \right)$ of 
a complex variable $\zeta\in \Cn$
analytic in a region $\abs{\zeta}>r$ for some $r\in\R^{+}$.
We define the inverse $\mathcal{Z}$-transform of such a function $f$ 
to be the sequence:
\begin{equation}
    \mathcal{Z}^{-1}\left[ f(\zeta) \right]_{k} \equiv
    \frac{1}{2\pi\imath} \oint_{C} f\left( \zeta \right) \zeta^{k-1} \ud \zeta,
    \quad k \in \N.
    \label{eq:invztransfdef}
\end{equation}
In equation \eqref{eq:invztransfdef} the contour $C\subset \Cn$
is a counterclockwise closed path enclosing the origin and entirely
in the region of convergence of $f$.
From the definition of inverse $\mathcal{Z}$-transform \eqref{eq:invztransfdef} 
it can be readily proved that the sequence $\left\{ d_{k} \right\}_{k\in\N}$ 
corresponding to the generating function \eqref{eq:genfunc} is given by:
\begin{equation}
    d_{k} = \mathcal{Z}^{-1}\left[g\left( \frac{1}{\zeta} \right)\right]_{k}.
    \label{eq:ztransform}
\end{equation}

We note that the  general asymptotic behaviour of the
sequence $\left\{ d_{k} \right\}_{l\in\N_{0}}$ can be obtained even without 
computing the inverse $\mathcal{Z}$-transform.
This is the content of the following proposition:
\begin{prop}[\cite{FlajoletOdlyzko1990}]
    Assume that a sequence $\left\{ d_{k} \right\}_{k\in\N_{0}}$
    possesses a generating function of radius of convergence
    $\rho>0$ and of the following form:
    \begin{equation}
        g = A\left( z \right) + B\left( z \right)\left( 1-\frac{z}{\rho} \right)^{-\beta},
        \quad \beta\in\R\setminus\left\{ -n \right\}_{n\in\N}.
        \label{eq:gasy}
    \end{equation}
    where $A$ and $B$ are analytic functions for $\abs{z}<r$ such that
    $B(\rho)\neq 0$.
    Then the asymptotic behaviour of the sequence $\left\{ d_{k} \right\}_{k\in\N_{0}}$
    as $k\to\infty$ is given by:
    \begin{equation}
        d_{k} \sim \frac{B\left( \rho \right)}{\Gamma\left( \beta \right)}
        k^{\beta-1}\rho^{-k},
        \quad
        k \to \infty,
        \label{eq:dasy}
    \end{equation}
    where $\Gamma(z)$ is the Euler Gamma function.
    If additionally $\rho\equiv1$, then
    \begin{equation}
        d_{k} \sim k^{\beta-1},
        \quad
        k \to \infty,
        \label{eq:dasyint}
    \end{equation}
    i.e. the growth is asymptotically polynomial of degree $\beta-1$.
    \label{rem:asydeg}
\end{prop}

Notice that when the generating function is polynomial and the radius of 
convergence is $\rho=1$, usually the denominator can be factorised as follows:
\begin{equation}
    Q\left( z \right) = \left( 1-z \right)^{\beta_{0}} 
    \prod_{i=1}^{K} \left( 1-z^{\beta_{i}} \right).
    \label{eq:factor}
\end{equation}
That is, the polynomial $Q$ is the product of the term $\left( 1-z \right)^{\beta_{0}}$
with some cyclotomic polynomials $p_{i}\left( z \right)=1-z^{\beta_{i}}$.
From Proposition \eqref{rem:asydeg}, using the factorisation properties of
cyclotomic polynomials, we get that $d_{k} \sim k^{\beta}$ as $k\to\infty$,
where $\beta=\beta_{0}+K$.
We will make constant use of this observation to estimate the asymptotic
growth of degrees of the face-centered quad equations in Section \ref{sec:aecafcc}.

\subsubsection{Examples}

We conclude this section with some illustrative examples of the computation
of algebraic entropy that has been formulated above for face-centered quad equations.

\begin{example}
    Consider the following face-centered quad equation:
    \begin{equation}
        Ax_{m+1, n+1} x_{m-1, n-1}+B x_{m-1, n+1}x_{m+1, n-1}
        =
        C x_{m, n}^2,
        \label{eq:example1}
    \end{equation}
    where $A$, $B$, and $C$ are arbitrary constants.

    Considering initial conditions in $\PTN$ we have a single
    sequence of degrees of iterates:
    \begin{equation}
        1, 2, 6, 14, 26, 42, 62, 86, 114, 146, 182, 222, 266\dots
        \label{eq:ex1degmy}
    \end{equation}
    The corresponding generating function is:
    \begin{equation}
        g_{\PTN}\left( z \right) = \frac{z^3+3 z^2-z+1}{(1-z)^3}.
        \label{eq:ex1gfmy}
    \end{equation}
    From this, we infer that the algebraic entropy is zero
    (all the singularities of $g_{\PTN}$ lie on the unit circle)
    and from Proposition \ref{rem:asydeg} that the growth is
    quadratic.
    
    On the other hand,
    considering initial conditions in $\PNpPNm$ we have the single
    sequence of degrees of iterates:
    \begin{equation}
        1, 4, 12, 26, 46, 72, 104, 142, 186, 236, 292, 354, 422\dots
        \label{eq:ex1degkels}
    \end{equation}
    The corresponding generating function is:
    \begin{equation}
        g_{\PNpPNm} \left( z \right) =
        \frac{z^3+3 z^2+z+1}{(1-z)^3}.
        \label{eq:ex1gfkels}
    \end{equation}
    Again, from Proposition \ref{rem:asydeg}, we have integrability
    and quadratic growth.
    \hfill
    $\square$
\end{example}

\begin{example}
    Consider the following deformation of face-centered quad equation
    \eqref{eq:example1}:
    \begin{equation}
        Ax_{m+1, n+1} x_{m-1, n-1}+B x_{m-1, n+1}x_{m+1, n-1}
        =
        C x_{m, n}^2+D x_{m, n},
        \label{eq:example2}
    \end{equation}
    where $A$, $B$, $C$, and $D$ are arbitrary constants.
    Indeed, if $D=0$ then \eqref{eq:example2} reduces to 
    \eqref{eq:example1}.

    Considering initial conditions in $\PTN$ we have a single
    sequence of degrees of iterates:
    \begin{equation}
        1, 2, 6, 16, 40, 100, 252, 636, 1604, 4044, 10196\dots
        \label{eq:ex2degmy}
    \end{equation}
    The corresponding generating function is:
    \begin{equation}
        g_{\PTN}\left( z \right) =
        -\frac{2 z^2-z+1}{2 z^3-2 z^2+3 z-1}.
        \label{eq:ex2gfmy}
    \end{equation}
    From this, we infer that the algebraic entropy is positive.
    Indeed, the smallest singularity of $g_{\PTN}$ in absolute value is:
    \begin{equation}
        z_{0} =\frac{1}{3}+ \frac{1}{6} \left(8+6 \sqrt{78}\right)^{1/3}
        -\frac{7}{3 \left(8+6 \sqrt{78}\right)^{1/3}}
        \simeq 0.3966082528\dots < 1
        \label{eq:ex2z0}
    \end{equation}
    The value of the entropy is then 
    $S = \log z_{0}^{-1}\simeq \log\left( 2.521379706\dots \right)$.
    Notice that despite the growth being exponential, it is not
    maximal.
    Indeed, since \eqref{eq:example2} has degree 2 in $x_{l,m}$, from
    equation \eqref{eq:aemax} the maximal entropy is
    $S_{2} = \log\left( 2 + \sqrt{5} \right)$, which is clearly greater
    than $S$.
    
    If we fix our initial conditions in $\PNpPNm$ we have the single
    sequence of degrees of iterates:
    \begin{equation}
        1, 4, 12, 30, 74, 186, 470, 1186, 2990, 7538, 19006\dots,
        \label{eq:ex2degkels}
    \end{equation}
    fitted by the generating function:
    \begin{equation}
        g_{\PNpPNm} \left( z \right) =
        -\frac{2 z^2+z+1}{2 z^3-2 z^2+3 z-1}.
        \label{eq:ex2gfkels}
    \end{equation}
    Since the denominator of $g_{\PNpPNm}$ is the same as $g_{\PTN}$,
    we obtain, as expected, the same value of the algebraic entropy.
    \hfill
    $\square$
\end{example}

\begin{example}
    Consider the following face-centered quad equation:
    \begin{equation}
        \begin{gathered}
            P_{0}\left( x_{m,n} \right) x_{m-1,n+1} x_{m+1,n+1} x_{m-1,n-1} x_{m+1,n-1}
            \\
            +P_{1}\left( x_{m,n} \right)
            \left(x_{m-1,n+1}-x_{m+1,n+1}\right) 
            \left(x_{m-1,n-1}-x_{m+1,n-1}\right)
            =P_{2}\left( x_{m,n} \right).
        \end{gathered}
        \label{eq:example3}
    \end{equation}
    where $P_{i}\left( \xi \right)\in\Cn\left[ \xi \right]$ and
    $\deg P_{i} =2$.

    Considering initial conditions either in $\PTN$ or in $\PNpPNm$ we 
    have a single sequence of degrees of iterates:
    \begin{equation}
        1, 5, 21, 89, 377, 1597, 6765, 28657\dots
        \label{eq:ex3degmy}
    \end{equation}
    The corresponding generating function is:
    \begin{equation}
        g\left( z \right) =
        -\frac{1+z}{z^2+4 z-1}.
        \label{eq:ex3gfmy}
    \end{equation}
    So, we infer that the algebraic entropy is positive and \emph{maximal}.
    Indeed, the smallest singularity of $g$ in absolute value is
    $z_{0} = -2+\sqrt{2}$.
    This gives value of the entropy as 
    $S_{2} = \log z_{0}^{-1} = \log\left( 2 + \sqrt{5} \right)$,
    that is the maximal value of the entropy for $M=2$ in formula \eqref{eq:aemax}.
    \hfill
    $\square$
\end{example}

\section{Arrangements of CAFCC equations in the lattice and algebraic entropy}
\label{sec:aecafcc}

In this section we present the results of the heuristic computations
of the algebraic entropy on the three types of CAFCC equations.
The results of these computations are summarised as follows:
\begin{itemize}
    \item Type-A CAFCC equations possess quadratic growth.%
    \item Type-B CAFCC equations possess linear growth.%
    \item Type-C CAFCC equations possess quadratic growth.%
\end{itemize}

In the following subsections we will give the details of arrangements
of equations in the lattice and the
explicit growth patterns for each CAFCC equation found from
\cite{Kels2020cafcc} under the two types of initial conditions
\eqref{initialconditionsdef}.  In particular, it is sometimes the case
that individual type-B or type-C equations have exponential growth
patterns, but certain pairs of such equations together will give
sub-exponential growth patterns.  These cases of pairs of equations
possess multiple growth patterns, each with the same asymptotic
behaviour.  Furthermore, for one particular type-C equation a new
type-C CAFCC equation was needed to achieve quadratic degree growth
($C1_{(\delta_1=1)}$ in \eqref{c1d}), and this equation was not found from
the original method used to obtain CAFCC equations \cite{Kels2020cafcc}.

To describe the arrangements of equations and initial conditions,
it is useful to colour the vertices of the rotated bipartite square
lattice such that black vertices are connected only to white vertices
(and vice versa).  The convention will be that black vertices are
associated to variables $\ccx$ and the white vertices are associated
to variables $\ccy$.  Such a black and white colouring of the lattice
is shown in Figure \ref{fig:lattice-A}.  In the example of Figure
\ref{fig:lattice-A}, the initial conditions are indicated by the square
vertices, and the evolution of the system of face-centered quad equations
would be in the north-east direction.  The equations in the lattice can
be distinguished according to the numbers of single- and double-line
edges that are connected to a vertex.    For example, according to the
graphical representation of face-centered quad equations shown in Figures
\ref{fig-face} and \ref{fig:3fig4quad}, Figure \ref{fig:lattice-A}
involves all type-A equations, Figure \ref{fig:lattice-Bsym} involves
all type-B equations, and Figure \ref{fig:lattice-Csym} involves all
type-C equations.  There are more complicated arrangements of three or
more different equations in the lattice that are also expected to give a
polynomial degree growth, but for this paper only the minimal number of
different equations (either individual or pairs) to achieve polynomial
degree growth are considered.  %

\begin{figure}[htb!]
\centering
\begin{tikzpicture}[scale=0.85]

\draw[-] (-4,2)--(-2,4);
\draw[-] (-3,1)--(0,4);
\draw[-] (-2,0)--(2,4);
\draw[-] (-1,-1)--(4,4);
\draw[-] (0,-2)--(4,2);
\draw[-] (1,-3)--(4,0);
\draw[-] (2,-4)--(4,-2);

\draw[-] (-4,4)--(4,-4);
\draw[-] (-2,4)--(4,-2);
\draw[-] (0,4)--(4,0);
\draw[-] (2,4)--(4,2);

\foreach \x in {-4,-2,...,4}{
 \foreach \y in {-4,-2,...,4}{
  \fill[black] (\x,\y) circle (3.0pt);}}

\foreach \x in {-3,-1,...,3}{
 \foreach \y in {-3,-1,...,3}{
  \filldraw[fill=white,draw=black] (\x,\y) circle (3.0pt);}}

\foreach \x in {-4,-2,...,4}{
\filldraw[fill=black,draw=black] (\x,-\x) \Square{3.0pt};}
\foreach \x in {-3,-1,...,3}{
\filldraw[fill=white,draw=black] (\x,-\x) \Square{3.0pt};}

\foreach \x in {-3,-1,...,3}{
\filldraw[fill=black,draw=black] (\x-1,-\x-1) \Square{3.0pt};
}%
\foreach \x in {-2,0,...,2}{
\filldraw[fill=white,draw=black] (\x-1,-\x-1) \Square{3.0pt};
}%

\filldraw[fill=white,draw=white] (-4.5,-4.5)--(1.0,-4.5)--(-4.5,1.0)--(-4.5,4.5);

\foreach \x in {-4,...,3}{
 \draw[-,dotted] (\x+0.5,-4.5)--(\x+0.5,4.5);
 \draw[-,dotted] (-4.5,\x+0.5)--(4.5,\x+0.5);}

\foreach \x in {-4,-2,...,3}{
 \fill (\x+0.5,-4.5) circle (0.01pt)
 node[below=2pt]{\small $\ccqa$};
 \fill (\x+1.5,-4.5) circle (0.01pt)
 node[below=2pt]{\small $\ccqb$};
 \fill (\x+0.5,4.5) circle (0.01pt)
 node[above=2pt]{\small $\ccqa$};
 \fill (\x+1.5,4.5) circle (0.01pt)
 node[above=2pt]{\small $\ccqb$};
 \fill (-4.5,\x+0.5) circle (0.01pt)
 node[left=2pt]{\small $\ccpa$};
 \fill (-4.5,\x+1.5) circle (0.01pt)
 node[left=2pt]{\small $\ccpb$};
 \fill (4.5,\x+0.5) circle (0.01pt)
 node[right=2pt]{\small $\ccpa$};
 \fill (4.5,\x+1.5) circle (0.01pt)
 node[right=2pt]{\small $\ccpb$};}

\end{tikzpicture}
\caption{Arrangement of type-A equations in the lattice with initial conditions (square vertices) for evolution in the NE direction.  Type-A equations \eqref{Awhite} are centered at white vertices labelled $\ccy$ and type-A equations \eqref{Ablack} are centered at black vertices labelled $\ccx$.}
\label{fig:lattice-A}
\end{figure}
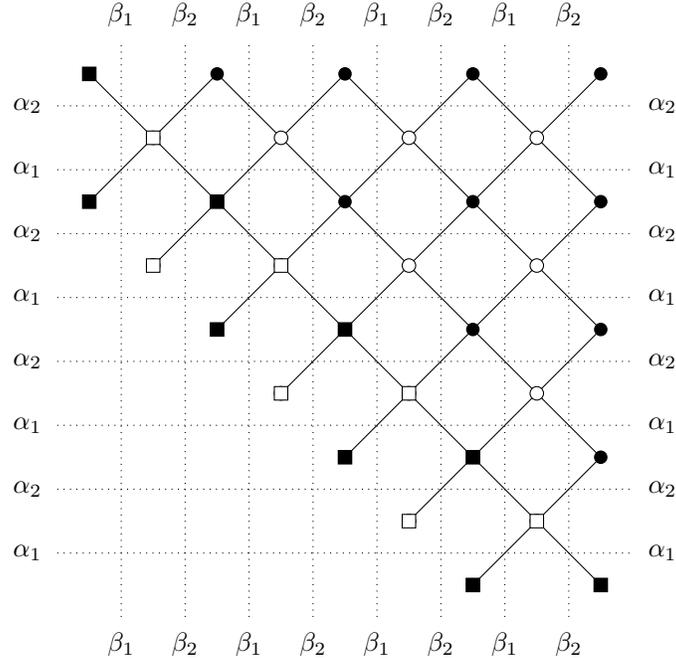

In the following, $\hat{\ccpp}$ and $\hat{\ccqq}$ are used to denote %
\begin{align}
\hat{\ccpp}=(\ccpb,\ccpa),\qquad \hat{\ccqq}=(\ccqb,\ccqa).
\end{align}

\subsection{Type-A}

\subsubsection{Arrangement of equations in the lattice}

The arrangement of type-A equations in the lattice is the most straightforward and is shown in the example of Figure \ref{fig:lattice-A}.  The type-A equations centered at white vertices are
\begin{align}\label{Awhite}
\A{\ccy}{\cca}{\ccb}{\ccc}{\ccd}{\ccpp}{\ccqq}=0,
\end{align}
and the type-A equations centered at black vertices are 
\begin{align}\label{Ablack}
\A{\ccx}{\ccya}{\ccyb}{\ccyc}{\ccyd}{\hat{\ccpp}}{\hat{\ccqq}}=0.
\end{align}

\subsubsection{Degree growths for type-A equations}

The two type-A equations $A3_{\left( \delta \right)}$ and $A2_{\left( \delta_1;\,\delta_2 \right)}$ are given in equations \eqref{a3d} and \eqref{a2dd}, respectively.  For values $\delta=0,1$, and $(\delta_1,\delta_2)=(0,0),(1,0),(1,1)$, these equations each share the same growth patterns when iterated in the lattice of Figure \ref{fig:lattice-A}.
In $\PTN$ the growth pattern in all directions is:
\begin{equation}
    1, 3, 7, 13, 21, 31, 43, 57, 73, 91, 111, 133, 157\dots,
    \label{eq:Admy}
\end{equation}
while in $\PNpPNm$  the growth in all directions is:
\begin{equation}
    1, 5, 13, 25, 41, 61, 85, 113, 145, 181, 221, 265, 313\dots.
    \label{eq:Akels}
\end{equation}
The generating functions for these patterns are respectively given by
\begin{equation}
    g_{\PTN}(z) = \frac{z^2+1}{(1-z)^3},
    \qquad
    g_{\PNpPNm}(z)=\frac{(z+1)^2}{(1-z)^3}.
    \label{eq:Agf}
\end{equation}
Since all the zeroes of the generating functions lie
on the unit circle  for each type-A equation 
the algebraic entropy vanishes.
Moreover, due to the presence of
$\left( 1-z \right)^{3}$ in the denominator we have from Proposition \ref{rem:asydeg} that the growth 
is \emph{quadratic}.

We emphasise two important facts on the growth pattern of
the type-A equations.
The first one is that the growth is the same in all directions
because type-A equations on the lattice are symmetric with
respect to the exchange $\left( m,n \right)\leftrightarrow\left( n,m \right)$.
The second is that the growth stays the same regardless of the value
of the parameter $\delta$ in $A3_{(\delta)}$, or the  value of the parameter $\delta_1$ in $A2_{(\delta_1;0)}$.
For $\delta\neq0$ or $\delta_1\neq0$, this follows from the fact that these parameters
can be scaled to $1$.  
On the other hand, when $\delta=0$ or $\delta_1=0$ there is no difference in
the growth of the degrees for the respective equations, so it can be seen that their behaviour does not depend on these parameters.

As observed in \cite{Kels2020cafcc}, the type-A CAFCC equations are equivalent 
to discrete Laplace-type equations that are associated to type-Q ABS equations \cite{ABS2003}.
Thus it is not surprising to find that their growth pattern in $\PTN$ 
is the same as the most general type-Q equation: the $Q_V$ equation
\cite{Viallet2009}. 
The growth pattern of $Q_{V}$ has been proven rigorously
in independent works \cite{Viallet2015,RobertsTran2017},
where the proof of \cite{RobertsTran2017} used the
$\gcd$-factorisation method that also worked for the case of a two-periodic 
extension of the $Q_{V}$ equation \cite{GSL_QV}.

\subsection{Type-B (individual)}

In the following the algebraic entropy for individual type-B equations in the lattice will be considered.  However, some examples of type-B equations have exponential growth when considered individually in the lattice.  Only the cases of individual type-B equations that achieve sub-exponential growth will be considered here, while the remaining type-B equations will be treated in Section \ref{sec:typeBpairs}.

\subsubsection{Arrangement of equations in the lattice}

The arrangement of type-B equations in the lattice is shown in Figure \ref{fig:lattice-Bsym}.  The type-B equations centered at white vertices are
\begin{align}\label{Bwhitesym}
\B{\ccy}{\cca}{\ccb}{\ccc}{\ccd}{\ccpp}{\ccqq}=0,
\end{align}
and type-B equations centered at black vertices are 
\begin{align}\label{Bblacksym}
\B{\ccx}{\ccya}{\ccyb}{\ccyc}{\ccyd}{\hat{\ccpp}}{\hat{\ccqq}}=0.
\end{align}
The arrangement of equations in Figure \ref{fig:lattice-Bsym} is equivalent that of Figure \ref{fig:lattice-A}, but with double edges instead of single edges that are used to represent type-B equations rather than type-A equations.  Because the same type-B equation is centered at both the black vertices and white vertices, these systems of equations will be referred to as $B(\x;\y)$ systems.

\begin{figure}[htb!]
\centering
\begin{tikzpicture}[scale=0.85]

\draw[-,double,thick] (-4,2)--(-2,4);
\draw[-,double,thick] (-3,1)--(0,4);
\draw[-,double,thick] (-2,0)--(2,4);
\draw[-,double,thick] (-1,-1)--(4,4);
\draw[-,double,thick] (0,-2)--(4,2);
\draw[-,double,thick] (1,-3)--(4,0);
\draw[-,double,thick] (2,-4)--(4,-2);

\draw[-,double,thick] (-4,4)--(4,-4);
\draw[-,double,thick] (-2,4)--(4,-2);
\draw[-,double,thick] (0,4)--(4,0);
\draw[-,double,thick] (2,4)--(4,2);

\foreach \x in {-4,-2,...,4}{
 \foreach \y in {-4,-2,...,4}{
  \fill[black] (\x,\y) circle (3.0pt);}}

\foreach \x in {-3,-1,...,3}{
 \foreach \y in {-3,-1,...,3}{
  \filldraw[fill=white,draw=black] (\x,\y) circle (3.0pt);}}

\foreach \x in {-4,-2,...,4}{
\filldraw[fill=black,draw=black] (\x,-\x) \Square{3.0pt};}
\foreach \x in {-3,-1,...,3}{
\filldraw[fill=white,draw=black] (\x,-\x) \Square{3.0pt};}

\foreach \x in {-3,-1,...,3}{
\filldraw[fill=black,draw=black] (\x-1,-\x-1) \Square{3.0pt};
}%
\foreach \x in {-2,0,...,2}{
\filldraw[fill=white,draw=black] (\x-1,-\x-1) \Square{3.0pt};
}%

\filldraw[fill=white,draw=white] (-4.5,-4.5)--(1.0,-4.5)--(-4.5,1.0)--(-4.5,4.5);

\foreach \x in {-4,...,3}{
 \draw[-,dotted] (\x+0.5,-4.5)--(\x+0.5,4.5);
 \draw[-,dotted] (-4.5,\x+0.5)--(4.5,\x+0.5);}

\foreach \x in {-4,-2,...,3}{
 \fill (\x+0.5,-4.5) circle (0.01pt)
 node[below=2pt]{\small $\ccqa$};
 \fill (\x+1.5,-4.5) circle (0.01pt)
 node[below=2pt]{\small $\ccqb$};
 \fill (\x+0.5,4.5) circle (0.01pt)
 node[above=2pt]{\small $\ccqa$};
 \fill (\x+1.5,4.5) circle (0.01pt)
 node[above=2pt]{\small $\ccqb$};
 \fill (-4.5,\x+0.5) circle (0.01pt)
 node[left=2pt]{\small $\ccpa$};
 \fill (-4.5,\x+1.5) circle (0.01pt)
 node[left=2pt]{\small $\ccpb$};
 \fill (4.5,\x+0.5) circle (0.01pt)
 node[right=2pt]{\small $\ccpa$};
 \fill (4.5,\x+1.5) circle (0.01pt)
 node[right=2pt]{\small $\ccpb$};}

\end{tikzpicture}
\caption{Arrangement of type-B equations in the lattice.  This is equivalent to the arrangement of Figure \ref{fig:lattice-A} but with double edges to indicate type-B equations instead of type-A equations.}
\label{fig:lattice-Bsym}
\end{figure}
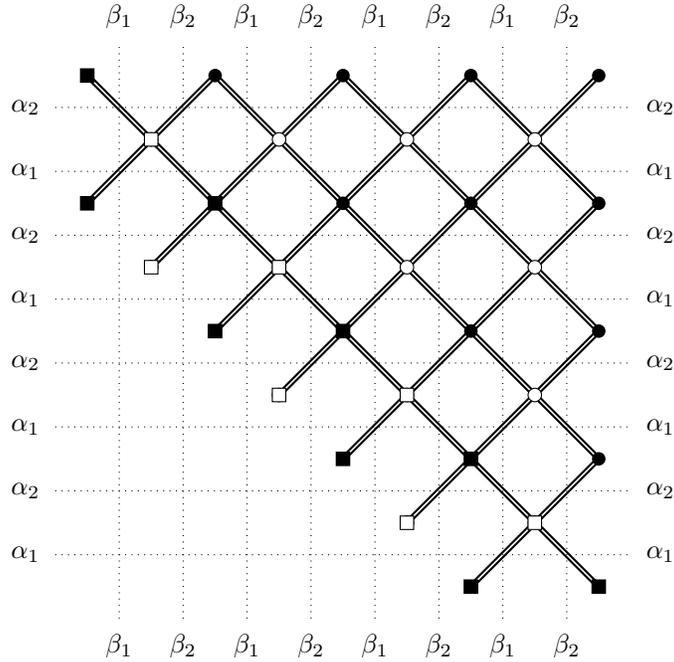

Note that there are two degenerate type-B equations given by
\begin{equation}
    B2_{\left( 0;\,0;\,0 \right)}=
    B3_{\left( 0;\,0;\,0 \right)}=
    y_{l+1, m+1} y_{l-1, m-1}-y_{l-1, m+1} y_{l+1, m-1}=0.
    \label{eq:triv}
\end{equation}
This equation \eqref{eq:triv} is not a true face-centered quad equation because it has no dependence on the central face variable.
Considering the double pass lattice
$\left( m,n \right)\rightarrow \left( 2M,2N \right)$
equation \eqref{eq:triv} can be reduced to the quad equation
\begin{equation}
    y_{M+1,N+1}y_{M,N}-y_{M+1,N}y_{M,N+1}=0.
    \label{eq:expwave}
\end{equation}
This equation is a well-known Darboux integrable quad equation
\cite{AdlerStartsev1999} with isotropic linear growth.  It is also a special case of the $D4$ quad equation that appeared in \cite{Boll2011}.

\subsubsection{$B2_{\left( 1;\,0;\,0 \right)}(\x;\y)$ system}

The equation $B2_{(\delta_1;\, \delta_2;\, \delta_3)}$ is given in \eqref{b2ddd}.
For the case of $(\delta_1,\delta_2,\delta_3)=(1,0,0)$,
This equation exhibits a single growth pattern in 
$\PTN$ and in $\PNpPNm$, respectively.

In $\PTN$ the growth pattern is the following:
\begin{equation}
    1, 2, 3, 4, 5, 6, 7, 8, 9, 10, 11, 12, 13\dots
    \label{eq:B2100}
\end{equation}
This growth is clearly linear and is fitted by the following
generating function:
\begin{equation}
    g_{\PTN}(z)%
    =\frac{1}{(1-z)^2}.
    \label{eq:gfB2100}
\end{equation}
The linearity of the growth \eqref{eq:B2100} could also be confirmed
using Proposition \ref{rem:asydeg}, because of the presence of
the term $\left( 1-z \right)^{2}$ in the denominator of
\eqref{eq:gfB2100}.

In $\PNpPNm$ the growth pattern is:
\begin{equation}
    1, 3, 5, 7, 9, 11, 13, 15, 17, 19, 21, 23, 25\dots
    \label{eq:B3100}
\end{equation}
This growth is clearly linear and is fitted by the following
generating function:
\begin{equation}
    g_{\PNpPNm}(z)%
    =\frac{1+z}{(1-z)^2}.
    \label{eq:gfB3100}
\end{equation}
The linearity  of the growth \eqref{eq:B3100} could also be inferred
from Proposition \ref{rem:asydeg}, because of the presence of
the term $\left( 1-z \right)^{2}$ in the denominator of
\eqref{eq:gfB3100}.

As expected, the asymptotic behaviour in $\PTN$ and $\PNpPNm$
is the same.
The above findings imply that the $B2_{\left( 1;\,0;\,0 \right)}$ equation
is expected to be linearisable.

\subsubsection{$B3_{\left( 1;\,0;\,0 \right)}(\x;\y)$ system}

The equation $B3_{(\delta_1;\, \delta_2;\, \delta_3)}$ is given in \eqref{b3ddd}.
For the case of $(\delta_1,\delta_2,\delta_3)=(1,0,0)$, this equation exhibits a single growth pattern both in 
$\PTN$ and in $\PNpPNm$.  These two growth patterns coincide and are given
by \eqref{eq:B3100} with generating function \eqref{eq:gfB3100}.  This implies that the $B3_{\left( 1;\,0;\,0 \right)}$ equation
is expected to be linearisable.

\subsection{Type-C (individual)}

In the following the algebraic entropy for individual type-C equations in the lattice will be considered. However, similarly to the previous case for type-B equations,  some examples of type-C equations only have exponential growth when considered individually in the lattice.  Only the cases of type-C equations that have sub-exponential growth will be considered here, and the remaining type-C equations will be treated in Section \ref{sec:typeCpairs}.

In the following, the integers $n\;(\textrm{mod }2)$ are taken to be elements of $\{1,2\}$.

\subsubsection{Arrangement of equations in the lattice}

The arrangement of type-C equations in the lattice is shown in Figure \ref{fig:lattice-Csym}.  
The type-C equations centered at white vertices are
\begin{align}
\C{\ccy}{\cca}{\ccb}{\ccc}{\ccd}{\ccpp}{\ccqq}=0.
\end{align}
According to Figure \ref{fig:3fig4quad}, the type-C equations centered at black vertices of Figure \ref{fig:lattice-Csym} are rotated by $180^\circ$, which is given by
\begin{align}
\C{\ccx}{\ccyd}{\ccyc}{\ccyb}{\ccya}{\ccpp}{\ccqq}=0.
\end{align}
Because the same type-C equation is centered at both the black vertices and white vertices, these systems of equations will be referred to as $C(\x;\y)$ systems.

\begin{figure}[htb!]
\centering
\begin{tikzpicture}[scale=0.85]

\foreach \x in {-4,-2,...,3}{
 \draw[-] (\x,4)--(\x+1,3);
 \draw[-] (\x+2,4)--(\x+1,3);
 \draw[-,double,thick] (\x,2)--(\x+1,3);
 \draw[-,double,thick] (\x+2,2)--(\x+1,3);}
 
\foreach \x in {-2,0,...,3}{
 \draw[-] (\x,2)--(\x+1,1);
 \draw[-] (\x+2,2)--(\x+1,1);
 \draw[-,double,thick] (\x,0)--(\x+1,1);
 \draw[-,double,thick] (\x+2,0)--(\x+1,1);}
 
\foreach \x in {0,2,...,3}{
 \draw[-] (\x,0)--(\x+1,-1);
 \draw[-] (\x+2,0)--(\x+1,-1);
 \draw[-,double,thick] (\x,-2)--(\x+1,-1);
 \draw[-,double,thick] (\x+2,-2)--(\x+1,-1);} 
 
\foreach \x in {2}{
 \draw[-] (\x,-2)--(\x+1,-3);
 \draw[-] (\x+2,-2)--(\x+1,-3);
 \draw[-,double,thick] (\x,-4)--(\x+1,-3);
 \draw[-,double,thick] (\x+2,-4)--(\x+1,-3);} 
 
 \draw[-] (-2,2)--(-3,1);
 \draw[-] (0,0)--(-1,-1);
 \draw[-] (2,-2)--(1,-3);

\foreach \x in {-4,-2,...,4}{
 \foreach \y in {-4,-2,...,4}{
  \fill[black] (\x,\y) circle (3.0pt);}}

\foreach \x in {-3,-1,...,3}{
 \foreach \y in {-3,-1,...,3}{
  \filldraw[fill=white,draw=black] (\x,\y) circle (3.0pt);}}
  
\foreach \x in {-4,-2,...,4}{
\filldraw[fill=black,draw=black] (\x,-\x) \Square{3.0pt};}
\foreach \x in {-3,-1,...,3}{
\filldraw[fill=white,draw=black] (\x,-\x) \Square{3.0pt};}

\foreach \x in {-3,-1,...,3}{
\filldraw[fill=black,draw=black] (\x-1,-\x-1) \Square{3.0pt};
}%
\foreach \x in {-2,0,...,2}{
\filldraw[fill=white,draw=black] (\x-1,-\x-1) \Square{3.0pt};
}%

\filldraw[fill=white,draw=white] (-4.5,-4.5)--(1.0,-4.5)--(-4.5,1.0)--(-4.5,4.5);

\foreach \x in {-4,...,3}{
 \draw[-,dotted] (\x+0.5,-4.5)--(\x+0.5,4.5);
 \draw[-,dotted] (-4.5,\x+0.5)--(4.5,\x+0.5);}

\foreach \x in {-4,-2,...,3}{
 \fill (\x+0.5,-4.5) circle (0.01pt)
 node[below=2pt]{\small $\ccqa$};
 \fill (\x+1.5,-4.5) circle (0.01pt)
 node[below=2pt]{\small $\ccqb$};
 \fill (\x+0.5,4.5) circle (0.01pt)
 node[above=2pt]{\small $\ccqa$};
 \fill (\x+1.5,4.5) circle (0.01pt)
 node[above=2pt]{\small $\ccqb$};
 \fill (-4.5,\x+0.5) circle (0.01pt)
 node[left=2pt]{\small $\ccpa$};
 \fill (-4.5,\x+1.5) circle (0.01pt)
 node[left=2pt]{\small $\ccpb$};
 \fill (4.5,\x+0.5) circle (0.01pt)
 node[right=2pt]{\small $\ccpa$};
 \fill (4.5,\x+1.5) circle (0.01pt)
 node[right=2pt]{\small $\ccpb$};}

\end{tikzpicture}
\caption{Arrangement of type-C equations in the lattice.  Equations centered at black vertices are rotated $180^\circ$ relative to equations centered at white vertices.}
\label{fig:lattice-Csym}
\end{figure}
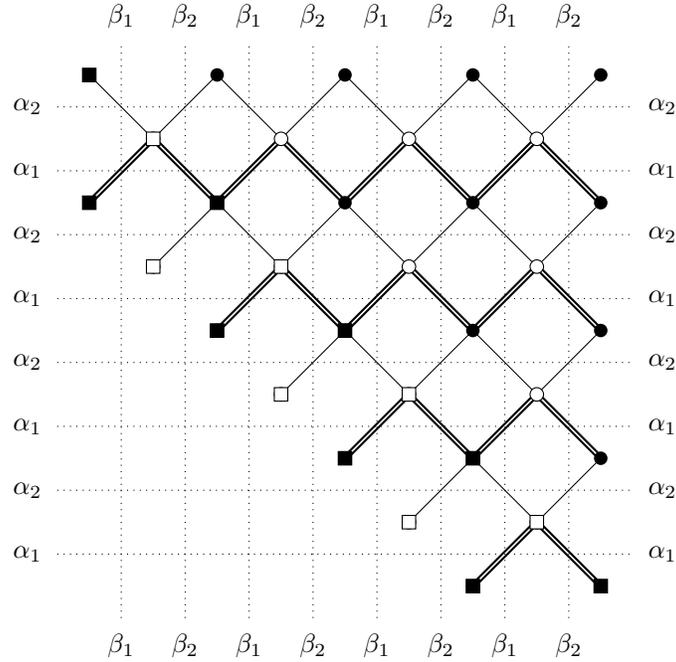

Type-C equations have the most varying behaviour in comparison to both
to type-A and type-B equations.
All type-C equations are found to possess more than one growth pattern, and these patterns may  differ greatly depending on the different 
directions of evolution.
Despite this behaviour, it is interesting to note that the maximal degree growth of type-C equations is isotropic.

\subsubsection{$C1_{(0)}(\x;\y)$ system}

The equation $C1_{(\delta_1)}$ is given in \eqref{c1d}.  
For the case of $\delta=0$, this equation exhibits two different  growth patterns in  both
$\PTN$ and in $\PNpPNm$, respectively.
The patterns for the different sets of initial conditions are different from each other.
In both cases the patterns appear in reverse order in the
evolution matrix in the SW and SE directions.

In $\PTN$ the two different degree patterns are given in \eqref{eq:C1x} and \eqref{eq:C1y}.  These two degree patterns are fitted by the two different generating functions defined by 
\begin{align}
g_{\PTN}^{(\ell)}(z)=\frac{h^{(\ell)}(z)}{(1-z^3)(1-z)^2}, \qquad \ell=1,2,
\end{align}
where
\begin{align}
h^{(\ell)}(z)=\left\{\begin{array}{rl}
z^4-z^3+2 z^2+z+1, & \ell=1, \\
(z^2-z+1) (z+1)^2, & \ell=2.
\end{array}\right.
\end{align}
By Proposition \ref{rem:asydeg}, the growth of degree patterns associated to each of $g_{\PTN}^{(\ell)}(z)$ ($\ell=1,2$) is quadratic.  The maximal growth coincides with the pattern \eqref{eq:C1x} fitted by $g_{\PTN}^{(\ell=1)}(z)$.

In $\PNpPNm$ the two different degree patterns are given in \eqref{eq:C1xkels} and \eqref{eq:C1ykels}.  These two degree patterns are fitted by the two different generating functions defined by 
\begin{align}
g_{\PNpPNm}^{(\ell)}(z)=\frac{h^{(\ell)}(z)}{(1-z^7)(1-z^3)(1-z)}, \qquad \ell=1,2,
\end{align}
where
\begin{align}
h^{(\ell)}(z)=\left\{\begin{array}{rl}
z^{14}-z^{11}+z^{10}+4 z^8+5 z^7+8 z^6+8 z^5+7 z^4+7 z^3+8 z^2+3 z+1, & \ell=1, \\
(z+1) (z^{11}-z^{10}+2 z^9+4 z^7+3 z^6+3 z^5+4 z^4+4 z^3+2 z^2+3 z+1), & \ell=2.
\end{array}\right.
\end{align}
By Proposition \ref{rem:asydeg}, the growth of degree patterns associated to each of $g_{\PNpPNm}^{(\ell)}(z)$ ($\ell=1,2$) is quadratic.  The maximal degree growth is given by \eqref{eq:C1maxkels}, and is fitted by the following generating function:
\begin{equation}
    G_{\PNpPNm}\left( z \right) =
    \frac{(z+1) (z^{13}-z^{12}+z^9+3 z^7+2 z^6+6 z^5+3 z^4+3 z^3+4 z^2+3 z+1)}{%
    (1-z^7)(1-z^3)(1-z)}.
    \label{eq:gfC1maxkels}
\end{equation}
Its asymptotic behaviour is quadratic.

As expected, the asymptotic behaviour in $\PTN$ and $\PNpPNm$
is the same.
The above findings imply that the $C1_{(0)}$ equation is expected to 
be integrable.

\subsubsection{$C2_{\left( 1;\,0;\,0 \right)}(\x;\y)$ system}

The equation $C2_{(\delta_1;\,\delta_2;\,\delta_3)}$ is given in \eqref{c2ddd}.
For the case of $(\delta_1,\delta_2,\delta_3)=(1,0,0)$, this equation exhibits two different growth patterns in both
$\PTN$ and in $\PNpPNm$.
The patterns for the different sets of initial conditions are different from each other.
In both cases the patterns appear in reverse order in the
evolution matrix in the SW and SE directions.

The two different degree patterns in $\PTN$ are given in \eqref{eq:C2100x} and \eqref{eq:C2100y}.  These two degree patterns are fitted by the two different generating functions defined by
\begin{align}%
g_{\PTN}^{(\ell)}(z)=\frac{h^{(\ell)}(z)}{(1-z^3)(1-z)^2}, \qquad \ell=1,2,
\end{align}
where
\begin{align}
h^{(\ell)}(z)=\left\{\begin{array}{rl}
-(z^2+1) (z^3-z^2-z-1), & \ell=1, \\
(z+1) (z^2+1), & \ell=2.
\end{array}\right.
\end{align}
By Proposition \ref{rem:asydeg}, the growth of degree patterns associated to each of $g_{\PTN}^{(\ell)}(z)$ ($\ell=1,2$) is quadratic.  The maximal growth coincides with the pattern \eqref{eq:C2100x} fitted by $g_{\PTN}^{(\ell=1)}(z)$.

The two different degree patterns in $\PNpPNm$ are given in \eqref{eq:C2100xkels} and \eqref{eq:C2100ykels}.  These two degree patterns fitted by the two different generating functions defined by
\begin{align}%
g_{\PNpPNm}^{(\ell)}(z)=\frac{h^{(\ell)}(z)}{(1+z^2)(1-z)^3}, \qquad \ell=1,2,
\end{align}
where
\begin{align}
h^{(\ell)}(z)=\left\{\begin{array}{rl}
z^{12}-2 z^{11}+2 z^{10}-2 z^9+z^8-2 z^5+3 z^4-2 z^3+4 z^2+z+1, & \ell=1, \\
z^{10}-2 z^9+2 z^8-2 z^7+z^6-z^4+3 z^3+2 z+1, & \ell=2.
\end{array}\right.
\end{align}

By Proposition \ref{rem:asydeg}, the growth of degree patterns associated to each of $g_{\PNpPNm}^{(\ell)}(z)$ ($\ell=1,2$) is quadratic.  The maximal degree growth is given by \eqref{eq:C2100maxkels}, and is fitted by the following generating function
\begin{equation}
    G_{\PNpPNm}\left( z \right) =
    \frac{z^{12}-2 z^{11}+2 z^{10}-2 z^9+z^8-z^6+z^5-z^4+2 z^3+z^2+2 z+1}{(1+z^2) (1-z)^3}.
\end{equation}
Its asymptotic behaviour is quadratic.

As expected, the asymptotic behaviour in $\PTN$ and $\PNpPNm$
is the same.
The above findings imply that the $C2_{\left( 1;\,0;\,0 \right)}$ equation 
is expected to be integrable.

\subsubsection{$C3_{\left( 0;\,0;\,0 \right)}(\x;\y)$ system}

The equation $C3_{(\delta_1;\,\delta_2;\,\delta_3)}$ is given in \eqref{c3ddd}.
For the case of $(\delta_1,\delta_2,\delta_3)=(0,0,0)$, this equation exhibits two different  growth patterns in both 
$\PTN$ and in $\PNpPNm$.
The patterns for the different sets of initial conditions are different from each other.
In both cases the patterns appear in reverse order in the
evolution matrix in the SW and SE directions.

The two different degree patterns in $\PTN$ are given in \eqref{eq:C3000x} and \eqref{eq:C3000y}. These two degree patterns  are fitted by the two different generating functions defined by
\begin{align}%
g_{\PTN}^{(\ell)}(z)=\frac{h^{(\ell)}(z)}{(1-z^3)(1-z)^2}, \qquad \ell=1,2,
\end{align}
where
\begin{align}
h^{(\ell)}(z)=\left\{\begin{array}{rl}
2 z^2+z+1, & \ell=1, \\
z^4+z^2+z+1, & \ell=2.
\end{array}\right.
\end{align}
By Proposition \ref{rem:asydeg}, the growth of degree patterns associated to each of $g_{\PTN}^{(\ell)}(z)$ ($\ell=1,2$) is quadratic.  The maximal growth coincides with the pattern \eqref{eq:C3000x} fitted by $g_{\PTN}^{(\ell=1)}(z)$.

The different degree patterns in $\PNpPNm$ are given in \eqref{eq:C3000xkels} and \eqref{eq:C3000ykels}. These two degree patterns  are  fitted by the two different generating functions defined by
\begin{align}%
g_{\PNpPNm}^{(\ell)}(z)=\frac{h^{(\ell)}(z)}{(1-z^4)(1-z^3)(1-z)}, \qquad \ell=1,2,
\end{align}
where
\begin{align}
h^{(\ell)}(z)=\left\{\begin{array}{rl}
-(z^8-5 z^5-6 z^4-8 z^3-8 z^2-3 z-1), & \ell=1, \\
z^7+z^6+4 z^5+7 z^4+7 z^3+5 z^2+4 z+1, & \ell=2.
\end{array}\right.
\end{align}
By Proposition \ref{rem:asydeg}, the growth of degree patterns associated to each of $g_{\PNpPNm}^{(\ell)}(z)$ ($\ell=1,2$) is quadratic.  The maximal degree growth is given by \eqref{eq:C3000maxkels}, and is fitted by the following generating function
\begin{equation}
    G_{\PNpPNm}\left( z \right) =
    \frac{z^9-z^6-5 z^5-5 z^4-8 z^3-7 z^2-4 z-1}{(z^4-1)(z^3-1)(z-1)}.
\end{equation}
Its asymptotic behaviour is quadratic.

As expected, the asymptotic behaviour in $\PTN$ and $\PNpPNm$
is the same.
The above findings imply that the $C3_{\left( 0;\,0;\,0 \right)}$ equation 
is expected to be integrable.

\subsubsection{$C3_{\left( 1;\,0;\,0 \right)}(\x;\y)$ system}

The equation $C3_{(\delta_1;\,\delta_2;\,\delta_3)}$ is given in \eqref{c3ddd}.
For the case of $(\delta_1,\delta_2,\delta_3)=(1,0,0)$, this equation exhibits two different  growth pattern both in 
$\PTN$ and in $\PNpPNm$.
The patterns for the different sets of initial conditions are different from each other.
In both cases the patterns appear in reverse order in the
evolution matrix in the SW and SE directions.

The different degree patterns in $\PTN$ are given in \eqref{eq:C3100x} and \eqref{eq:C3100y}. These two degree patterns  are  fitted by the two different generating functions defined by
\begin{align}%
g_{\PTN}^{(\ell)}(z)=\frac{h^{(\ell)}(z)}{(1-z^3)(1-z)^2}, \qquad \ell=1,2,
\end{align}
where
\begin{align}
h^{(\ell)}(z)=\left\{\begin{array}{rl}
-(z^5-z^3-2 z^2-z-1), & \ell=1, \\
2 z^2+z+1, & \ell=2.
\end{array}\right.
\end{align}
By Proposition \ref{rem:asydeg}, the growth of degree patterns associated to each of $g_{\PTN}^{(\ell)}(z)$ ($\ell=1,2$) is quadratic.  The maximal growth coincides with the pattern \eqref{eq:C3100x} fitted by $g_{\PTN}^{(\ell=1)}(z)$.

The different degree patterns in $\PNpPNm$  are given in \eqref{eq:C3100xkels} and \eqref{eq:C3100ykels}. These two degree patterns  are  fitted by the two different generating functions defined by
\begin{align}%
g_{\PNpPNm}^{(\ell)}(z)=\frac{h^{(\ell)}(z)}{(z^2+1) (z-1)^3}, \qquad \ell=1,2,
\end{align}
where
\begin{align}
h^{(\ell)}(z)=\left\{\begin{array}{rl}
z^8-2 z^7+2 z^6+z^5-3 z^4+2 z^3-4 z^2-z-1, & \ell=1, \\
z^6-2 z^5+3 z^4-4 z^3-2 z-1, & \ell=2.
\end{array}\right.
\end{align}
By Proposition \ref{rem:asydeg} the growth of degree patterns associated to each of $g_{\PNpPNm}^{(\ell)}(z)$ ($\ell=1,2$) is quadratic.  The maximal degree growth is given by \eqref{eq:C3100maxkels}, and is fitted by the following generating function
\begin{equation}
    G_{\PNpPNm}\left( z \right) =
    \frac{z^8-2 z^7+3 z^6-2 z^5+z^4-2 z^3-z^2-2 z-1}{(z^2+1) (z-1)^3}.
\end{equation}
Its asymptotic behaviour is quadratic.

As expected, the asymptotic behaviour in $\PTN$ and $\PNpPNm$
is the same.
The above findings imply that the $C3_{\left( 1;\,0;\,0 \right)}$ equation 
is expected to be integrable.

\subsection{Type-B (pairs)}\label{sec:typeBpairs}

There remains four type-B equations to be treated that give exponential degree growth in the lattice arrangement of Figure \ref{fig:lattice-Bsym}. These four type-B equations are $B2_{\left(1;\,0;\,1\right)}$, $B2_{\left(1;\,1;\,0\right)}$, $B3_{\left(\frac{1}{2};\,0;\,\frac{1}{2}\right)}$, and $B3_{\left(\frac{1}{2};\,\frac{1}{2};\,0\right)}$.  Using different type-B equations centered at black and white vertices respectively there can be found linear degree growths for certain pairs of type-B equations.  These pairs of type-B equations are investigated below.

\subsubsection{Arrangement of equations in the lattice}

This case involves a specific pair of type-B equations which will be denoted as $B$ and $\overline{B}$.  The two different type-B equations will be distinguished graphically by different orientations of directed edges, as shown in Figure \ref{fig:2typeB}.  Only the relative orientation matters, so if $B$ is associated to the left of Figure \ref{fig:2typeC}, then $\overline{B}$ is associated to the right of Figure \ref{fig:2typeB}, and vice versa.

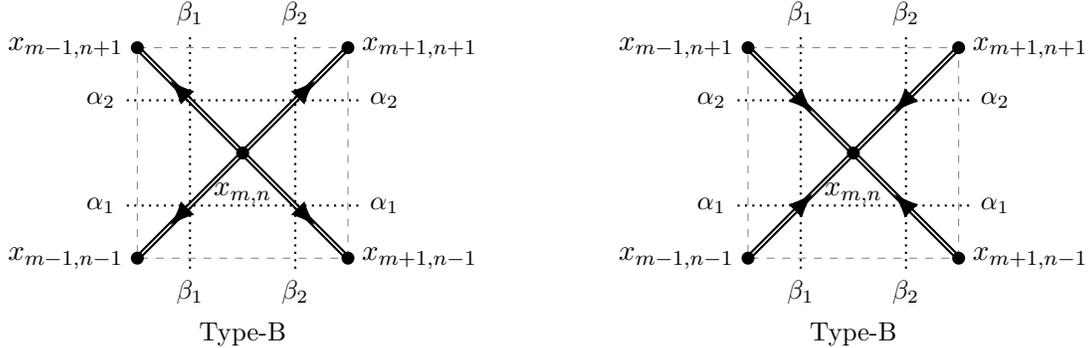
\begin{figure}[tbh]
\centering
\begin{tikzpicture}[scale=0.7]

\draw[-,gray,very thin,dashed] (5,-1)--(5,3)--(1,3)--(1,-1)--(5,-1);
\draw[->-,double,thick] (3,1)--(1,3);\draw[->-,double,thick] (3,1)--(5,3);
\draw[->-,double,thick] (3,1)--(1,-1);\draw[->-,double,thick] (3,1)--(5,-1);
\fill (0.8,-0.0) circle (0.1pt)
node[left=0.5pt]{\color{black}\small $\ccpa$};
\fill (5.2,-0.0) circle (0.1pt)
node[right=0.5pt]{\color{black}\small $\ccpa$};
\fill (4,-1.2) circle (0.1pt)
node[below=0.5pt]{\color{black}\small $\ccqb$};
\fill (4,3.2) circle (0.1pt)
node[above=0.5pt]{\color{black}\small $\ccqb$};
\fill (0.8,2.0) circle (0.1pt)
node[left=0.5pt]{\color{black}\small $\ccpb$};
\fill (5.2,2.0) circle (0.1pt)
node[right=0.5pt]{\color{black}\small $\ccpb$};
\fill (2,-1.2) circle (0.1pt)
node[below=0.5pt]{\color{black}\small $\ccqa$};
\fill (2,3.2) circle (0.1pt)
node[above=0.5pt]{\color{black}\small $\ccqa$};
\fill (3,1) circle (3.5pt)
node[below=7.5pt]{\color{black} $\ccx$};
\fill (1,-1) circle (3.5pt)
node[left=1.5pt]{\color{black} $\ccc$};
\fill (1,3) circle (3.5pt)
node[left=1.5pt]{\color{black} $\cca$};
\fill (5,3) circle (3.5pt)
node[right=1.5pt]{\color{black} $\ccb$};
\fill (5,-1) circle (3.5pt)
node[right=1.5pt]{\color{black} $\ccd$};

\draw[-,dotted,thick] (2,-1.2)--(2,3.2);\draw[-,dotted,thick] (4,-1.2)--(4,3.2);
\draw[-,dotted,thick] (0.8,-0.0)--(5.2,-0.0);\draw[-,dotted,thick] (0.8,2.0)--(5.2,2.0);

\fill (3,-2) circle (0.01pt)
node[below=0.5pt]{\color{black}\small Type-B};

\begin{scope}[xshift=330pt]

\draw[-,gray,very thin,dashed] (5,-1)--(5,3)--(1,3)--(1,-1)--(5,-1);
\draw[-,thick] (5,3)--(3,1)--(1,3);
\draw[->-,double,thick] (1,3)--(3-0.2,1+0.2);\draw[-,double,thick](3-0.2,1+0.2)--(3,1);
\draw[->-,double,thick] (5,3)--(3+0.2,1+0.2);\draw[-,double,thick](3+0.2,1+0.2)--(3,1);
\draw[->-,double,thick] (1,-1)--(3-0.2,1-0.2);\draw[-,double,thick](3-0.2,1-0.2)--(3,1);
\draw[->-,double,thick] (5,-1)--(3+0.2,1-0.2);\draw[-,double,thick](3+0.2,1-0.2)--(3,1);
\fill (0.8,-0.0) circle (0.1pt)
node[left=0.5pt]{\color{black}\small $\ccpa$};
\fill (5.2,-0.0) circle (0.1pt)
node[right=0.5pt]{\color{black}\small $\ccpa$};
\fill (4,-1.2) circle (0.1pt)
node[below=0.5pt]{\color{black}\small $\ccqb$};
\fill (4,3.2) circle (0.1pt)
node[above=0.5pt]{\color{black}\small $\ccqb$};
\fill (0.8,2.0) circle (0.1pt)
node[left=0.5pt]{\color{black}\small $\ccpb$};
\fill (5.2,2.0) circle (0.1pt)
node[right=0.5pt]{\color{black}\small $\ccpb$};
\fill (2,-1.2) circle (0.1pt)
node[below=0.5pt]{\color{black}\small $\ccqa$};
\fill (2,3.2) circle (0.1pt)
node[above=0.5pt]{\color{black}\small $\ccqa$};
\fill (3,1) circle (3.5pt)
node[below=7.5pt]{\color{black} $\ccx$};
\fill (1,-1) circle (3.5pt)
node[left=1.5pt]{\color{black} $\ccc$};
\fill (1,3) circle (3.5pt)
node[left=1.5pt]{\color{black} $\cca$};
\fill (5,3) circle (3.5pt)
node[right=1.5pt]{\color{black} $\ccb$};
\fill (5,-1) circle (3.5pt)
node[right=1.5pt]{\color{black} $\ccd$};

\draw[-,dotted,thick] (2,-1.2)--(2,3.2);\draw[-,dotted,thick] (4,-1.2)--(4,3.2);
\draw[-,dotted,thick] (0.8,-0.0)--(5.2,-0.0);\draw[-,dotted,thick] (0.8,2.0)--(5.2,2.0);

\fill (3,-2) circle (0.01pt)
node[below=0.5pt]{\color{black}\small Type-B};

\end{scope}

\end{tikzpicture}
\caption{Two different type-B equations which are distinguished graphically by the orientations of double-line edges.}
\label{fig:2typeB}
\end{figure}

The arrangement of the pair of type-B equations of Figure \ref{fig:2typeB} in the lattice is indicated in Figure \ref{fig:lattice-Basym}. The type-B equations centered at white vertices are
\begin{align}
\Bf{\ccy}{\cca}{\ccb}{\ccc}{\ccd}{\ccpp}{\ccqq}=0,
\end{align}
and the type-B equations centered at black vertices are 
\begin{align}
\B{\ccx}{\ccya}{\ccyb}{\ccyc}{\ccyd}{\hat{\ccpp}}{\hat{\ccqq}}=0.
\end{align}
Since the two different type-B equations $B$ and $\overline{B}$ are respectively centered at black and white vertices of Figure \ref{fig:lattice-Basym}, these systems of equations will be referred to as  $B(\x)+\overline{B}(\y)$ systems.

\begin{figure}[htb!]
\centering
\begin{tikzpicture}[scale=0.85]

\foreach \x in {-4,-2,...,3}{
 \draw[->-,double,thick] (\x,4)--(\x+1,3);
 \draw[->-,double,thick] (\x+2,4)--(\x+1,3);
 \draw[->-,double,thick] (\x,2)--(\x+1,3);
 \draw[->-,double,thick] (\x+2,2)--(\x+1,3);}
 
\foreach \x in {-2,0,...,3}{
 \draw[->-,double,thick] (\x,2)--(\x+1,1);
 \draw[->-,double,thick] (\x+2,2)--(\x+1,1);
 \draw[->-,double,thick] (\x,0)--(\x+1,1);
 \draw[->-,double,thick] (\x+2,0)--(\x+1,1);}
 
\foreach \x in {0,2,...,3}{
 \draw[->-,double,thick] (\x,0)--(\x+1,-1);
 \draw[->-,double,thick] (\x+2,0)--(\x+1,-1);
 \draw[->-,double,thick] (\x,-2)--(\x+1,-1);
 \draw[->-,double,thick] (\x+2,-2)--(\x+1,-1);} 
 
\foreach \x in {2}{
 \draw[->-,double,thick] (\x,-2)--(\x+1,-3);
 \draw[->-,double,thick] (\x+2,-2)--(\x+1,-3);
 \draw[->-,double,thick] (\x,-4)--(\x+1,-3);
 \draw[->-,double,thick] (\x+2,-4)--(\x+1,-3);} 
 
 \draw[->-,double,thick] (-2,2)--(-3,1);
 \draw[->-,double,thick] (0,0)--(-1,-1);
 \draw[->-,double,thick] (2,-2)--(1,-3);

\foreach \x in {-4,-2,...,4}{
 \foreach \y in {-4,-2,...,4}{
  \fill[black] (\x,\y) circle (3.0pt);}}

\foreach \x in {-3,-1,...,3}{
 \foreach \y in {-3,-1,...,3}{
  \filldraw[fill=white,draw=black] (\x,\y) circle (3.0pt);}}
  
\foreach \x in {-4,-2,...,4}{
\filldraw[fill=black,draw=black] (\x,-\x) \Square{3.0pt};}
\foreach \x in {-3,-1,...,3}{
\filldraw[fill=white,draw=black] (\x,-\x) \Square{3.0pt};}

\foreach \x in {-3,-1,...,3}{
\filldraw[fill=black,draw=black] (\x-1,-\x-1) \Square{3.0pt};
}%
\foreach \x in {-2,0,...,2}{
\filldraw[fill=white,draw=black] (\x-1,-\x-1) \Square{3.0pt};
}%

\filldraw[fill=white,draw=white] (-4.5,-4.5)--(1.0,-4.5)--(-4.5,1.0)--(-4.5,4.5);

\foreach \x in {-4,...,3}{
 \draw[-,dotted] (\x+0.5,-4.5)--(\x+0.5,4.5);
 \draw[-,dotted] (-4.5,\x+0.5)--(4.5,\x+0.5);}

\foreach \x in {-4,-2,...,3}{
 \fill (\x+0.5,-4.5) circle (0.01pt)
 node[below=2pt]{\small $\ccqa$};
 \fill (\x+1.5,-4.5) circle (0.01pt)
 node[below=2pt]{\small $\ccqb$};
 \fill (\x+0.5,4.5) circle (0.01pt)
 node[above=2pt]{\small $\ccqa$};
 \fill (\x+1.5,4.5) circle (0.01pt)
 node[above=2pt]{\small $\ccqb$};
 \fill (-4.5,\x+0.5) circle (0.01pt)
 node[left=2pt]{\small $\ccpa$};
 \fill (-4.5,\x+1.5) circle (0.01pt)
 node[left=2pt]{\small $\ccpb$};
 \fill (4.5,\x+0.5) circle (0.01pt)
 node[right=2pt]{\small $\ccpa$};
 \fill (4.5,\x+1.5) circle (0.01pt)
 node[right=2pt]{\small $\ccpb$};}

\end{tikzpicture}
\caption{Arrangement of a pair of type-B equations in the lattice.  This is equivalent to the type-B lattice arrangement of Figure \ref{fig:lattice-Bsym}, but with directed edges used to distinguish the two different type-B equations indicated in Figure \ref{fig:2typeB}.}
\label{fig:lattice-Basym}
\end{figure}
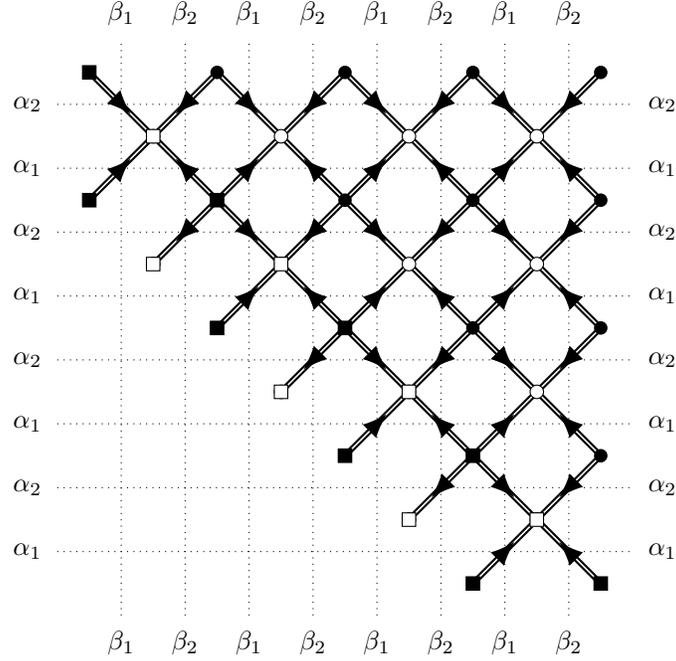

\subsubsection{$B2_{\left(1;\,0;\,1\right)}(\x)+B2_{\left(1;\,1;\,0\right)}(\y)$ system}
\label{ssub:B3101B3110}

The equation $B2_{(\delta_1;\,\delta_2;\,\delta_3)}$ is given in \eqref{b2ddd}.
The $B2_{\left(1;\,0;\,1\right)}(\x)+B2_{\left(1;\,1;\,0\right)}(\y)$ system exhibits two different  growth pattern both in 
$\PTN$ and in $\PNpPNm$.
The patterns for the different sets of initial conditions are different from each other.

In $\PTN$ the first growth pattern (of equations centered at $\x$ vertices)
is the same as given in \eqref{eq:B3100}.
Thus the growth is linear with generating function given by \eqref{eq:gfB3100}.
The second growth pattern (of equations centered at $\y$ vertices)
is instead:
\begin{equation}
    1, 3, 7, 11, 15, 19, 23, 27, 31, 35, 39, 43\dots.
    \label{eq:B3101B3110y}
\end{equation}
This growth is clearly linear and is fitted by the following
generating function:
\begin{equation}
    g_{\PTN}(z)%
    =\frac{2z^2+z+1}{(1-z)^2}.
    \label{eq:gfB3101B3110y}
\end{equation}
The linearity of the growth \eqref{eq:B3101B3110y} could also be confirmed
using Proposition \ref{rem:asydeg}.
The maximal growth coincides with the pattern \eqref{eq:B3101B3110y}.

In $\PNpPNm$ the first growth pattern (of equations centered in $\x$ vertices)
is:
\begin{equation}
    1, 4, 7, 10, 13, 16, 19, 22, 25, 28, 31, 34, 37\dots
    \label{eq:B3101B3110xkels}
\end{equation}
This growth is clearly linear and is fitted by the following
generating function:
\begin{equation}
    g_{\PNpPNm}(z)%
    =\frac{1+2z}{(1-z)^2}.
    \label{eq:gfB3101B3110xkels}
\end{equation}
The linearity  of the growth \eqref{eq:B3101B3110xkels} follows
from Proposition \ref{rem:asydeg}.
The second growth pattern (of equations centered in $\y$ vertices)
is instead:
\begin{equation}
    1, 4, 10, 16, 22, 28, 34, 40, 46, 52, 58, 64\dots
    \label{eq:B3101B3110ykels}
\end{equation}
This growth is clearly linear and is fitted by the following
generating function:
\begin{equation}
    g_{\PNpPNm}(z)%
    =\frac{3z^2+2z+1}{(1-z)^2}.
    \label{eq:gfB3101B3110ykels}
\end{equation}
The linearity  of the growth \eqref{eq:B3101B3110ykels} also follows 
from Proposition \ref{rem:asydeg}.
The maximal growth coincides with the pattern \eqref{eq:B3101B3110ykels}.

As expected, the asymptotic behaviour in $\PTN$ and $\PNpPNm$
is the same.
The above findings imply that the 
$B2_{\left(1;\,0;\,1\right)}(\x)+B2_{\left(1;\,1;\,0\right)}(\y)$ system
is expected to be linearisable.

\subsubsection{$B3_{\left(\frac{1}{2};\,0;\,\frac{1}{2}\right)}(\x)+B3_{\left(\frac{1}{2};\,\frac{1}{2};\,0\right)}(\y)$ system}

The equation $B3_{(\delta_1;\,\delta_2;\,\delta_3)}$ is given in \eqref{b3ddd}. The $B3_{\left(\frac{1}{2};\,0;\,\frac{1}{2}\right)}(\x)+B3_{\left(\frac{1}{2};\,\frac{1}{2};\,0\right)}(\y)$ system exhibits two different  growth pattern both in 
$\PTN$ and in $\PNpPNm$.
The patterns for the different sets of initial conditions are different from each other,
and coincide with those of the $B2_{\left(1;\,0;\,1\right)}(\x)+B2_{\left(1;\,1;\,0\right)}(\y)$ system.
Thus we refer back to subsection \ref{ssub:B3101B3110} for a complete description
of each of these patterns.
The above findings imply that the 
$B3_{\left(\frac{1}{2};\,0;\,\frac{1}{2}\right)}(\x)+B3_{\left(\frac{1}{2};\,\frac{1}{2};\,0\right)}(\y)$ system
is expected to be linearisable.

\subsection{Type-C (pairs)}\label{sec:typeCpairs}

There remains six type-C equations to be treated that give exponential degree growth in the lattice arrangement of Figure \ref{fig:lattice-Csym}. These six type-C equations are $C1_{(1)}$, $C2_{(0;\,0;\,0)}$, $C2_{\left(1;\,0;\,1\right)}$, $C2_{\left(1;\,1;\,0\right)}$, $C3_{\left(\frac{1}{2};\,0;\,\frac{1}{2}\right)}$, and $C3_{\left(\frac{1}{2};\,\frac{1}{2};\,0\right)}$.  Similarly to the type-B equations that were treated in Section \ref{sec:typeBpairs}, the remaining type-C equations can have a quadratic degree growth when combining certain pairs of the equations.  In fact, the equation $C1_{(1)}$ in \eqref{c1d} is a new equation that was found in this paper, motivated by the expectation there should be some equation to pair with $C2_{(0;\,0;\,0)}$ to achieve the quadratic degree growth.  The different pairs of type-C equations are investigated below.

In the following, the integers $n\;(\textrm{mod }4)$ are taken to be elements of $\{1,2,3,4\}$.

\subsubsection{Arrangement of equations in the lattice}

This case involves a specific pair of type-C equations which will be denoted as $C$ and $\overline{C}$.  The two different type-C equations are distinguished graphically by different orientations of directed edges, as shown in Figure \ref{fig:2typeC}.  Only the relative orientation matters, so if $C$ is associated to the left of Figure \ref{fig:2typeC}, then $\overline{C}$ is associated to the right of Figure \ref{fig:2typeC}, and vice versa.

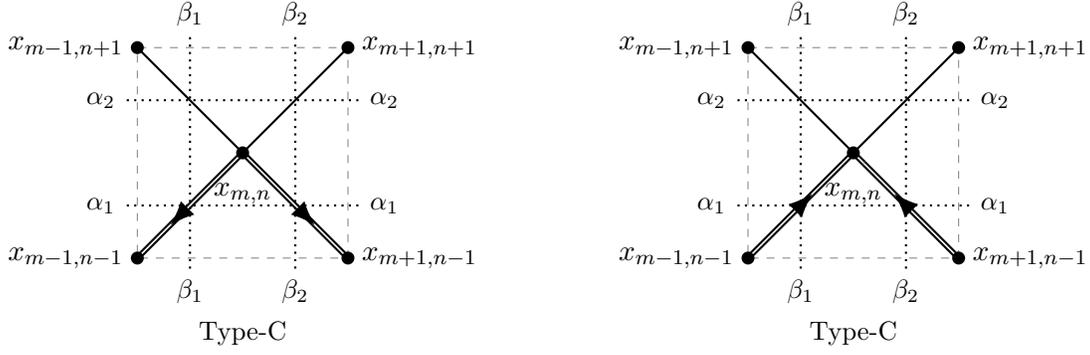
\begin{figure}[tbh]
\centering
\begin{tikzpicture}[scale=0.7]

\draw[-,gray,very thin,dashed] (5,-1)--(5,3)--(1,3)--(1,-1)--(5,-1);
\draw[-,thick] (5,3)--(3,1)--(1,3);
\draw[->-,double,thick] (3,1)--(1,-1);\draw[->-,double,thick] (3,1)--(5,-1);
\fill (0.8,-0.0) circle (0.1pt)
node[left=0.5pt]{\color{black}\small $\ccpa$};
\fill (5.2,-0.0) circle (0.1pt)
node[right=0.5pt]{\color{black}\small $\ccpa$};
\fill (4,-1.2) circle (0.1pt)
node[below=0.5pt]{\color{black}\small $\ccqb$};
\fill (4,3.2) circle (0.1pt)
node[above=0.5pt]{\color{black}\small $\ccqb$};
\fill (0.8,2.0) circle (0.1pt)
node[left=0.5pt]{\color{black}\small $\ccpb$};
\fill (5.2,2.0) circle (0.1pt)
node[right=0.5pt]{\color{black}\small $\ccpb$};
\fill (2,-1.2) circle (0.1pt)
node[below=0.5pt]{\color{black}\small $\ccqa$};
\fill (2,3.2) circle (0.1pt)
node[above=0.5pt]{\color{black}\small $\ccqa$};
\fill (3,1) circle (3.5pt)
node[below=7.5pt]{\color{black} $\ccx$};
\fill (1,-1) circle (3.5pt)
node[left=1.5pt]{\color{black} $\ccc$};
\fill (1,3) circle (3.5pt)
node[left=1.5pt]{\color{black} $\cca$};
\fill (5,3) circle (3.5pt)
node[right=1.5pt]{\color{black} $\ccb$};
\fill (5,-1) circle (3.5pt)
node[right=1.5pt]{\color{black} $\ccd$};

\draw[-,dotted,thick] (2,-1.2)--(2,3.2);\draw[-,dotted,thick] (4,-1.2)--(4,3.2);
\draw[-,dotted,thick] (0.8,-0.0)--(5.2,-0.0);\draw[-,dotted,thick] (0.8,2.0)--(5.2,2.0);

\fill (3,-2) circle (0.01pt)
node[below=0.5pt]{\color{black}\small Type-C};

\begin{scope}[xshift=330pt]

\draw[-,gray,very thin,dashed] (5,-1)--(5,3)--(1,3)--(1,-1)--(5,-1);
\draw[-,thick] (5,3)--(3,1)--(1,3);
\draw[->-,double,thick] (1,-1)--(3-0.2,1-0.2);\draw[-,double,thick](3-0.2,1-0.2)--(3,1);
\draw[->-,double,thick] (5,-1)--(3+0.2,1-0.2);\draw[-,double,thick](3+0.2,1-0.2)--(3,1);
\fill (0.8,-0.0) circle (0.1pt)
node[left=0.5pt]{\color{black}\small $\ccpa$};
\fill (5.2,-0.0) circle (0.1pt)
node[right=0.5pt]{\color{black}\small $\ccpa$};
\fill (4,-1.2) circle (0.1pt)
node[below=0.5pt]{\color{black}\small $\ccqb$};
\fill (4,3.2) circle (0.1pt)
node[above=0.5pt]{\color{black}\small $\ccqb$};
\fill (0.8,2.0) circle (0.1pt)
node[left=0.5pt]{\color{black}\small $\ccpb$};
\fill (5.2,2.0) circle (0.1pt)
node[right=0.5pt]{\color{black}\small $\ccpb$};
\fill (2,-1.2) circle (0.1pt)
node[below=0.5pt]{\color{black}\small $\ccqa$};
\fill (2,3.2) circle (0.1pt)
node[above=0.5pt]{\color{black}\small $\ccqa$};
\fill (3,1) circle (3.5pt)
node[below=7.5pt]{\color{black} $\ccx$};
\fill (1,-1) circle (3.5pt)
node[left=1.5pt]{\color{black} $\ccc$};
\fill (1,3) circle (3.5pt)
node[left=1.5pt]{\color{black} $\cca$};
\fill (5,3) circle (3.5pt)
node[right=1.5pt]{\color{black} $\ccb$};
\fill (5,-1) circle (3.5pt)
node[right=1.5pt]{\color{black} $\ccd$};

\draw[-,dotted,thick] (2,-1.2)--(2,3.2);\draw[-,dotted,thick] (4,-1.2)--(4,3.2);
\draw[-,dotted,thick] (0.8,-0.0)--(5.2,-0.0);\draw[-,dotted,thick] (0.8,2.0)--(5.2,2.0);

\fill (3,-2) circle (0.01pt)
node[below=0.5pt]{\color{black}\small Type-C};

\end{scope}

\end{tikzpicture}
\caption{Two different type-C equations which are distinguished graphically by the orientations of double-line edges.}
\label{fig:2typeC}
\end{figure}

The arrangement of the pair of type-C equations of Figure \ref{fig:2typeC} in the lattice is indicated in Figure \ref{fig:lattice-Casym}.  First, there are two types of white vertices in Figure \ref{fig:lattice-Casym} which are distinguished by the orientation of directed edges they are connected to.  The following two type-C equations are centered at the two respective types of white vertices
\begin{align}
\C{\ccy}{\cca}{\ccb}{\ccc}{\ccd}{\ccpp}{\ccqq}=0, \\
\Cf{\ccy}{\cca}{\ccb}{\ccc}{\ccd}{\ccpp}{\ccqq}=0.
\end{align}
Similarly, there are two types of black vertices in Figure \ref{fig:lattice-Casym} which are distinguished by the orientation of directed edges they are connected to.  The following two type-C equations are centered at the two respective types of black vertices
\begin{align}
\Cf{\ccx}{\ccyd}{\ccyc}{\ccyb}{\ccya}{\ccpp}{\ccqq}=0, \\
\C{\ccx}{\ccyd}{\ccyc}{\ccyb}{\ccya}{\ccpp}{\ccqq}=0.
\end{align}
The choice of which equation is centered at which type of black vertex, should be consistent with the orientation of directed edges from the choice made for equations centered at the white vertices.  Since the two different type-C equations $C$ and $\overline{C}$ each appear centered at both black and white vertices of Figure \ref{fig:lattice-Basym}, but with a different ordering, these systems of equations will be referred to as  $C(\x;\y)+\overline{C}(\y;\x)$ systems.

\begin{figure}[htb!]
\centering
\begin{tikzpicture}[scale=0.85]

\foreach \x in {-4,-2,...,3}{
 \draw[-] (\x,4)--(\x+1,3);
 \draw[-] (\x+2,4)--(\x+1,3);
 \draw[->-,double,thick] (\x,2)--(\x+1,3);
 \draw[->-,double,thick] (\x+2,2)--(\x+1,3);}
 
\foreach \x in {-2,0,...,3}{
 \draw[-] (\x,2)--(\x+1,1);
 \draw[-] (\x+2,2)--(\x+1,1);
 \draw[->-,double,thick] (\x+1,1)--(\x,0);
 \draw[->-,double,thick] (\x+1,1)--(\x+2,0);}
 
\foreach \x in {0,2,...,3}{
 \draw[-] (\x,0)--(\x+1,-1);
 \draw[-] (\x+2,0)--(\x+1,-1);
 \draw[->-,double,thick] (\x,-2)--(\x+1,-1);
 \draw[->-,double,thick] (\x+2,-2)--(\x+1,-1);} 
 
\foreach \x in {2}{
 \draw[-] (\x,-2)--(\x+1,-3);
 \draw[-] (\x+2,-2)--(\x+1,-3);
 \draw[->-,double,thick] (\x+1,-3)--(\x,-4);
 \draw[->-,double,thick] (\x+1,-3)--(\x+2,-4);} 
 
 \draw[-] (-2,2)--(-3,1);
 \draw[-] (0,0)--(-1,-1);
 \draw[-] (2,-2)--(1,-3);

\foreach \x in {-4,-2,...,4}{
 \foreach \y in {-4,-2,...,4}{
  \fill[black] (\x,\y) circle (3.0pt);}}

\foreach \x in {-3,-1,...,3}{
 \foreach \y in {-3,-1,...,3}{
  \filldraw[fill=white,draw=black] (\x,\y) circle (3.0pt);}}
  
\foreach \x in {-4,-2,...,4}{
\filldraw[fill=black,draw=black] (\x,-\x) \Square{3.0pt};}
\foreach \x in {-3,-1,...,3}{
\filldraw[fill=white,draw=black] (\x,-\x) \Square{3.0pt};}

\foreach \x in {-3,-1,...,3}{
\filldraw[fill=black,draw=black] (\x-1,-\x-1) \Square{3.0pt};
}%
\foreach \x in {-2,0,...,2}{
\filldraw[fill=white,draw=black] (\x-1,-\x-1) \Square{3.0pt};
}%

\filldraw[fill=white,draw=white] (-4.5,-4.5)--(1.0,-4.5)--(-4.5,1.0)--(-4.5,4.5);

\foreach \x in {-4,...,3}{
 \draw[-,dotted] (\x+0.5,-4.5)--(\x+0.5,4.5);
 \draw[-,dotted] (-4.5,\x+0.5)--(4.5,\x+0.5);}

\foreach \x in {-4,-2,...,3}{
 \fill (\x+0.5,-4.5) circle (0.01pt)
 node[below=2pt]{\small $\ccqa$};
 \fill (\x+1.5,-4.5) circle (0.01pt)
 node[below=2pt]{\small $\ccqb$};
 \fill (\x+0.5,4.5) circle (0.01pt)
 node[above=2pt]{\small $\ccqa$};
 \fill (\x+1.5,4.5) circle (0.01pt)
 node[above=2pt]{\small $\ccqb$};
 \fill (-4.5,\x+0.5) circle (0.01pt)
 node[left=2pt]{\small $\ccpa$};
 \fill (-4.5,\x+1.5) circle (0.01pt)
 node[left=2pt]{\small $\ccpb$};
 \fill (4.5,\x+0.5) circle (0.01pt)
 node[right=2pt]{\small $\ccpa$};
 \fill (4.5,\x+1.5) circle (0.01pt)
 node[right=2pt]{\small $\ccpb$};}

\end{tikzpicture}
\caption{Arrangement of a pair of type-C equations in the lattice.  This is equivalent to the type-C lattice arrangement of Figure \ref{fig:lattice-Csym}, but with directed edges used to distinguish the two different type-C equations indicated in Figure \ref{fig:2typeC}.}
\label{fig:lattice-Casym}
\end{figure}
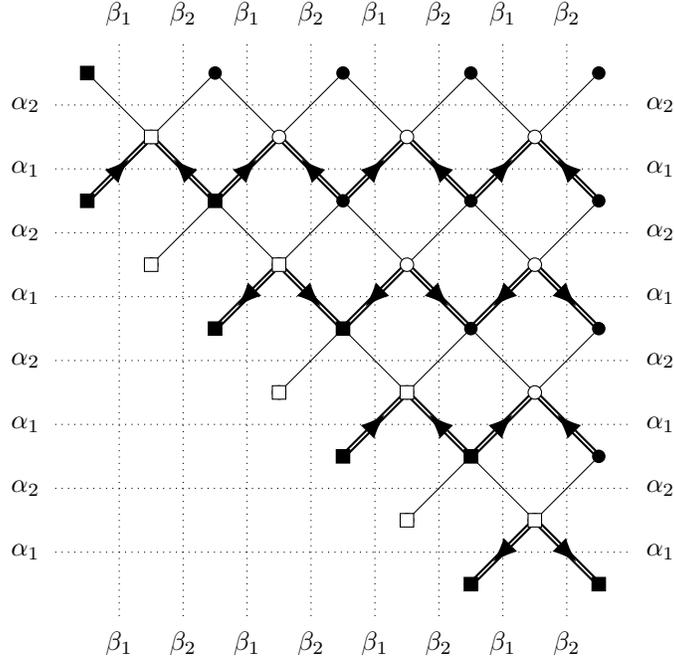

There is a general typical behaviour for the growth patterns for pairs of type-C equations in the different directions of the lattice.  It turns out that the generating functions for the growth patterns in the NE,SW,NW,SE directions may always be written as
\begin{align}\label{typeCdirections}
\begin{array}{llr}
NE:&g_{\mathcal{X}}^{(\ell,\ell)},& \qquad \ell=1,2,3,4, \\[0.2cm]
SW:&g_{\mathcal{X}}^{(\ell-1,\ell-1)},& \qquad \ell=1,2,3,4, \\[0.2cm]
NW:&g_{\mathcal{X}}^{(2-\ell,-\ell)},& \qquad \ell=1,2,3,4, \\[0.2cm]
SE:&g_{\mathcal{X}}^{(1-\ell,-1-\ell)},& \qquad \ell=1,2,3,4,
\end{array}
\end{align}
where $\mathcal{X}$ is either $\PTN$ or $\PNpPNm$, and $g_{\mathcal{X}}^{(\ell_1,\ell_2)}=zh_1^{(\ell_1)}(z^2)+h_2^{(\ell_2)}(z^2)$ is a generating function that is split into odd and even components, $h_1^{(\ell_1)}(z)$ and $h_2^{(\ell_2)}(z)$ respectively.  In the following the generating function $g_{\mathcal{X}}^{(\ell_1,\ell_2)}$ is given for each case, which was found by fitting the relevant degree growths listed in Appendix \ref{app:degrees}.

\subsubsection{$C1_{\left( 1 \right)}(\x;\y)+C2_{\left( 0;\,0;\,0 \right)}(\y;\x)$ system}

The equation $C1_{(\delta_1)}$ is given in \eqref{c1d}, 
and the equation $C2_{(\delta_1;\,\delta_2;\,\delta_3)}$ is given in \eqref{c2ddd}.  
The $C1_{\left( 1 \right)}(\x;\y)+C2_{\left( 0;\,0;\,0 \right)}(\y;\x)$ system 
of equations exhibits four different growth patterns both in 
$\PTN$ and in $\PNpPNm$.
The patterns for the different sets of initial conditions are different from each other.
The maximal growth is isotropic.

\paragraph{$\PTN$.}

The generating function $g_{\PTN}^{(\ell_1,\ell_2)}(z)$ is %
\begin{align}\label{c1plusc1000}
g_{\PTN}^{(\ell_1,\ell_2)}(z)=\frac{zh_1^{(\ell_1)}(z^2)+h_2^{(\ell_2)}(z^2)}{(1-z^6)(1-z^2)^2}, \qquad \ell_1,\ell_2=1,2,3,4 \;(\mathrm{mod}\,4),
\end{align}
where
\begin{align}
h_1^{(\ell)}(z)=\left\{\begin{array}{rl}
z^4+z^3+6z^2+5z+3, & \ell=1\;(\mathrm{mod}\,4), \\
z^4+3z^3+6z^2+3z+3, & \ell=2\;(\mathrm{mod}\,4), \\
2z^3+5z^2+6z+3, & \ell=3\;(\mathrm{mod}\,4), \\
z^4+2z^3+6z^2+4z+3, & \ell=4\;(\mathrm{mod}\,4),
\end{array}\right.
\end{align}
and
\begin{align}
h_2^{(\ell)}(z)=\left\{\begin{array}{rl}
z^4+5z^3+4z^2+5z+1, & \ell=1\;(\mathrm{mod}\,4), \\
3z^4+3z^3+6z^2+3z+1, & \ell=2\;(\mathrm{mod}\,4), \\
5z^3+5z^2+5z+1, & \ell=3\;(\mathrm{mod}\,4), \\
2z^4+4z^3+5z^2+4z+1, & \ell=4\;(\mathrm{mod}\,4).
\end{array}\right.
\end{align}

The four different degree patterns for each of the NE/SW and NW/SE directions 
are given in \eqref{eq:C11C2000NE1}--\eqref{eq:C11C2000NE4} and 
\eqref{eq:C11C2000NW1}--\eqref{eq:C11C2000NW4}, respectively.
The generating functions for the growth patterns in the NE, SW, NW, and SE 
directions are respectively given in terms of \eqref{c1plusc1000} as indicated in \eqref{typeCdirections}. 
By Proposition \ref{rem:asydeg} the growth of degree patterns associated 
to each of these generating functions is quadratic.

The maximal growth with initial conditions in $\PTN$ is
isotropic and given by the quadratic degree pattern \eqref{eq:C11C2000NE3} fitted by $g_{\PTN}^{(3,3)}(z)$.

\paragraph{$\PNpPNm$.}

The generating function $g_{\PNpPNm}^{(\ell_1,\ell_2)}(z)$ is %
\begin{align}\label{c1plusc1000kels}
g_{\PNpPNm}^{(\ell_1,\ell_2)}(z)=\frac{zh_1^{(\ell_1)}(z^2)+h_2^{(\ell_2)}(z^2)}{(1-z^{22})(1-z^6)(1-z^2)}, \qquad \ell_1,\ell_2=1,2,3,4 \;(\mathrm{mod}\,4),
\end{align}
where
\begin{align}
h_1^{(\ell)}(z)=\left\{\begin{array}{rl}
z^{16}+z^{14}+3z^{13}+15z^{12}+24z^{11}+32z^{10}+28z^9+32z^8 & \\
+28z^7+31z^6+28z^5+30z^4+29z^3+26z^2+16z+4, & \ell=1\;(\mathrm{mod}\,4), \\[0.2cm]
z^{15}+ z^{14} + 7 z^{13} + 16 z^{12}+ 27 z^{11}+ 28 z^{10} + 32 z^9  + 28 z^8  & \\
+ 31 z^7  + 29 z^6  + 30 z^5  + 29 z^4  + 28 z^3  + 24 z^2  + 12 z + 5, & \ell=2\;(\mathrm{mod}\,4), \\[0.2cm]
-z^{14}   + 4 z^{13}   + 13 z^{12}   + 25 z^{11}   + 31 z^{10}   + 28 z^9  + 32 z^8  & \\  
 + 28 z^7 + 32 z^6  + 28 z^5  + 31 z^4  + 30 z^3  + 26 z^2 + 17 z + 4, & \ell=3\;(\mathrm{mod}\,4), \\[0.2cm]
z^{14}   + 5 z^{13}   + 16 z^{12}   + 26 z^{11}   + 28 z^{10}   + 32 z ^9 + 28 z^8 & \\
+ 32 z^7  + 28 z^6  + 31 z^5  + 29 z^4  + 29 z^3  + 25 z^2  + 13 z + 5, & \ell=4\;(\mathrm{mod}\,4),
\end{array}\right.
\end{align}
and
\begin{align}
h_2^{(\ell)}(z)=\left\{\begin{array}{rl}
z^{16}+z^{14}+10z^{13}+20z^{12}+28z^{11}+30z^{10}+30z^9+29z^8 & \\
+31z^7+28z^6+31z^5+28z^4+31z^3+18z^2+11z+1, & \ell=1\;(\mathrm{mod}\,4),  \\[0.2cm]
z^{15}+4z^{14}+12z^{13}+20z^{12}+29z^{11}+30z^{10}+29z^9+31z^8 & \\
+28z^7+32z^6+28z^5+31z^4+24z^3+19z^2+9z+1, & \ell=2\;(\mathrm{mod}\,4), \\[0.2cm]
-z^{15}+9z^{13}+21z^{12}+27z^{11}+31z^{10}+29z^9+30z^8 & \\
+30z^7+29z^6+31z^5+29z^4+32z^3+19z^2+11z+1, & \ell=3\;(\mathrm{mod}\,4), \\[0.2cm]
2z^{14}+12z^{13}+19z^{12}+30z^{11}+29z^{10}+30z^9+30z^8 & \\
+29z^7+31z^6+28z^5+32z^4+26z^3+20z^2+9z+1, & \ell=4\;(\mathrm{mod}\,4).
\end{array}\right.
\end{align}

The four different degree patterns for each of the NE/SW and NW/SE directions are given in \eqref{eq:C11C2000NEkels1}--\eqref{eq:C11C2000NEkels4} and \eqref{eq:C11C2000NWkels1}--\eqref{eq:C11C2000NWkels4}, resepectively.  The generating functions for the growth patterns in the NE,SW,NW,SE directions are respectively given in terms of \eqref{c1plusc1000kels}  as indicated in \eqref{typeCdirections}.  By Proposition \ref{rem:asydeg} the growth of degree patterns associated to each of these generating functions is quadratic.

The maximal growth with initial conditions in $\PNpPNm$ is
isotropic and given by \eqref{eq:C11C2000NWkelsmax}.  
The maximal degree growth is fitted by the following generating function:
\begin{multline}
g_{\PNpPNm}(z)%
= \\
\frac{z^{16}-z^{13}-4 z^{12}-5 z^{11}-8 z^{10}-8 z^9-7 z^8-7 z^7-8 z^6-9 z^5-6 z^4-8 z^3-7 z^2-4 z-1}{%
        (z^{11}-1)(z^3-1)(z-1)}.
    \label{eq:gfC11C200NWkelsmax}
\end{multline}
By Proposition \ref{rem:asydeg} the growth of degree patterns associated 
to $G_{\PNpPNm}(z)$ is quadratic.

As expected, the asymptotic behaviour in $\PTN$ and $\PNpPNm$
is the same.
The above findings imply that the 
$C1_{\left(1\right)}(\x;\y)+C2_{\left(0;\,0;\,0\right)}(\y;\x)$ system
is expected to be integrable.

\subsubsection{$C2_{\left(1;\,0;\,1\right)}(\x;\y)+C2_{\left(1;\,1;\,0\right)}(\y;\x)$ system}\label{sss:C2101C2110}

The equation $C2_{(\delta_1;\,\delta_2;\,\delta_3)}$ is given in \eqref{c2ddd}.  
The $C2_{\left(1;\,0;\,1\right)}(\x;\y)+C2_{\left(1;\,1;\,0\right)}(\y;\x)$ 
system of equations exhibits four different growth patterns both in 
$\PTN$ and in $\PNpPNm$, respectively.
The patterns for the different sets of initial conditions are different from each other.
The maximal growth is isotropic.

\paragraph{$\PTN$.}

The generating function $g_{\PTN}^{(\ell_1,\ell_2)}(z)$ is %
\begin{align}\label{c101plusc110}
g_{\PTN}^{(\ell_1,\ell_2)}(z)=\frac{zh_1^{(\ell_1)}(z^2)+h_2^{(\ell_2)}(z^2)}{(z^6-1)(z^2-1)^2}, \qquad \ell_1,\ell_2=1,2,3,4 \;(\mathrm{mod}\,4),
\end{align}
where
\begin{align}
h_1^{(\ell)}(z)=\left\{\begin{array}{rl}
3 z^5  + z^4  - 2 z^3  - 8 z^2  - 7 z - 3, & \ell=1\;(\mathrm{mod}\,4), \\
z^5  + 2 z^4  - 2 z^3  - 7 z^2  - 7 z - 3, & \ell=2\;(\mathrm{mod}\,4), \\
2 z^4  - 3 z^3  - 5 z^2  - 7 z - 3, & \ell=3\;(\mathrm{mod}\,4), \\
-2 z^3  - 5 z^2  - 6 z - 3, & \ell=4\;(\mathrm{mod}\,4),
\end{array}\right.
\end{align}
and
\begin{align}
h_2^{(\ell)}(z)=\left\{\begin{array}{rl}
z^6  + 3 z^5  - 6 z^3  - 8 z^2  - 5 z - 1, & \ell=1\;(\mathrm{mod}\,4), \\
3 z^5  - z^4  - 4 z^3  - 8 z^2  - 5 z - 1, & \ell=2\;(\mathrm{mod}\,4), \\
z^5  - 4 z^3  - 7 z^2  - 5 z - 1, & \ell=3\;(\mathrm{mod}\,4), \\
-5 z^3  - 5 z^2  - 5 z - 1, & \ell=4\;(\mathrm{mod}\,4).
\end{array}\right.
\end{align}

The four different degree patterns for each of the NE/SW and NW/SE directions are given in \eqref{eq:C2101C2110NE1}--\eqref{eq:C2101C2110NE4} and \eqref{eq:C2101C2110NW1}--\eqref{eq:C2101C2110NW4}, resepectively.  The generating functions for the growth patterns in the NE,SW,NW,SE directions are respectively given in terms of \eqref{c101plusc110}   as indicated in \eqref{typeCdirections}.  By Proposition \ref{rem:asydeg} the growth of degree patterns associated to each of these generating functions is quadratic.

The maximal growth with initial conditions in $\PTN$ is
isotropic and given by the quadratic degree pattern \eqref{eq:C2101C2110NE1} fitted by $g_{\PTN}^{(1,1)}(z)$.%

\paragraph{$\PNpPNm$.}

The generating function $g_{\PNpPNm}^{(\ell_1,\ell_2)}(z)$ is %
\begin{align}\label{c101plusc1102}
g_{\PNpPNm}^{(\ell_1,\ell_2)}(z)=\frac{zh_1^{(\ell_1)}(z^2)+h_2^{(\ell_2)}(z^2)}{(1-z^{8})(1-z^6)(1-z^2)}, \qquad \ell_1,\ell_2=1,2,3,4 \;(\mathrm{mod}\,4),
\end{align}
where
\begin{align}
h_1^{(\ell)}(z)=\left\{\begin{array}{rl}
z^{12}   - z^9  - 7 z^8  - 7 z^7  - 3 z^6  + 14 z^5  & \\
+ 31 z^4  + 39 z^3  + 34 z^2  + 18 z + 5, & \ell=1\;(\mathrm{mod}\,4), \\[0.2cm]
z^{12}   - z^9  - 3 z^8  - 6 z^7  - z^6  + 13 z^5 & \\
+ 29 z^4  + 37 z^3  + 32 z^2  + 18 z + 5, & \ell=2\;(\mathrm{mod}\,4), \\[0.2cm]
z^{11}   - z^9  - z^8  - 5 z^7  + 3 z^6  + 13 z^5  & \\
+ 27 z^4  + 36 z^3  + 28 z^2  + 19 z + 4, & \ell=3\;(\mathrm{mod}\,4), \\[0.2cm]
z^{10}   - z^8 - 2 z^7  + 3 z^6  + 16 z^5   & \\
+ 27 z^4  + 33 z^3  + 26 z^2  + 16 z + 5, & \ell=4\;(\mathrm{mod}\,4),
\end{array}\right.
\end{align}
and
\begin{align}
h_2^{(\ell)}(z)=\left\{\begin{array}{rl}
z^{13}   - z^{10}   - 3 z^9  - 7 z^8  - 7 z^7  + 5 z^6  & \\
+ 22 z^5  + 37 z^4  + 39 z^3  + 26 z^2  + 11 z + 1, & \ell=1\;(\mathrm{mod}\,4),  \\[0.2cm]
z^{12}   - z^9  - 6 z^8  - 5 z^7  + 5 z^6  + 21 z^5 & \\
 + 36 z^4  + 35 z^3  + 26 z^2  + 11 z + 1, & \ell=2\;(\mathrm{mod}\,4), \\[0.2cm]
z^{12}   - 2 z^9  - 3 z^8  - 2 z^7  + 8 z^6  + 21 z^5 &\\
 + 32 z^4  + 33 z^3  + 24 z^2  + 11 z + 1, & \ell=3\;(\mathrm{mod}\,4), \\[0.2cm]
z^{10}   - z^9  - 2 z^7  + 10 z^6  + 21 z^5 &\\
  + 31 z^4  + 32 z^3  + 20 z^2  + 11 z + 1, & \ell=4\;(\mathrm{mod}\,4).
\end{array}\right.
\end{align}

The four different degree patterns for each of the NE/SW and NW/SE directions are given in \eqref{eq:C2101C2110NEkels1}--\eqref{eq:C2101C2110NEkels4} and \eqref{eq:C2101C2110NWkels1}--\eqref{eq:C2101C2110NWkels4}, resepectively.  The generating functions for the growth patterns in the NE,SW,NW,SE directions are respectively given in terms of \eqref{c101plusc1102}  as indicated in \eqref{typeCdirections}.  By Proposition \ref{rem:asydeg} the growth of degree patterns associated to each of these generating functions is quadratic.

The maximal growth with initial conditions in $\PNpPNm$ is
isotropic and given by the quadratic degree pattern \eqref{eq:C2101C2110NEkels1} fitted by $g_{\PNpPNm}^{(1,1)}(z)$.%

As expected, the asymptotic behaviour in $\PTN$ and $\PNpPNm$
is the same.
The above findings imply that the 
$C2_{\left(1;\,0;\,1\right)}(\x;\y)+C2_{\left(1;\,1;\,0\right)}(\y;\x)$ system
is expected to be integrable.

\subsubsection{$C3_{\left(\frac{1}{2};\,0;\,\frac{1}{2}\right)}(\x;\y)+C3_{\left(\frac{1}{2};\,\frac{1}{2};\,0\right)}(\y;\x)$ system}

The equation $C3_{(\delta_1;\,\delta_2;\,\delta_3)}$ is given in \eqref{c3ddd}.
The $C3_{\left(\frac{1}{2};\,0;\,\frac{1}{2}\right)}(\x;\y)+C3_{\left(\frac{1}{2};\,\frac{1}{2};\,0\right)}(\y;\x)$ system of equations exhibits four different growth patterns both in 
$\PTN$ and in $\PNpPNm$, respectively.
The patterns for the different sets of initial conditions are different from each other.
The maximal growth is isotropic.

\paragraph{$\PTN$.}

In $\PTN$ the growth patterns exactly coincide with those found for the
$C2_{\left(1;\,0;\,1\right)}(\x;\y)+C3_{\left(1;\,1;\,0\right)}(\y;\x)$ 
system given in Section \ref{sss:C2101C2110}.

\paragraph{$\PNpPNm$.}

The generating function $g_{\PNpPNm}^{(\ell_1,\ell_2)}(z)$ is %
\begin{align}\label{c101plusc1103}
g_{\PNpPNm}^{(\ell_1,\ell_2)}(z)=\frac{zh_1^{(\ell_1)}(z^2)+h_2^{(\ell_2)}(z^2)}{(z^{8}-1)(z^6-1)(z^2-1)}, \qquad \ell_1,\ell_2=1,2,3,4 \;(\mathrm{mod}\,4),
\end{align}
where
\begin{align}
h_1^{(\ell)}(z)=\left\{\begin{array}{rl}
7 z^8  + 8 z^7  + 3 z^6  - 14 z^5  - 32 z^4  - 39 z^3  - 34 z^2  - 18 z - 5, & \ell=1\;(\mathrm{mod}\,4), \\%
3 z^8  + 6 z^7  + 2 z^6  - 13 z^5  - 29 z^4  - 38 z^3  - 32 z^2  - 18 z - 5, & \ell=2\;(\mathrm{mod}\,4), \\%
z^9  + 5 z^7  - 3 z^6  - 12 z^5  - 27 z^4  - 36 z^3  - 29 z^2  - 19 z - 4, & \ell=3\;(\mathrm{mod}\,4), \\%
z^8  + z^7  - 3 z^6  - 15 z^5  - 27 z^4  - 33 z^3  - 27 z^2  - 16 z - 5, & \ell=4\;(\mathrm{mod}\,4),
\end{array}\right.
\end{align}
and
\begin{align}
h_2^{(\ell)}(z)=\left\{\begin{array}{rl}
3 z^9  + 8 z^8  + 7 z^7  - 5 z^6  - 23 z^5  - 37 z^4  - 39 z^3  - 26 z^2  - 11 z - 1, & \ell=1\;(\mathrm{mod}\,4),  \\
6 z^8 + 5 z^7  - 4 z^6  - 21 z^5  - 36 z^4  - 36 z^3  - 26 z^2  - 11 z - 1, & \ell=2\;(\mathrm{mod}\,4), \\%
z^9  + 3 z^8  + 2 z^7  - 7 z^6  - 21 z^5  - 32 z^4  - 34 z^3  - 24 z^2 - 11 z - 1, & \ell=3\;(\mathrm{mod}\,4), \\
z^9  + z^7  - 10 z^6  - 20 z^5  - 31 z^4  - 32 z^3  - 21 z^2  - 11 z - 1, & \ell=4\;(\mathrm{mod}\,4).
\end{array}\right.
\end{align}

The four different degree patterns for each of the NE/SW and NW/SE directions are given in \eqref{eq:C3101C3110NEkels1}--\eqref{eq:C3101C3110NEkels4} and \eqref{eq:C3101C3110NWkels1}--\eqref{eq:C3101C3110NWkels4}, resepectively.  The generating functions for the growth patterns in the NE,SW,NW,SE directions are respectively given in terms of \eqref{c101plusc1103}   as indicated in \eqref{typeCdirections}.  By Proposition \ref{rem:asydeg} the growth of degree patterns associated to each of these generating functions is quadratic.

The maximal growth with initial conditions in $\PNpPNm$ is
isotropic and  is given by the quadratic degree pattern \eqref{eq:C3101C3110NEkels1} fitted by $g_{\PNpPNm}^{(1,1)}(z)$.%

As expected, the asymptotic behaviour in $\PTN$ and $\PNpPNm$
is the same.
The above findings imply that the 
$C3_{\left(\frac{1}{2};\,0;\,\frac{1}{2}\right)}(\x;\y)+C3_{\left(\frac{1}{2};\,\frac{1}{2};\,0\right)}(\y;\x)$ system
is expected to be integrable.

\section{Summary and outlook}
\label{sec:concl}

In this paper we introduced the concept of algebraic entropy for
face-centered quad equations, which we have developed in analogy with
the concept of algebraic
entropy for regular quad equations
\cite{Tremblay2001,Viallet2006,Viallet2009}.
A key step was the identification of standard sets of initial conditions,
which we call the fundamental double staircases, which allows
us to build the sequence of iterates.
In particular, we use this concept to analyse the growth of
each of the known equations that satisfies the property of
consistency-around-a-face-centered-cube (CAFCC) \cite{Kels2020cafcc}.
The arrangement of CAFCC equations in the $\Z^{2}$ lattice is non-trivial
as they are non-autonomous equations.
For all cases the appropriate arrangements in the lattice were found
which gave
a quadratic growth of degrees for type-A and type-C equations,
and a linear growth of degrees for type-B equations.
Degree sequences where computed for two different sets of initial conditions,
namely the $\mathcal{X}_{1}$ and the $\mathcal{X}_{2}$ spaces.
Analysing the obtained sequences of degrees, we note that, in general, 
using initial conditions in  the space $\mathcal{X}_1$ the obtained sequences 
are less noisy.
This suggests that using initial conditions in the space $\mathcal{X}_{1}$
is more suitable from the computational point of view.

CAFCC equations are multidimensionally consistent, and this property
was used to obtain the Lax pairs of the equations \cite{Kels2020lax}.
Multidimensional consistency is a property usually associated to
Bianchi-like identities and B\"acklund transformations.
It is well known that Lax pairs and B\"acklund transformations are
associated to both linearisable and integrable equations.
While in the integrable case the solution of the Lax pair provides
genuine nontrivial solutions \cite{CalogeroDeGasperisIST_I}, in the
linearisable
case the Lax pair is fake
\cite{CalogeroNucci1991,HayButler2013,HayButler2015,GSL_Gallipoli15}.
Our growth analysis implies that type-A and type-C equations must be
genuinely integrable due to the quadratic growth.
On the other hand, since type-B equations have linear growth, this
suggests that these equations are linearisable and that their associated
Lax pairs are fake.
This is analogous to the results presented in \cite{GSL_general} about the
trapezoidal $H4$ and the $H6$ equations \cite{Boll2011,ABS2009}.
The linear behaviour of the trapezoidal $H4$ and the $H6$
equations was explained in \cite{GSY_DarbouxI,GSY_DarbouxII} through the
concept of Darboux integrability for quad equations
\cite{AdlerStartsev1999}.
We conjecture that type-B equations satisfy an analogous notion
for face-centered quad equations.
Such a notion has not yet been discussed, but it is an stimulating open
issue for future research.

In this paper only the minimal number of different equations (either
individual equations or pairs) were
used to achieve the vanishing algebraic entropy.
However, there are also more complicated arrangements of equations
which are also expected to lead to polynomial degree growth of the
equations.
For example, the type-C equations can be used as a link between type-A and
type-B equations (similarly to their use in CAFCC) to construct a system
consisting of each of type-A, -B, and -C equations in the lattice. This
arrangement of equations could also be reasonably expected
to pass the algebraic entropy test and it would be interesting to
explore further.
This could be considered as an analogue for the case of algebraic entropy
for non-standard
lattice arrangements of $H4$ and $H6$ quad equations
\cite{HietarintaViallet2012}.

\section*{Acknowledgements}

The authors thank Prof.~C-M.~Viallet for reading the manuscript
and providing helpful comments.  
We also thank the anonymous referee for their helpful comments in correcting
some errors in the first version of the paper.

GG has been supported by the Australian
Research Council through grant DP200100210 (Prof.~N.~Joshi and A/Prof.~M.~Radnovi\v{c})
and by Fondo Sociale Europeo del Friuli Venezia Giulia,
Programma operativo regionale 2014--2020 FP195673001 (A/Prof. T. Grava
and A/Prof. D. Guzzetti).

\begin{appendices}
\numberwithin{equation}{section}

\section{CAFCC equations}\label{app:equations}

\subsection*{\texorpdfstring{$A3_{(\delta)}$}{A3(delta)}, \texorpdfstring{$B3_{(\delta_1;\,\delta_2;\,\delta_3)}$}{B3(delta1;delta2;delta3)}, \texorpdfstring{$C3_{(\delta_1;\,\delta_2;\,\delta_3)}$}{C3(delta1;delta2;delta3)}}\label{app:b3}

\begin{multline}\label{a3d}
A3_{(\delta)}(x;x_a,x_b,x_c,x_d;\al,\bt)= \\ 
\begin{aligned}
& x\bigl((\tfrac{\beta_1}{\beta_2}-\tfrac{\beta_2}{\beta_1})(x_ax_b-x_cx_d) + (\tfrac{\alpha_1}{\alpha_2}-\tfrac{\alpha_2}{\alpha_1})(x_ax_c-x_bx_d) - (\tfrac{\alpha_1\alpha_2}{\beta_1\beta_2}-\tfrac{\beta_1\beta_2}{\alpha_1\alpha_2})(x_ax_d-x_bx_c)\bigr)  \\
&+(\tfrac{\alpha_2}{\beta_1}-\tfrac{\beta_1}{\alpha_2})(x_a x^2 - x_bx_cx_d) - (\tfrac{\alpha_2}{\beta_2}-\tfrac{\beta_2}{\alpha_2})(x_b x^2 - x_ax_cx_d)   
 - (\tfrac{\alpha_1}{\beta_1}-\tfrac{\beta_1}{\alpha_1})(x_c x^2 - x_ax_bx_d)   \\
 &+ (\tfrac{\alpha_1}{\beta_2}-\tfrac{\beta_2}{\alpha_1})(x_d x^2 - x_ax_bx_c) - \delta(\tfrac{\alpha_1}{\alpha_2}-\tfrac{\alpha_2}{\alpha_1})(\tfrac{\beta_1}{\beta_2}-\tfrac{\beta_2}{\beta_1})\bigl(\tfrac{\alpha_1\alpha_2}{\beta_1\beta_2}-\tfrac{\beta_1\beta_2}{\alpha_1\alpha_2}\bigr)x \\
& + \delta\Bigl((\tfrac{\alpha_1}{\beta_1}-\tfrac{\beta_1}{\alpha_1})(\tfrac{\alpha_2}{\beta_2}-\tfrac{\beta_2}{\alpha_2})\bigl((\tfrac{\alpha_1}{\beta_2}-\tfrac{\beta_2}{\alpha_1})x_a + (\tfrac{\alpha_2}{\beta_1}-\tfrac{\beta_1}{\alpha_2})x_d\bigr)
 \\
& \phantom{+\delta\Bigl(} - (\tfrac{\alpha_1}{\beta_2}-\tfrac{\beta_2}{\alpha_1})(\tfrac{\alpha_2}{\beta_1}-\tfrac{\beta_1}{\alpha_2})\bigl((\tfrac{\alpha_1}{\beta_1}-\tfrac{\beta_1}{\alpha_1})x_b + (\tfrac{\alpha_2}{\beta_2}-\tfrac{\beta_2}{\alpha_2})x_c\bigr) 
\Bigr) =0.
\end{aligned}
\end{multline}

\begin{multline}\label{b3ddd}
B3_{(\delta_1;\,\delta_2;\,\delta_3)}(x;x_a,x_b,x_c,x_d;\al,\bt)= \\ 
    \begin{aligned}
& x_b x_c-x_a x_d + \tfrac{\delta_2}{2}(\tfrac{\alpha_2}{\alpha_1}-\tfrac{\alpha_1}{\alpha_2})(\tfrac{\beta_1}{\beta_2}-\tfrac{\beta_2}{\beta_1}) + 
{\delta_2}(\tfrac{\alpha_1}{\beta_2} x_a - \tfrac{\alpha_1}{\beta_1} x_b - \tfrac{\alpha_2}{\beta_2} x_c + \tfrac{\alpha_2}{\beta_1} x_d)x \\
& +\frac{\delta_1}{x}(\tfrac{\beta_2}{\alpha_1} x_a - \tfrac{\beta_1}{\alpha_1} x_b - \tfrac{\beta_2}{\alpha_2} x_c + \tfrac{\beta_1}{\alpha_2} x_d)+ 
\frac{\delta_3}{x}\bigl(x_a x_b \alpha_2(\tfrac{x_d}{\beta_2}-\tfrac{x_c}{\beta_1})+x_c x_d \alpha_1(\tfrac{x_a}{\beta_1}-\tfrac{x_b}{\beta_2})\bigr)=0.
    \end{aligned}
\end{multline}

\begin{multline}\label{c3ddd}
C3_{(\delta_1;\,\delta_2;\,\delta_3)}(x;x_a,x_b,x_c,x_d;\al,\bt)= \\ 
    \begin{aligned}
& \Bigl(\alpha_2^2 (x_b x_c - x_a x_d) + \beta_1 \beta_2 (x_a x_c - x_b x_d) + \alpha_2(\tfrac{\beta_2}{\beta_1}-\tfrac{\beta_1}{\beta_2})\bigl({\delta_1} \alpha_1 - {\delta_3}\tfrac{\beta_1 \beta_2}{\alpha_1}x_a x_b  \\
&\phantom{\Bigl(}+  {\delta_2}\tfrac{\beta_1 \beta_2}{\alpha_1} x_c x_d\bigr)\Bigr)x %
 +\alpha_2 x_a x_b (\beta_2 x_d-\beta_1 x_c) + {\delta_1}\alpha_1\bigl(\beta_1 x_b-\beta_2 x_a + \alpha_2^2 (\tfrac{x_a}{\beta_2} - \tfrac{x_b}{\beta_1})\bigr)  \\
& + {\delta_2}\Bigl(\tfrac{(\alpha_2^2 - \beta_1^2)(\alpha_2^2 - \beta_2^2)}{2\alpha_2\beta_1\beta_2}(\beta_2 x_d-\beta_1 x_c) + \tfrac{x_c x_d}{\alpha_1}\bigl(\beta_1 \beta_2 (\beta_1 x_b-\beta_2 x_a) + \alpha_2^2 (\beta_1 x_a - \beta_2 x_b)\bigr)\Bigr) \\
& +\Bigl(\alpha_2 (\beta_1 x_d - \beta_2 x_c) -{\delta_3}\bigl(\alpha_2^2 (\beta_1 x_b-\beta_2 x_a) + \beta_1 \beta_2(\beta_1 x_a - \beta_2 x_b)\bigr)\alpha_1^{-1}\Bigr)x^2=0.
    \end{aligned}
\end{multline}

The equation \eqref{a3d} satisfies CAFCC for the two values $\delta=0,1$.  The equations \eqref{a3d}, \eqref{b3ddd}, \eqref{c3ddd}, collectively satisfy CAFCC for the four values of $(\delta_1,\delta_2,\delta_3)=(0,0,0)$, $(1,0,0)$, $(\tfrac{1}{2},0,\tfrac{1}{2}),(\tfrac{1}{2},\tfrac{1}{2},0)$, where the parameter for \eqref{a3d} is $\delta=2\delta_2$.

\subsection*{\texorpdfstring{$A2_{(\delta_1;\,\delta_2)}$}{A2(delta1;delta2)}, \texorpdfstring{$B2_{(\delta_1;\,\delta_2;\,\delta_3)}$}{B2(delta1;delta2;delta3)}, \texorpdfstring{$C2_{(\delta_1;\,\delta_2;\,\delta_3)}$}{C2(delta1;delta2;delta3)}}\label{app:b2}

Define $\theta_{ij}$ by
\begin{align}%
    \theta_{ij}=\theta_i-\theta_j, \qquad i,j\in\{1,2,3,4\},\qquad (\theta_1,\theta_2,\theta_3,\theta_4)=(\alpha_1,\alpha_2,\beta_1,\beta_2).
\end{align}

\begin{multline}\label{a2dd}
 A2_{(\delta_1;\,\delta_2)}(x;x_a,x_b,x_c,x_d;\al,\bt)= \\ 
\begin{aligned}
& \theta_{23}(x_ax^2 -x_bx_cx_d) - \theta_{24}(x_bx^2 -x_ax_cx_d) - \theta_{13}(x_cx^2 -x_ax_bx_d) + \theta_{14}(x_dx^2 -x_ax_bx_c)  \\
& + \Bigl(\theta_{34}(x_ax_b-x_cx_d) + \theta_{12}(x_ax_c-x_bx_d) - (\theta_{13}+\theta_{24})(x_ax_d-x_bx_c)  \Bigr)x \\
& + {\delta_1}\Bigl(\theta_{14}\theta_{23}(\theta_{13}x_b+\theta_{24}x_c)(2x-\theta_{12}\theta_{34})^{\delta_2} - \theta_{13}\theta_{24}(\theta_{13}x_a+\theta_{23}x_d)(2x+\theta_{12}\theta_{34})^{\delta_2}
   \Bigr) \\
& + {\delta_1}x\theta_{12}\theta_{34}(\theta_{13}+\theta_{24})\bigl(x+x_a+x_b+x_c+x_d - \theta_{12}^2-\theta_{13}\theta_{23}-\theta_{14}\theta_{24}\bigr)^{\delta_2}  \\
& + {\delta_2}\Bigl( x_a \theta_{13}\theta_{14}(\theta_{24}\theta_{14}^2-\theta_{34}x_b) - x_b\theta_{13}\theta_{23}(\theta_{14}\theta_{13}^2-\theta_{12}x_d) - x_c\theta_{14}\theta_{24}(\theta_{23}\theta_{24}^2+\theta_{12}x_a)     \\
& \;+ x_d\theta_{23}\theta_{24}(\theta_{13}\theta_{23}^2 +\theta_{34}x_c) +
     \bigl(x_a x_d\theta_{13}\theta_{42} + x_b x_c\theta_{23}\theta_{14} + {\textstyle \prod\limits_{1\leq i<j\leq4}}\theta_{ij}\bigr)(\theta_{13}+\theta_{24})\!\Bigr)\! =0.
\end{aligned}
\end{multline}

\begin{multline}\label{b2ddd}
B2_{(\delta_1;\,\delta_2;\,\delta_3)}(x;x_a,x_b,x_c,x_d;\al,\bt)={\delta_1}\Bigl( \theta_{12}\theta_{43}(\theta_{12}^2-\theta_{13}\theta_{14}-\theta_{23}\theta_{24})^{\delta_2}(-x_a-x_b-x_c-x_d)^{\delta_3} \\ 
    \begin{aligned}
& +\bigl(x_a(x+\theta_{14}\theta_{41}^{\delta_3})^{1+{\delta_2}}-x_b(x+\theta_{13}\theta_{31}^{\delta_3})^{1+{\delta_2}}-x_c(x+\theta_{24}\theta_{42}^{\delta_3})^{1+{\delta_2}}+x_d(x+\theta_{23}\theta_{32}^{\delta_3})^{1+{\delta_2}}\bigr)\!\Bigr)   \\
& + {\delta_3}\bigl((x_a x_c-x_b x_d)\theta_{12} + (x_a x_b-x_c x_d)\theta_{34} - x_bx_c(x_a+x_d) + x_ax_d(x_b+x_c)\bigr) \\
& + 2{\delta_2}( \theta_{12}\theta_{34})x^2 + (2{\delta_2} x +{\delta_3})\theta_{12}\theta_{34}(\theta_{13}+\theta_{24})+ (x_a x_d - x_b x_c) (\theta_{31}+\theta_{42})^{\delta_3}(-1)^{\delta_2} =0.
    \end{aligned}
\end{multline}

\begin{multline}\label{c2ddd}
C2_{(\delta_1;\,\delta_2;\,\delta_3)}(x;x_a,x_b,x_c,x_d;\al,\bt)= \\ 
    \begin{aligned}
& (x_d -x_c)(x^2 + x_a x_b) + \theta_{34} (x^2 - x_a x_b)(\theta_{13}+\theta_{14})^{\delta_3} +2 {\delta_3}(\theta_{23} x_a - \theta_{24} x_b)x^2 \\
& + \Bigl((x_a + x_b + 2 {\delta_2} \theta_{23}\theta_{24})(x_c - x_d) - (x_a-x_b)(\theta_{23}+\theta_{24})(\theta_{13}+\theta_{14})^{\delta_3} + 2 {\delta_3} \theta_{34} x_a x_b \Bigr)x \\
& + {\delta_1}\bigl(x_a \theta_{24} -x_b \theta_{23} + \theta_{34}(\delta_2\theta_{23}\theta_{24}-x)\bigr)
\Bigl(\theta_{13}^{1+{\delta_2}+{\delta_3}}+\theta_{14}^{1+{\delta_2}+{\delta_3}} + 2 {\delta_2} x_c x_d -(x_c+x_d)(\theta_{13}+\theta_{14})^{\delta_2}\Bigr) \\
& +\delta_1 \theta_{23}\theta_{24}\bigl(x_c-x_d+\theta_{34}(\theta_{13}+\theta_{14}- 2 x)^{{\delta_3}}\bigr)(x_a+x_b-\theta_{34}^2-\theta_{23}\theta_{24})^{\delta_2} =0.
    \end{aligned}
\end{multline}

The equation \eqref{a2dd} satisfies CAFCC  for the three values $(\delta_1,\delta_2)=(0,0),(1,0), (1,1)$. The equations \eqref{a2dd}, \eqref{b2ddd}, \eqref{c2ddd}, collectively satisfy CAFCC  with the four values $(\delta_1,\delta_2,\delta_3)=(0,0,0),(1,0,0),(1,0,1),(1,1,0)$.

\subsection*{\texorpdfstring{$A2_{(\delta_1;\,\delta_2)}$, $D1$, $C1_{(\delta_1)}$}{A1(delta1;delta2), D1, C1(delta1)}}\label{app:b1}

\begin{align}\label{d1}
 D1(x_a,x_b,x_c,x_d)= x_a-x_b-x_c+x_d=0.
\end{align}
 \begin{align}\label{c1d}
\begin{split}
& C1(x;x_a,x_b,x_c,x_d;\al,\bt)=\,
 (x_c-x_d)x^2 +
 \Bigl(2(\beta_1-\beta_2)\bigl(-\frac{x_c+x_d}{2}\bigr)^{\delta_1} -(x_a+x_b)(x_c-x_d)\Bigr)x \\
& + 2 \bigl((\beta_2-\alpha_2)x_a +(\alpha_2-\beta_1)x_b\bigr)\bigl(-\frac{x_c+x_d}{2}\bigr)^{\delta_1} + \bigl(x_ax_b- \delta_1 (\alpha_2-\beta_1)(\alpha_2-\beta_2)\bigr)(x_c-x_d) =0.
\end{split}
\end{align}
The equations \eqref{a2dd}, \eqref{d1}, \eqref{c1d}, collectively satisfy CAFCC for the values $(\delta_1,\delta_2)=(0,0),(1,0)$.

\subsection*{Four-leg expressions}

The above equations have equivalent four-leg type expressions that are given respectively in \eqref{4leg}--\eqref{4legc}, with the functions given in Tables \ref{table-A} and \ref{table-BC2} below.  The abbreviation {\it add.}, indicates an additive form of one of the equations \eqref{4leg}, \eqref{4legb}, \eqref{4legc}, given respectively by
\begin{align}
    a(x;x_a;\alpha_2,\beta_1)+a(x;x_d;\alpha_1,\beta_2)-a(x;x_b;\alpha_2,\beta_2)-a(x;x_c;\alpha_1,\beta_1)=0, \\
    b(x;x_a;\alpha_2,\beta_1)+b(x;x_d;\alpha_1,\beta_2)-b(x;x_b;\alpha_2,\beta_2)-b(x;x_c;\alpha_1,\beta_1)=0, \\
    a(x;x_a;\alpha_2,\beta_1)+c(x;x_d;\alpha_1,\beta_2)-a(x;x_b;\alpha_2,\beta_2)-c(x;x_c;\alpha_1,\beta_1)=0.
\end{align}

\begin{table}[htb!]
\centering
\begin{tabular}{c|c}%

 Type-A & $a(x;y;\alpha,\beta)$ %
 
 \\
 
 \hline
 
 \\[-0.4cm]

$A3_{(\delta=1)}$ & $\displaystyle
\frac{\alpha^2+\beta^2\overline{x}^2-2\alpha\beta \overline{x}y}
     {\beta^2+\alpha^2\overline{x}^2-2\alpha\beta \overline{x}y}$

\\[0.4cm]

$A3_{(\delta=0)}$ & $\displaystyle
\frac{\beta x-\alpha y}{\alpha x-\beta y}$

\\[0.4cm]

$A2_{(\delta_1=1;\,\delta_2=1)}$ & $\displaystyle
\frac{(\sqrt{x}+\alpha-\beta)^2-y}{(\sqrt{x}-\alpha+\beta)^2-y}$

\\[0.4cm]

$A2_{(\delta_1=1;\,\delta_2=0)}$ & $\displaystyle
\frac{-x+y+\alpha-\beta}{x-y+\alpha-\beta}$

\\[0.4cm]

$A2_{(\delta_1=0;\,\delta_2=0)}$ & $\displaystyle
\phantom{(add.) }\quad  \frac{\alpha-\beta}{x-y} \quad (add.)$

\\[0.25cm]

\hline 
\end{tabular}
\caption{A list of $a(x;y;\alpha,\beta)$ in \eqref{4leg} for type-A equations \eqref{a3d}, \eqref{a2dd}.  Here $\overline{x}=x+\sqrt{x^2-1}$.}%
\label{table-A}
\end{table}


\begin{table}[htb!]
\centering
\begin{tabular}{c|c}

Type-B  & $b(x;y;\alpha,\beta)$ 
 
 \\
 
 \hline
 
 \\[-0.4cm]

$B3_{(\delta_1=\frac{1}{2};\,\delta_2=\frac{1}{2};\,\delta_3=0)}$ &  $\displaystyle
\beta^2+\alpha^2x^2-2\alpha\beta xy$

\\[0.2cm]

$B3_{(\delta_1=\frac{1}{2};\,\delta_2=0;\,\delta_3=\frac{1}{2})}$  & $\displaystyle
\frac{\alpha y-\beta\overline{x}}{\alpha\overline{x}y-\beta}$

\\[0.2cm]

$B3_{(\delta_1=1;\,\delta_2=0;\,\delta_3=0)}$ &  $\displaystyle
\beta-\alpha xy$

\\[0.2cm]

 $B3_{(\delta_1=0;\,\delta_2=0;\,\delta_3=0)}$  & $\displaystyle y$

\\[0.2cm]

$B2_{(\delta_1=1;\,\delta_2=1;\,\delta_3=0)}$ & $\displaystyle
(x+\alpha-\beta)^2-y$

\\[0.2cm]

$B2_{(\delta_1=1;\,\delta_2=0;\,\delta_3=1)}$  & $\displaystyle
\frac{\sqrt{x}+y+\alpha-\beta}{-\sqrt{x}+y+\alpha-\beta}$

\\[0.2cm]

$B2_{(\delta_1=1;\,\delta_2=0;\,\delta_3=0)}$ & $\displaystyle
x+y+\alpha-\beta$

\\[0.2cm]

$B2_{(\delta_1=0;\,\delta_2=0;\,\delta_3=0)}$ & $\displaystyle y$

\\[0.2cm]

$D1$  & $\displaystyle
\phantom{(add.) }\quad  y \quad (add.)$

\\[0.2cm]

$D1$  & $\displaystyle
\phantom{(add.) }\quad  y \quad (add.)$

\\[0.05cm]

\hline 
\end{tabular}
\hspace{0.0cm}
\begin{tabular}{c|c}%

 Type-C & $c(x;y;\alpha,\beta)$ %
 
 \\
 
 \hline
 
 \\[-0.4cm]

 $C3_{(\delta_1=\frac{1}{2};\,\delta_2=\frac{1}{2};\,\delta_3=0)}$  
&
$\displaystyle
\frac{\alpha-\beta\overline{x}y}{\alpha\overline{x}-\beta y}$

\\[0.2cm]

$C3_{(\delta_1=\frac{1}{2};\,\delta_2=0;\,\delta_3=\frac{1}{2})}$
&
$\bigl(\tfrac{\alpha^2}{\beta}+\beta x^2-2\alpha x y\bigr)$

\\[0.2cm]

$C3_{(\delta_1=1;\,\delta_2=0;\,\delta_3=0)}$      
&
$xy-\tfrac{\alpha}{\beta}$

\\[0.2cm]

$C3_{(\delta_1=0;\,\delta_2=0;\,\delta_3=0)}$ 
&
$y$

\\[0.2cm]

$C2_{(\delta_1=1;\,\delta_2=1;\,\delta_3=0)}$
&
$\displaystyle
\frac{-\sqrt{x}+y-\alpha+\beta}{\sqrt{x}+y-\alpha+\beta}$

\\[0.2cm]

$C2_{(\delta_1=1;\,\delta_2=0;\,\delta_3=1)}$ 
&
$(x-\alpha+\beta)^2-y$

\\[0.2cm]

$C2_{(\delta_1=1;\,\delta_2=0;\,\delta_3=0)}$ 
&
$\displaystyle
x+y-\alpha+\beta$

\\[0.2cm]

$C2_{(\delta_1=0;\,\delta_2=0;\,\delta_3=0)}$ 
&
$\phantom{(a) }$ \,  $-\frac{y+\beta}{2x}  \quad (add.)$

\\[0.2cm]

$C1_{(\delta_1=1)}$ 
& 
$y$

\\[0.2cm]

$C1_{(\delta_1=0)}$ 
& 
$\phantom{(add.) }\quad  y \quad (add.)$

\\[0.05cm]

\hline 
\end{tabular}
\caption{Left: A list of $b(x;y;\alpha,\beta)$ in \eqref{4legb} for type-B equations \eqref{b3ddd}, \eqref{b2ddd}. Right: A list of $c(x;y;\alpha,\beta)$ in \eqref{4legc} for type-C equations \eqref{c3ddd}, \eqref{c2ddd}.  For $C3_{(\delta_1;\,\delta_2;\,\delta_3)}$, $C2_{(\delta_1;\,\delta_2;\,\delta_3)}$, and  $C1_{(\delta_1)}$, the $a(x;y;\alpha,\beta)$ are respectively given by $A3_{(2\delta_2)}$, $A2_{(\delta_1;\,\delta_2)}$, and $A2_{(\delta_1;\,\delta_2=0)}$.  Here $\overline{x}=x+\sqrt{x^2-1}$.}%
\label{table-BC2}
\end{table}

\section{Lists of degree growths for type-C equations}\label{app:degrees}

This Appendix gives the lists of degree growths that have been obtained for type-C equations with the two types of initial conditions \eqref{initialconditionsdef}.%

\subsection*{\texorpdfstring{$C1_{(0)}(\x;\y)$}{C1(0)} system}

\paragraph{$\PTN$.}
The two different growth patterns are
\begin{gather}
    1, 3, 7, 11, 17, 25, 33, 43, 55, 67, 81, 97, 113\dots
    \label{eq:C1x} \\
    1, 3, 5, 9, 15, 21, 29, 39, 49, 61, 75, 89\dots
    \label{eq:C1y}
\end{gather}

\paragraph{$\PNpPNm$.}
The two different growth patterns are
\begin{gather}
    \begin{gathered}
        1, 4, 12, 20, 30, 46, 62, 78, 101, 125, 149, 178, 
        210, 242, 278, 317, 357, 401,
        \\
        446, 494, 546, 598, 653, 713, 773, 834, 902, 970, 
        1038, 1113, 1189\dots
    \end{gathered}
    \label{eq:C1xkels} \\
    \begin{gathered}
        1, 5, 10, 17, 29, 41, 54, 74, 94, 114, 141, 169, 197,
        230, 266, 302, 342, 
        \\
        385, 429, 477, 526, 578, 634, 690, 749, 813, 877, 942, 1014, 1086\dots
    \end{gathered}
    \label{eq:C1ykels}
\end{gather}

\paragraph{Maximal.}

The maximal growth in $\PTN$ is given by \eqref{eq:C1x}.  The maximal growth in $\PNpPNm$ is:
\begin{equation}
    \begin{gathered}
        1, 5, 12, 20, 30, 46, 62, 78, 101, 125, 149, 178, 210, 242, 278, 
        317, 357, 
        \\
        401, 446, 494, 546, 598, 653, 713, 773, 834, 902, 970, 1038, 1113\dots
    \end{gathered}
    \label{eq:C1maxkels}
\end{equation}

\subsection*{\texorpdfstring{$C2_{\left( 1;\,0;\,0 \right)}(\x;\y)$}{C2(1;0;0)} system}

\paragraph{$\PTN$.}
The two different growth patterns are
\begin{gather}
    1, 3, 7, 12, 19, 27, 36, 47, 59, 72, 87, 103, 120\dots
    \label{eq:C2100x} \\
    1, 3, 6, 11, 17, 24, 33, 43, 54, 67, 81, 96\dots
    \label{eq:C2100y}
\end{gather}

\paragraph{$\PNpPNm$.}
The two different growth patterns are
\begin{gather}
    \begin{gathered}
        1, 4, 12, 22, 34, 49, 67, 87, 110, 135, 163, 193, 225, 260, 
        \\
        298, 338, 380, 425, 473, 523, 575, 630, 688, 748, 810\dots
    \end{gathered}
    \label{eq:C2100xkels} \\
    \begin{gathered}
        1, 5, 11, 20, 32, 46, 63, 82, 104, 128, 154, 183, 215, 249, 
        \\
        285, 324, 366, 410, 456, 505, 557, 611, 667, 726\dots
    \end{gathered}
    \label{eq:C2100ykels}
\end{gather}

\paragraph{Maximal.}

The maximal growth in $\PTN$ is given by \eqref{eq:C2100x}.  The maximal growth in $\PNpPNm$  is:
\begin{equation}
    \begin{gathered}
        1, 5, 12, 22, 34, 49, 67, 87, 110, 135, 163, 193, 225, 260, 
        \\
        298, 338, 380, 425, 473, 523, 575, 630, 688, 748\dots
    \end{gathered}
    \label{eq:C2100maxkels}
\end{equation}

\subsection*{\texorpdfstring{$C3_{\left( 0;\,0;\,0 \right)}(\x;\y)$}{C3(0;0;0)} system}

\paragraph{$\PTN$.}
The two different growth patterns are
\begin{gather}
    1, 3, 7, 12, 18, 26, 35, 45, 57, 70, 84, 100, 117\dots
    \label{eq:C3000x} \\
    1, 3, 6, 10, 16, 23, 31, 41, 52, 64, 78, 93\dots
    \label{eq:C3000y}
\end{gather}

\paragraph{$\PNpPNm$.}
The two different growth patterns are
\begin{gather}
    \begin{gathered}
        1, 4, 12, 21, 31, 47, 64, 82, 104, 129, 155, 185, 216, 250, 288, 327,367,
        \\
        413, 460, 508, 560, 615, 671, 731, 792, 856, 924, 993, 1063, 1139, 1216\dots
    \end{gathered}
    \label{eq:C3000xkels} \\
    \begin{gathered}
        1, 5, 10, 18, 30, 43, 57, 77, 98, 120, 146, 175, 205, 239, 274, 312, 354, 
        \\
        397, 441, 491, 542, 594, 650, 709, 769, 833, 898, 966, 1038, 1111\dots
    \end{gathered}
    \label{eq:C3000ykels}
\end{gather}

\paragraph{Maximal.}

The maximal growth in $\PTN$ is given by \eqref{eq:C3000x}.  The maximal growth in $\PNpPNm$   is:
\begin{equation}
    \begin{gathered}
        1, 5, 12, 21, 31, 47, 64, 82, 104, 129, 155, 185, 216, 250, 288, 327, 367,
        \\
        413, 460, 508, 560, 615, 671, 731, 792, 856, 924, 993, 1063, 1139\dots
    \end{gathered}
    \label{eq:C3000maxkels}
\end{equation}

\subsection*{\texorpdfstring{$C3_{\left( 1;\,0;\,0 \right)}(\x;\y)$}{C3(1;0;0)} system}

\paragraph{$\PTN$.}
The two different growth patterns are
\begin{gather}
    1, 3, 7, 13, 20, 28, 38, 49, 61, 75, 90, 106, 124\dots
    \label{eq:C3100x} \\
    1, 3, 7, 12, 18, 26, 35, 45, 57, 70, 84, 100\dots
    \label{eq:C3100y}
\end{gather}

\paragraph{$\PNpPNm$.}
The two different growth patterns are
\begin{gather}
    \begin{gathered}
        1, 4, 12, 22, 34, 50, 68, 88, 111, 137, 165, 195, 228, 264, 
        \\
        302, 342, 385, 431, 479, 529, 582, 638, 696, 756, 819\dots
    \end{gathered}
    \label{eq:C3100xkels} \\
    \begin{gathered}
        1, 5, 11, 21, 33, 47, 64, 84, 106, 130, 157, 187, 219, 253, 
        \\
        290,330, 372, 416, 463, 513, 565, 619, 676, 736\dots
    \end{gathered}
     \label{eq:C3100ykels} 
\end{gather}

\paragraph{Maximal.}

The maximal growth in $\PTN$ is given by \eqref{eq:C3100x}.  The maximal growth in $\PNpPNm$   is:
\begin{equation}
    \begin{gathered}
        1, 5, 12, 22, 34, 50, 68, 88, 111, 137, 165, 195, 228, 264, 
        \\
        302, 342,385, 431, 479, 529, 582, 638, 696, 756\dots
    \end{gathered}
    \label{eq:C3100maxkels}
\end{equation}

\subsection*{\texorpdfstring{$C1_{\left( 1 \right)}(\x;\y)+C2_{\left( 0;\,0;\,0 \right)}(\y;\x)$}{C1(1)+C2(0;0;0)} system}

\paragraph{$\PTN$ -- NE \& SW directions.}
The four different growth patterns are:%
\begin{gather}
    1, 3, 7, 11, 17, 25, 33, 43, 55, 67, 81, 97, 113\dots
    \label{eq:C11C2000NE1} \\
    1, 3, 5, 9, 15, 21, 29, 39, 49, 61, 75, 89\dots
    \label{eq:C11C2000NE2} \\
    1, 3, 7, 12, 18, 26, 35, 45, 57, 70, 84, 100\dots
    \label{eq:C11C2000NE3} \\
    1, 3, 6, 10, 16, 23, 31, 41, 52, 64, 78\dots
    \label{eq:C11C2000NE4}
\end{gather}

\paragraph{$\PTN$ -- NW \& SE directions.}
The four different growth patterns are:%
\begin{gather}
        1, 3, 7, 11, 18, 25, 35, 43, 57, 67, 84, 97, 117, 131, 155, 171, 198, 217, 247, 267, 301\dots
    \label{eq:C11C2000NW1} \\
        1, 3, 5, 10, 15, 23, 29, 41, 49, 64, 75, 93, 105, 127, 141, 166, 183, 211, 229, 261\dots
    \label{eq:C11C2000NW2} \\
        1, 3, 7, 12, 17, 26, 33, 45, 55, 70, 81, 100, 113, 135, 151, 176, 193, 222, 241, 273\dots
    \label{eq:C11C2000NW3}\\
    1, 3, 6, 9, 16, 21, 31, 39, 52, 61, 78, 89, 109, 123, 146, 161, 188, 205, 235\dots
    \label{eq:C11C2000NW4}
\end{gather}

\paragraph{$\PNpPNm$ -- NE \& SW directions.}
The four different growth patterns are:%
\begin{gather}
    \begin{gathered}
        1, 4, 12, 20, 30, 46, 62, 79, 101, 125, 150, 179, 210, 
        243, 280, 317, 358, 403, 448, 495, 548,
        \\
         601, 655, 715, 776, 838, 904, 973, 1043, 1117, 1192, 
        1270, 1352, 1434, 1519, 1609, 1699\dots
    \end{gathered}
    \label{eq:C11C2000NEkels1} \\
    \begin{gathered}
        1, 5, 10, 17, 29, 41, 54, 74, 94, 115, 141, 169, 
        198, 231, 266, 303, 344, 385, 430, 479,528,
        \\
         579, 
        636, 693, 751, 815, 880, 946, 1016, 1089, 1163, 1241, 1320, 1402, 1488, 1574\dots
    \end{gathered}
    \label{eq:C11C2000NEkels2} \\
    \begin{gathered}
        1, 4, 12, 21, 31, 47, 64, 81, 104, 129, 154, 183, 216, 249, 286, 325, 366, 411,457, 505, 
        \\
        558, 612, 666, 727, 789, 851, 918, 988, 1058, 1132, 1209, 1287, 1369, 1452, 1538, 1628\dots
    \end{gathered}
    \label{eq:C11C2000NEkels3} \\
    \begin{gathered}
        1, 5, 10, 18, 30, 43, 57, 77, 98, 119, 146, 175, 204, 
        237, 274, 311, 352, 395, 440, 489,
        \\
        539, 591, 648, 706, 764, 829, 895, 961, 1032, 1106, 1180, 1258, 1339, 1421, 1507\dots
    \end{gathered}
    \label{eq:C11C2000NEkels4}
\end{gather}

\paragraph{$\PNpPNm$ -- NW \& SE directions.}
The four different growth patterns are:%
\begin{gather}
    \begin{gathered}
        1, 4, 12, 20, 31, 46, 64, 79, 104, 125, 154, 179, 216, 
        243, 286, 317, 366, 403, 457, 495, 
        \\
        558,601, 666, 715, 789, 838, 918, 973, 1058, 1117, 1209, 
        1270, 1369, 1434, 1538, 1609,
        \\
        1719, 1790,1909, 1986, 2107, 2189, 2319, 2402, 2537,
        2626, 2767, 2860, 3007, 3102, 3257,
        \\
        3356,3515, 3620, 3786, 3891, 4065, 4176, 4353, 4468, 
        4654, 4771, 4962, 5084, 5281\dots
    \end{gathered}
    \label{eq:C11C2000NWkels1} \\
    \begin{gathered}
        1, 5, 10, 18, 29, 43, 54, 77, 94, 119, 141, 175, 198, 237, 266, 311, 344, 395,430, 489, 528,
        \\
          591, 636, 706, 751, 829, 880, 961, 1016, 1106, 1163, 1258, 1320, 1421, 1488, 1594,  1663,
        \\
          1778, 1851, 1969, 2048, 2174, 2253, 2386, 2471, 2608, 2697, 2842, 2933, 3084, 3179, 
        \\
           3336, 3437, 3599, 3701, 3872, 3979, 4152, 4265, 4447, 4560, 4748, 4867, 5060\dots
    \end{gathered}
    \label{eq:C11C2000NWkels2} \\
    \begin{gathered}
        1, 4, 12, 21, 30, 47, 62, 81, 101, 129, 150, 183, 210, 249, 280, 325, 358, 411, 448,505, 
        \\
         548,  612, 655, 727, 776, 851, 904, 988, 1043, 1132, 1192, 1287, 1352,1452,1519, 1628, 
          \\
           1699, 1811, 1888, 2008, 2085, 2212, 2295, 2426, 2513, 2652, 2741, 2886, 2979, 3130,
          \\
            3229, 3385, 3485, 3650, 3755, 3922, 4033, 4209, 4320, 4502, 4619, 4806, 4927, 5121\dots
    \end{gathered}
    \label{eq:C11C2000NWkels3} \\
    \begin{gathered}
        1, 5, 10, 17, 30, 41, 57, 74, 98, 115, 146, 169, 204, 231, 274, 303, 352, 385, 440, 479,
        \\
        539, 579, 648, 693, 764, 815, 895, 946, 1032, 1089, 1180, 1241, 1339, 1402, 1507, 1574, 
        \\
        1684, 1757, 1873, 1946, 2071, 2150, 2277, 2361, 2497, 2582, 2723, 2814, 2961, 3056, 
        \\
        3209, 3306, 3467, 3568, 3733, 3840, 4012, 4119, 4299, 4412, 4595, 4712, 4904\dots
    \end{gathered}
    \label{eq:C11C2000NWkels4}
\end{gather}

\paragraph{Maximal.}

The maximal growth with initial conditions in $\PTN$ is
isotropic and given by \eqref{eq:C11C2000NE3}.

The maximal growth with initial conditions in $\PNpPNm$ is
isotropic and given by %
\begin{equation}
    \begin{gathered}
        1, 5, 12, 21, 31, 47, 64, 81, 104, 129, 154, 183, 216, 249, 286,325, 366, 411, 457,
        \\
        505, 558, 612, 666, 727, 789, 851, 918, 988, 1058, 1132, 1209, 1287, 1369, 1452, 1538\dots
    \end{gathered}
    \label{eq:C11C2000NWkelsmax}
\end{equation}

\subsection*{\texorpdfstring{$C2_{\left(1;\,0;\,1\right)}(\x;\y)+C2_{\left(1;\,1;\,0\right)}(\y;\x)$}{C2(1;0;1)+C2(1;1;0)} system}

\paragraph{$\PTN$ -- NE \& SW directions.}
The four different growth patterns are:%
\begin{gather}
        1, 3, 7, 13, 21, 31, 42, 54, 68, 83, 99, 117, 136, 156, 178, 201, 225, 251, 278, 
    306, 336\dots
    \label{eq:C2101C2110NE1} \\
    1, 3, 7, 13, 21, 30, 40, 52, 65, 79, 95, 112, 130, 150, 171, 193, 217, 
    242, 268, 296, 325, 355\dots
    \label{eq:C2101C2110NE2} \\
    1, 3, 7, 13, 20, 28, 38, 49, 61, 75, 90, 106, 124, 143, 163, 
    185, 208, 232, 258, 285, 313, 343\dots
    \label{eq:C2101C2110NE3} \\
    1, 3, 7, 12, 18, 26, 35, 45, 57, 70, 84, 100, 117, 135, 155, 
    176, 198, 222, 247, 273, 301, 330\dots
    \label{eq:C2101C2110NE4}
\end{gather}

\paragraph{$\PTN$ -- NW \& SE directions.}
The four different growth patterns are:%
\begin{gather}
    \begin{gathered}
        1, 3, 7, 13, 20, 31, 38, 54, 61, 83, 90, 117, 124, 156, 163, 
        \\
        201, 208, 251, 258, 306, 313, 367, 374, 433, 440\dots
    \end{gathered}
    \label{eq:C2101C2110NW1} \\
     \begin{gathered}
        1, 3, 7, 12, 21, 26, 40, 45, 65, 70, 95, 100, 130, 135, 171, 
        \\
        176, 217, 222, 268, 273, 325, 330, 387, 392\dots
    \end{gathered}
    \label{eq:C2101C2110NW2} \\
     \begin{gathered}
        1, 3, 7, 13, 21, 28, 42, 49, 68, 75, 99, 106, 136, 143, 
        \\
        178, 185,225, 232, 278, 285, 336, 343, 399, 406\dots
    \end{gathered}
    \label{eq:C2101C2110NW3}\\
    \begin{gathered}
        1, 3, 7, 13, 18, 30, 35, 52, 57, 79, 84, 112, 117, 150, 
        \\
        155, 193, 198, 242, 247, 296, 301, 355, 360\dots
    \end{gathered}
    \label{eq:C2101C2110NW4}
\end{gather}

\paragraph{$\PNpPNm$ -- NE \& SW directions.}
The four different growth patterns are:%
\begin{gather}
    \begin{gathered}
        1, 5, 12, 23, 38, 57, 78, 101, 127, 155, 186, 221, 257, 296, 338, 382,
        \\
         428, 477, 529, 583, 640, 700, 762, 827, 894, 963, 1036, 1111,1188,
        \\
         1269, 1352, 1437, 1525, 1615, 1708, 1804, 1902, 2003, 2107, 2213, 2321\dots
    \end{gathered}
    \label{eq:C2101C2110NEkels1} \\
    \begin{gathered}
        1, 5, 12, 23, 38, 55, 74, 97, 122, 149, 180, 212, 247, 285, 325, 368,  
        \\
        414,462, 512, 565,621, 679, 740, 804, 870, 939, 1010, 1083, 1160,
        \\
         1239, 1320, 
        1405, 1492, 1581, 1673, 1767, 1864, 1964, 2066, 2171\dots
    \end{gathered}
    \label{eq:C2101C2110NEkels2} \\
    \begin{gathered}
        1, 4, 12, 23, 36, 51, 70, 91, 114, 141, 170, 201, 236, 272, 
        311, 353, 397, 443, 493, 545,
        \\
          600, 657, 717, 779, 844, 911, 981, 1054, 1130, 1207,
        1288, 1371, 1456, 1544, 1635, 1728\dots
    \end{gathered}
    \label{eq:C2101C2110NEkels3} \\
    \begin{gathered}
        1, 5, 12, 21, 32, 47, 65, 85, 108, 133, 160, 191, 
        223, 258, 296, 337, 380, 426, 474, 525, 
        \\
        578, 634, 692, 754, 818, 884, 953, 1025, 1098, 1175, 
        1254, 1336, 1421, 1508, 1597\dots
    \end{gathered}
    \label{eq:C2101C2110NEkels4}
\end{gather}

\paragraph{$\PNpPNm$ -- NW \& SE directions.}
The four different growth patterns are:%
\begin{gather}
    \begin{gathered}
        1, 5, 12, 23, 36, 57, 70, 101, 114, 155, 170, 221, 236, 296, 311, 382, 397,477, 493,583,
        \\
         600, 700, 717, 827, 844, 963, 981, 1111, 1130, 1269, 1288, 1437, 1456, 1615,1635,
        \\
         1804, 1825, 2003, 2024, 2213, 2234, 2432, 2454, 2662, 2685, 2903, 2926, 3154, 3177\dots
    \end{gathered}
    \label{eq:C2101C2110NWkels1} \\
    \begin{gathered}
        1, 5, 12, 21, 38, 47, 74, 85, 122, 133, 180, 191, 247, 258, 325, 337, 414, 426, 512, 
        \\
        525, 621, 634, 740, 754, 870, 884, 1010, 1025, 1160, 1175, 1320, 1336, 1492, 1508, 
        \\
        1673, 1690, 1864, 1881, 2066, 2084, 2279, 2297, 2501, 2520, 2734, 2753, 2977, 2997\dots
    \end{gathered}
    \label{eq:C2101C2110NWkels2}\\
    \begin{gathered}
        1, 4, 12, 23, 38, 51, 78, 91, 127, 141, 186, 201, 257, 272, 338, 353, 428, 443, 529,
        \\
        545, 640, 657, 762, 779, 894, 911, 1036, 1054, 1188, 1207, 1352, 1371, 1525, 1544, 
        \\
        1708, 1728, 1902, 1923, 2107, 2128, 2321, 2342, 2546, 2568, 2781, 2804, 3027, 3050\dots
    \end{gathered}
    \label{eq:C2101C2110NWkels3} \\
    \begin{gathered}
        1, 5, 12, 23, 32, 55, 65, 97, 108, 149, 160, 212, 223, 285, 296, 368, 380, 462,474,
        \\
         565, 578, 679, 692, 804, 818, 939, 953, 1083, 1098, 1239, 1254, 1405, 1421,1581,
        \\
         1597, 1767, 1784, 1964, 1981, 2171, 2189, 2389, 2407, 2616, 2635, 2854, 2873\dots
    \end{gathered}
    \label{eq:C2101C2110NWkels4}
\end{gather}

\paragraph{Maximal.}

The maximal growth with initial conditions in $\PTN$ is
isotropic and given by \eqref{eq:C2101C2110NE1}.

The maximal growth with initial conditions in $\PNpPNm$ is
isotropic and given by \eqref{eq:C2101C2110NEkels1}.

\subsection*{\texorpdfstring{$C3_{\left(\frac{1}{2};\,0;\,\frac{1}{2}\right)}(\x;\y)+C3_{\left(\frac{1}{2};\,\frac{1}{2};\,0\right)}(\y;\x)$}{C3(1/2;0;1/2)+C2(1/2;1/2;0)} system}

\paragraph{$\PTN$ -- NE \& SW directions.}

The same growth patterns as were given in \eqref{eq:C2101C2110NE1}--\eqref{eq:C2101C2110NE4}.

\paragraph{$\PTN$ -- NW \& SE directions.}

The same growth patterns as were given in \eqref{eq:C2101C2110NW1}--\eqref{eq:C2101C2110NW4}.

\paragraph{$\PNpPNm$ -- NE \& SW directions.}
The four different growth patterns are:%
\begin{gather}
    \begin{gathered}
        1, 5, 12, 23, 38, 57, 78, 101, 127, 156, 187, 222, 258, 297, 339,383,
        \\
         429, 479, 531, 586, 643, 703, 765, 830, 897, 967, 1040, 1116, 1193,
        \\
         1274, 1357, 1442, 1530, 1621, 1714, 1811, 1909, 2010, 2114, 2220, 2328\dots
    \end{gathered}
    \label{eq:C3101C3110NEkels1} \\
    \begin{gathered}
        1, 5, 12, 23, 38, 55, 75, 98, 123, 150, 181, 213, 248, 286, 327, 370, 416,
        \\
         464, 515, 568, 624, 682, 744, 808, 874, 943, 1015, 1088, 1165, 1244,
        \\
        1326, 1411, 1498, 1587, 1680, 1774, 1871, 1971, 2074, 2179\dots
    \end{gathered}
    \label{eq:C3101C3110NEkels2} \\
    \begin{gathered}
        1, 4, 12, 23, 36, 52, 71, 92, 115, 142, 171, 202, 237, 
        274, 313, 355,399, 446,496,548,603,
        \\
        661, 721, 783, 848, 916, 986, 1059, 1135, 1213,
        1294, 1377, 1462, 1551, 1642, 1735\dots
    \end{gathered}
    \label{eq:C3101C3110NEkels3} \\
    \begin{gathered}
        1, 5, 12, 21, 33, 48, 66, 86, 109, 134, 161, 192, 225, 260, 299,
        340, 383, 429, 477, 528,  
        \\
        582,638, 697, 759, 823, 889, 958, 1030,
        1104, 1181, 1261, 1343, 1428, 1515, 1604\dots
    \end{gathered}
    \label{eq:C3101C3110NEkels4}
\end{gather}

\paragraph{$\PNpPNm$ -- NW \& SE directions.}
The four different growth patterns are:%
\begin{gather}
    \begin{gathered}
        1, 5, 12, 23, 36, 57, 71, 101, 115, 156, 171, 222, 237, 297, 313, 383, 399,
        \\
         479, 496, 586, 603, 703, 721, 830, 848, 967, 986, 1116, 1135, 1274, 1294,
        \\
          1442, 1462, 1621, 1642, 1811, 1832, 2010, 2032, 2220, 2242, 2440, 2463\dots
    \end{gathered}
    \label{eq:C3101C3110NWkels1} \\
    \begin{gathered}
        1, 5, 12, 21, 38, 48, 75, 86, 123, 134, 181, 192, 248, 260, 327, 340,
        \\
          416, 429, 515, 528, 624, 638, 744, 759, 874, 889, 1015, 1030, 1165, 1181,
        \\
        1326, 1343, 1498, 1515, 1680, 1697, 1871, 1889, 2074, 2093, 2287, 2306\dots
    \end{gathered}
    \label{eq:C3101C3110NWkels2} \\
    \begin{gathered}
        1, 4, 12, 23, 38, 52, 78, 92, 127, 142, 187, 202, 258,  274, 339, 355,
        \\
          429, 446, 531, 548, 643,  661, 765, 783, 897, 916, 1040, 1059, 1193, 1213,
        \\
        1357, 1377, 1530, 1551, 1714, 1735, 1909, 1931, 2114, 2136, 2328, 2351\dots
    \end{gathered}
    \label{eq:C3101C3110NWkels3} \\
    \begin{gathered}
        1, 5, 12, 23, 33, 55, 66, 98, 109, 150, 161, 213, 225, 286, 299,
        370, 
        \\
         383, 464, 477, 568, 582, 682, 697, 808, 823, 943, 958, 1088,
        1104, 1244, 
        \\
        1261, 1411, 1428, 1587, 1604, 1774, 1792, 1971, 1990, 2179, 2198\dots
    \end{gathered}
    \label{eq:C3101C3110NWkels4}
\end{gather}

\paragraph{Maximal.}

The maximal growth with initial conditions in $\PNpPNm$ is
isotropic and given by \eqref{eq:C3101C3110NEkels1}.

\end{appendices}

\bibliographystyle{plain}
\bibliography{bibliography}

\end{document}